\documentclass{article}
\usepackage{amsmath,amssymb,amsthm, amsfonts, mathrsfs}
\usepackage{hyperref}
\usepackage{xcolor}
\usepackage{textcomp}
\usepackage{enumerate}
\usepackage{calligra,indentfirst,epsfig}
\calligra
\hypersetup{
           breaklinks=true,   
           colorlinks=true,   
           pdfusetitle=true,  
        }

\DeclareMathOperator{\tr}{tr}
\DeclareMathOperator{\Tr}{Tr}

\DeclareMathOperator{\rank}{rank}

\def\F{\mathbb{F}}
\def\Z{\mathbb{Z}}

\newcommand\independent{\protect\mathpalette{\protect\independenT}{\perp}}
\def\independenT#1#2{\mathrel{\rlap{$#1#2$}\mkern3mu{#1#2}}}

\newtheorem{thm}{Theorem}[section]
\newtheorem{lem}{Lemma}[section]
\newtheorem{cor}{Corollary}[section]
\newtheorem{prop}{Proposition}[section]

\theoremstyle{remark}
\newtheorem{rem}{Remark}[section]
\newtheorem{ex}{Example}[section]
\theoremstyle{definition}
\newtheorem{defin}{Definition}[section]
\allowdisplaybreaks
\numberwithin{equation}{section}
\usepackage{breakurl}
\usepackage{url}

\begin{document} 
 
\title{\bf {Entanglement-Assisted Quantum Error-Correcting Codes over Local Frobenius Rings}\thanks{This work was presented in part at the 2022 IEEE International Symposium on Information Theory (ISIT 2022), held in Espoo, Finland.}}

\author{Tania Sidana and Navin Kashyap\thanks{The authors are with the Department of Electrical Communication Engineering, Indian Institute of Science, Bangalore 560012. Email: \{taniasidana,nkashyap\}@iisc.ac.in. This work was supported an IISc-IoE postdoctoral fellowship awarded to the first author.}}

\date{}
 
\maketitle 

 \begin{abstract} In this paper, we provide a framework for constructing entanglement-assisted quantum error-correcting codes (EAQECCs) from classical additive codes over a finite commutative local Frobenius ring $\mathcal{R}$. At the heart of the framework, and this is one of the main technical contributions of our paper, is a procedure to construct, for an additive code $\mathcal{C}$ over $\mathcal{R}$, a generating set for $\mathcal{C}$ that is in standard form, meaning that it consists purely of isotropic generators and hyperbolic pairs.  
Moreover, when $\mathcal{R}$ is a Galois ring, we give an exact expression for the minimum number of pairs of maximally entangled qudits required to construct an EAQECC from an additive code over $\mathcal{R}$, which significantly extends known results for EAQECCs over finite fields. We also demonstrate how adding extra coordinates to an additive code can give us a certain degree of flexibility in determining the parameters of the EAQECCs that result from our construction.
\end{abstract}


 \section{Introduction}
Quantum error-correcting codes (QECCs) protect quantum states against decoherence caused by the interaction between quantum states and their environment. The stabilizer framework, proposed by Gottesman \cite{gottesman}, is one of the main mechanisms for constructing QECCs. The construction is based on abelian subgroups of the Pauli group, and the resulting QECCs are called quantum stabilizer codes. The stabilizer framework encompasses the first QECC constructed by Shor \cite{shor}, as well as the construction from classical error-correcting codes, discovered independently by Calderbank and Shor \cite{calderbank} and Steane \cite{steane}. The latter construction, now known as the Calderbank-Shor-Steane (CSS) construction, uses dual-containing (or self-orthogonal) classical codes to form QECCs. 

Originally developed for qubits, the stabilizer framework was subsequently extended to higher-dimensional qudit spaces. The (Pauli) error group in this case is generated by unitary operators whose actions on qudits are defined by the algebra of an underlying finite field or ring. The extension of the stabilizer framework to error groups defined via finite fields was executed by Ashikhmin and Knill \cite{ashikhmin}, and Ketkar et al.\ \cite{ketkar}, while the same was done for finite commutative Frobenius rings by Nadella and Klappenecker \cite{nadella}, and Gluesing-Luerssen and Pllaha \cite{luerssen}.

The stabilizer formalism was extended in a different direction by Brun et al.\ \cite{brun}, who gave a method of constructing QECCs (over qubits) from \emph{non-abelian} subgroups of the Pauli group. The idea here was to add more dimensions to the Pauli group, so as to introduce extra degrees of freedom that can be used to ``abelianize'' the original non-abelian subgroup. The code construction required the existence of a small number of pre-shared entanglement qubits between the sender and receiver, where the receiver-end qubits are assumed to be error-free throughout. As a consequence, these codes were called entanglement-assisted quantum error correcting codes (EAQECCs). Within this framework, \emph{any} classical binary linear code can be used as the starting point for constructing an EAQECC. Wilde and Brun \cite{wilde} determined the minimal number of shared qubits required to construct an EAQECC starting from a given binary linear code, and more generally, starting from a given non-abelian subgroup of the Pauli group. 

The theory of EAQECCs extends readily to qudit spaces for which the (Pauli) error groups are defined by finite fields. Indeed, Wilde and Brun observe in \cite[Remark~1]{wilde} that their formula for the minimum number of shared qudits also applies to EAQECCs constructed from linear codes over any prime field; a formal proof of this was given by Luo et al.\ \cite{luo}. Later, Galindo et al.\ \cite{galindo} verified that this formula also applied to EAQECCs obtained from linear codes over an arbitrary finite field $\F_q$. More recently, Nadkarni and Garani \cite{nadkarni} derived an analogous formula for EAQECCs constructed from $\F_p$-additive codes over $\F_q$, where $p$ is the characteristic of $\F_q$.

In this paper, we extend the EAQECC formalism to qudit spaces on which error actions are defined by finite local commutative Frobenius rings. This enables us to construct EAQECCs starting from classical additive codes over such rings, which are overall a much richer class of (classical) codes. It must be pointed out here that Lee and Klappenecker \cite{lee} have previously attempted to construct EAQECCs from free linear codes over finite commutative (not necessarily local) Frobenius rings. However, their EAQECC construction relies crucially on Theorem~5 of their paper, in the proof of which we found a gap that could not readily be filled.\footnote{In the proof of \cite[Theorem 5]{lee}, Lee and Klappenecker replace $\mathbf{w}_k$ with $\mathbf{w}'_{k-2} = e_{k,i} \mathbf{w}_k - \cdots $. However, this may not result in a basis of $R^{2n},$ as $e_{k,i}$ need not be a unit in the ring $R$.}  By restricting our attention further to local rings, we are able to furnish a proof of the result (Theorem~\ref{thm:stdform}) needed for the construction of EAQECCs. Thus, one of the contributions of this paper is to provide a coding-theoretic framework to construct EAQECCs over finite commutative local Frobenius rings from first principles. Using this framework, we can complete the program initiated by Lee and Klappenecker \cite{lee} of constructing EAQECCs from linear codes over finite commutative (not necessarily local) Frobenius rings --- see the discussion following Remark \ref{rem:lee_klapp} towards the end of Section~\ref{sec3}.

We then consider the problem of determining the minimum number of pre-shared pairs of maximally entangled qudits required to construct an EAQECC within our framework, starting from an additive code over a finite commutative Frobenius ring. We succeed in deriving a formula for this number in the special case when the ring is a Galois ring. To get to the answer, we had to first derive it for the basic case of the integer rings $\Z_{p^a},$ which itself turned out to be a somewhat non-trivial task.

Finally, in the spirit of the propagation rules for EAQECCs proposed by Luo et al.\ \cite[Theorems~16 and 18]{luo1}, we explore how lengthening an additive code by inserting extra coordinates can affect the parameters and error-handling capabilities of the EAQECCs obtained via our construction. We describe two methods for lengthening. Using the first method, the EAQECC obtained from the lengthened code has the same dimension as the EAQECC from the original code while requiring fewer pairs of maximally entangled qudits; but this is usually at the expense of a loss in minimum distance. When we employ the second method of lengthening, the EAQECC obtained from the lengthened code requires the same number of maximally entangled qudit pairs as that obtained from the original code, but the minimum distance can now increase. However, the potential increase in minimum distance is always accompanied by a reduction in dimension. 

The remainder of this paper is organized as follows. In Section \ref{sec2}, we establish the basic definitions and notation needed to describe the construction of quantum stabilizer codes and EAQECCs. This section also contains statements of our main results.  In Section \ref{sec3}, we provide the means to construct EAQECCs from any additive code over a finite commutative local Frobenius ring. In Section \ref{sec4}, we derive an explicit form of the minimum number of pre-shared required to construct an EAQECC over a ring, $\Z_{p^a}$, of integers modulo a prime power. This result is vastly generalized in Section \ref{sec5} to cover EAQECCs over any finite Galois ring. In Section \ref{sec:lengthening}, we present two methods to lengthen an additive code so as to flexibly modify the parameters of the EAQECCs constructed from them. The paper ends in Section \ref{sec6} with some concluding remarks. Some supplementary material of a more technical nature are contained in the appendices.

\section{Preliminaries}\label{sec2}  Let $R$ be a finite commutative ring with unity. Let $\text{Hom}(R,\mathbb{C}^*)$ be the set of all additive characters of $R,$ i.e.,  the set of group homomorphisms from $(R,+)$ to $\mathbb{C}^*.$  A ring $R$ is a Frobenius ring if there exists an additive character $\chi$ such that  $\text{Hom}(R,\mathbb{C}^*)=R\cdot\chi.$ Any additive character with this property is called a generating character of $R.$ Finite fields, the rings $\mathbb{Z}_N$ of integers modulo $N,$ Galois rings and finite chain rings are a few examples of finite Frobenius rings.
%

Throughout this paper, we take $\mathcal{R}$ to be a finite commutative local Frobenius ring with generating character $\chi.$  Further, as $\mathcal{R}$ is a finite commutative local ring, the characteristic, $\text{char}{\mathcal R}$, of $\mathcal{R}$  is a power of a prime number, and the cardinality, $|\mathcal{R}|$, of $\mathcal{R}$ is also a power of prime number. Let  $|\mathcal{R}|=p^a=q,$ and $\text{char}{\mathcal R}=p^b,$ where $p$ is a prime.   Furthermore, with $\zeta = \exp\bigl(\frac{2 \pi i}{p^b}\bigr)$ a primitive $p^b$-th root of unity, we have that $\chi(r)\in \langle \zeta\rangle=\{1,\zeta,\zeta^2,\ldots,\zeta^{p^b-1}\}$ for each $r \in \mathcal{R}.$ 
 
A subset $\mathcal{C}$ of $\mathcal{R}^{n}$ is called an \emph{additive code} over $\mathcal{R}$ of length $n$ if $\mathcal{C}$ is an additive subgroup of $\mathcal{R}^{n}.$ Clearly, $\mathcal{C} \subseteq \mathcal{R}^{n}$ is an additive code if and only if $\mathcal{C} \subseteq \mathcal{R}^{n}$ is a module over $\mathbb{Z}_{p^b}.$ An additive code $\mathcal{C}$ over $\mathcal{R}$ will be called \emph{free} if $\mathcal{C}$ is a free module over $\mathbb{Z}_{p^b}$. If an additive code, $\mathcal{C}$, over $\mathcal{R}$ is generated (as a $\mathbb{Z}_{p^b}$-module) by a subset $\mathcal{G}$ of $\mathcal{R}^n$, we will write $\mathcal{C}=\langle \mathcal G \rangle$. Since $\mathbb{Z}_{p^b}$ is also a commutative local ring, all minimal (with respect to inclusion) generating sets of finitely generated $\mathbb{Z}_{p^b}$-modules have the same cardinality. In particular, all minimal generating sets of an additive code $\mathcal{C} \subseteq \mathcal{R}^{n}$ have the same cardinality, and this number is called the \emph{rank} of $\mathcal{C}$, denoted by $\rank(\mathcal{C}).$ 

Although the following result is well known for free modules over $\mathbb{Z}_{p^b},$ we rewrite the statement of the result in terms of free additive codes over $\mathcal{R}$, and give a proof for the sake of completeness.
\begin{prop}\label{prop:free} For a free additive code over $\mathcal{R},$ any minimal generating set is linearly independent over $\mathbb{Z}_{p^b}.$ 
\end{prop}
\begin{proof} Let $\mathcal{C}$ be a free additive code over $\mathcal{R}$, with $\{v_1,v_2,\ldots,v_t\}$ being a basis of $\mathcal{C}$ (as a free module over $\mathbb{Z}_{p^b}$), where $t=\rank(\mathcal{C})$. Recall that all minimal generating sets of $\mathcal{C}$ have the same cardinality $t ~(=\rank(\mathcal{C}))$; let $\{u_1,u_2,\ldots,u_t\}$ be any such set. Define a map $\Lambda: \mathcal{C} \to \mathcal{C}$ as $\Lambda\biggl(\sum\limits_{i=1 }^{t} r_i v_i\biggr)=\sum\limits_{i=1 }^{t} r_i u_i.$ Clearly, $\Lambda$ is a surjective $\mathbb{Z}_{p^b}$-module homomorphism. From this, using \cite[Proposition 1.2]{vasconcelos},  we get that $\Lambda$ is a $\mathbb{Z}_{p^b}$-module isomorphism. This implies that if $\sum\limits_{i=1 }^{t} r_i u_i ~(=\Lambda\biggl(\sum\limits_{i=1 }^{t} r_i v_i\biggr))~=0,$ then $\sum\limits_{i=1 }^{t} r_i v_i=0.$ As $\{v_1,v_2,\ldots,v_t\}$ is a basis of $\mathcal{C}$ as a $\mathbb{Z}_{p^b}$-module, we get $r_1=r_2=\ldots=r_t=0.$ Thus, $\{u_1,u_2,\ldots,u_t\}$ is linearly independent over $\mathbb{Z}_{p^b}$.
\end{proof}

For each $a=(a_0,a_1,\cdots,a_{n-1}) \in \mathcal{R}^n,$ the Hamming weight $\text{wt}_{H}(a)$ of $a$ is defined as $\text{wt}_{\mathrm{H}}(a):=|\{ i : a_i \neq 0\}|.$ The minimum Hamming distance $d_{\mathrm{H}}(\mathcal{C})$ of an additive code $\mathcal{C} \subseteq \mathcal{R}^{n}$ is defined as $d_{\mathrm{H}}(\mathcal{C}):=\min\{\text{wt}_{\text{H}}(a) :a \in \mathcal{C} \setminus \{0\}\}$. We in fact extend this notation to arbitrary subsets $A \subseteq \mathcal{R}^n$: $d_{\mathrm{H}}(A):=\min\{\text{wt}_{\text{H}}(a) :a \in A\setminus \{0\}\}$. If $A = \emptyset$, we set $d_H(A) := \infty$.
 
In this paper, we will mostly be concerned with additive codes of even blocklength. In this context, whenever we consider a tuple $(x,y) \in \mathcal{R}^{2n}$, it is to be implicitly understood that $x$ and $y$ both belong to $\mathcal{R}^n$. The \emph{symplectic weight} of a vector $(x,y)=(x_1,x_2,\ldots,x_n,y_1,y_2,\ldots,y_n) \in \mathcal{R}^{2n}$, denoted by $\text{wt}_{\mathrm{s}}(x,y),$ is defined as $$\text{wt}_{\mathrm{s}}(x,y):=|\{ i : (x_i,y_i) \neq 0\}|.$$ Analogous to the notion of minimum Hamming distance, the \emph{minimum symplectic distance} $d_{\mathrm{s}}(A)$ of a subset $A$ of $\mathcal{R}^{2n}$ is defined as  
 $$d_{\mathrm{s}}(A):=\min\{\text{wt}_{\mathrm{s}}(a) :a \in A\setminus \{0\}\}.$$
If $A = \emptyset$, we set $d_s(A) := \infty$. 

\begin{defin} \label{def:sip} The \emph{symplectic inner product} on $\mathcal{R}^{2n}$ is defined as $\langle (a,b) \mid (a',b')\rangle_{\mathrm{s}} := b \cdot a'-b' \cdot a$  for $a,b,a',b' \in \mathcal{R}^n.$ (Here $b \cdot a'$ and $b' \cdot a$ are the usual dot products in $\mathcal{R}^n.$) Furthermore, for a subset $\mathcal{C} \subseteq \mathcal{R}^{2n}$, we define
\begin{itemize}
\item the \emph{symplectic dual}, $\mathcal{C}^{\perp_\mathrm{s}}$, of $\mathcal{C}$ as
$$\mathcal{C}^{\perp_{\mathrm{s}}}=\{v \in \mathcal{R}^{2n} :  \langle c \mid v\rangle_{\mathrm{s}} = 0 \text{~for all~} c \in \mathcal{C}\},$$
\item and the \emph{$\chi$-symplectic dual}, $\mathcal{C}^{\perp_\chi}$, of $\mathcal{C}$ as $$\mathcal{C}^{\perp_\chi}=\{v \in \mathcal{R}^{2n} : \chi(\langle c \mid  v\rangle_{\mathrm{s}}) = 1 \text{~for all~} c \in \mathcal{C}\}.$$
\end{itemize}
\end{defin}

Note that for any subset $\mathcal{C} \subseteq \mathcal{R}^{2n}$, the duals defined above are additive codes in $\mathcal{R}^{2n}.$ The following lemma is a special case of Lemma 6 of Nadella and Klappenecker  \cite{nadella}.

\begin{lem}\label{card}  Let $\mathcal{C} \subseteq \mathcal{R}^{2n}$ be an additive code. Then
 $|\mathcal{C}||\mathcal{C}^{\perp_\chi}|=|\mathcal{R}^{2n}|.$
 \end{lem} 

The definitions below will be needed for our construction in Section~\ref{sec3} of EAQECCs from additive codes over $\mathcal{R}$.

\begin{defin}
\label{def:chi-self-ortho}
An additive code $\mathcal{C} \subseteq \mathcal{R}^{2n}$ is called \emph{$\chi$-self-orthogonal} if $\mathcal{C} \subseteq \mathcal{C}^{\perp_\chi}$.
\end{defin}

\begin{defin} \label{def:chi-ext}
 Let $\mathcal{C} \subseteq \mathcal{R}^{2n}$ be an additive code.
\begin{itemize}
\item A code $\mathcal{C}' \subseteq \mathcal{R}^{2(n+c)}$ is called a \emph{$\chi$-self-orthogonal extension} of $\mathcal{C}$ if $\mathcal{C}' \subseteq \mathcal{C}'^{\perp_\chi}$, and $\mathcal{C}$ can be obtained from $\mathcal{C}' $ by puncturing $\mathcal{C}' $ at the coordinates $n+1,n+2,\ldots,n+c, 2n+c+1,2n+c+2,\ldots,2n+2c$. The number $c$ is called the \emph{entanglement degree} of the extension.

\item A $\chi$-self-orthogonal extension of the code $\mathcal{C}$ with the least entanglement degree among all such extensions is called a \emph{minimal} $\chi$-self-orthogonal extension of $\mathcal{C}.$ 
\end{itemize}
\end{defin}

The reason for the nomenclature of ``entanglement degree'' will get clear when we describe EAQECCs in Sections~\ref{sec:eaqecc} and \ref{sec3}. Briefly, this is the number of pairs of maximally entangled qudits needed in the construction of an EAQECC from an additive code $\mathcal{C}$.

\subsection{Quantum stabilizer codes over local Frobenius rings}\label{sec:stabcodes}

Let $$\mathcal{B}=\{| x\rangle : x \in \mathcal{R}\}$$  be an orthonormal basis of ${\mathbb{C}}^{q}.$ The state of a unit system, a qudit, is a superposition of these basis states of the system and is given by

$$| \psi\rangle= \sum\limits_{x \in \mathcal{R}}^{} a_x  | x\rangle, \text{~where~} a_x \in \mathbb{C} \text{~and~}  \sum\limits_{x \in \mathcal{R}}^{} |a_x|^2 =1.$$
An orthonormal basis of the quantum state space of $n$ qudits ${\mathbb{C}}^{q^n}=({\mathbb{C}}^{q})^{\otimes n}$ is given by 

$$\mathcal{B}^{\otimes n}=\{| x\rangle =| x_1\rangle \otimes | x_2\rangle \otimes \cdots \otimes | x_n\rangle  ~:~  x=(x_1,x_2,\ldots,x_n) \in \mathcal{R}^n\}.$$

For $a \in \mathcal{R},$ define two linear maps $X(a)$ and $Z(b)$ on ${\mathbb{C}}^{q}$ by their action on the basis $\mathcal{B}$  as 
$$X(a)(| x\rangle) =| x+a\rangle \text{ and } Z(a)(| x\rangle) =\chi(ax)| x\rangle \text{~for all~} x \in \mathcal{R}.$$
Extend these maps to unitary maps on ${\mathbb{C}}^{q^n}$ as follows:
$$X(a)=X(a_1)\otimes X(a_2)\otimes \cdots\otimes X(a_n) \text{~and~} Z(a)=Z(a_1)\otimes Z(a_2)\otimes \cdots\otimes Z(a_n) \text{~for~} a=(a_1,a_2,\ldots,a_n) \in \mathcal{R}^n.$$
Clearly,  $$X(a)(| x\rangle) =| x+a\rangle \text{ and } Z(a)(| x\rangle) =\chi(a \cdot x)| x\rangle \text{~for all~} a, x \in \mathcal{R}^n,$$ where $a \cdot x=\sum\limits_{i=1}^{n} a_ix_i$ is the dot product in $\mathcal{R}^n.$
The set $\mathcal{E}_n(\mathcal{R}) := \{X(a)Z(b) : a,b \in \mathcal{R}^n\}$ is called an \emph{error basis} of the $n$-qudit error group defined below in Definition~\ref{def:Pauli}. The following lemma tells us about the multiplicative and commutative properties of elements of the error basis.

\begin{lem}\label{comm}\cite{luerssen} Let $P=X(a)Z(b)$ and $P'=X(a')Z(b')$ with $a,b,a',b' \in \mathcal{R}^n$ be two elements of $\mathcal{E}_n(\mathcal{R}).$ Then $P^\dagger=P^{-1}=\chi(b \cdot a)X(-a)Z(-b)$ and $PP'=\chi(b \cdot a')X(a+a')Z(b+b').$ Furthermore,  $P$ and $P'$ commute with each other if and only if $\chi(b \cdot a'-b' \cdot a)=1$ 
\end{lem}

To define the $n$-qudit error group, called the Pauli group, let us fix some notations first.
 Let \begin{equation*} N~=~\left\{\begin{array}{ll}
p^b & \text{ if~} p \text{~is odd} ;\\
2p^b & \text{ if~} p \text{~is even}.\\
\end{array}\right. \end{equation*}
Further, let $\omega \in \mathbb{C}^*$ be a primitive $N$-th root of unity.
\begin{defin} \cite{luerssen}  The \emph{Pauli group} $\mathcal{P}_n(\mathcal{R})$ is defined as 
$$\mathcal{P}_n(\mathcal{R}) := \{\omega^\ell X(a)Z(b) ~:~ 0 \leq \ell \leq N-1, a, b \in \mathcal{R}^n\}.$$ 
\label{def:Pauli}
\end{defin}
 Define a map $$\Psi : \mathcal{P}_n(\mathcal{R}) \to \mathcal{R}^{2n} \text{~as~} \Psi(\omega^\ell X(a)Z(b))=(a,b).$$ The map $\Psi$ is a surjective group homomorphism with $\text{ker~}\Psi =\{\omega^\ell I ~:~ 0 \leq\ell \leq N-1 \}.$
 
 The \emph{weight} of an operator $\omega^\ell X(a)Z(b)=\omega^\ell X(a_1)\otimes X(a_2)\otimes \cdots\otimes X(a_n)Z(a_1)\otimes Z(a_2)\otimes \cdots\otimes Z(a_n) \in \mathcal{P}_n(\mathcal{R})$ is defined as 
 $$\text{wt}(\omega^\ell X(a)Z(b)):=|\{ i : (a_i,b_i) \neq 0\}|.$$
 That is, the weight of an operator is the number of non-scalar components of the tensor product that forms the operator.
 
A quantum error correcting code over $\mathcal{R}$ is a $K$-dimensional subspace of $\mathbb{C}^{q^n}.$
A quantum code $\mathcal{Q}$ can detect an error $E \in \mathcal{P}_n(\mathcal{R})$ if and only if  $\langle u| E |v \rangle=\lambda_E\langle u | v \rangle$ for each $|u \rangle,|v \rangle \in \mathcal{Q}$, where $\lambda_E \in \mathbb{C}$ is a constant depending only on $E$.

 A quantum code $\mathcal{Q}$ has minimum distance $D$ if it can detect all errors of weight at most $D-1$ and cannot detect some error of weight $D$. A quantum code $\mathcal{Q}$ with minimum distance $D$ can correct any error of weight  at most  $\left\lceil \frac{D-1}{2}\right\rceil$. If $\mathcal{Q} \subseteq \mathbb{C}^{q^n}$ is a quantum code  of dimension $1,$ then clearly, for each  $E \in \mathcal{P}_n(\mathcal{R}),$ we have $\langle u | E | v \rangle=\lambda_E\langle u | v \rangle$ for each $|u \rangle,|v \rangle \in \mathcal{Q}.$ So in this case, the minimum distance of the quantum code $\mathcal{Q}$ is defined as the largest integer $D$ such that for each non-identity operator $E \in \mathcal{P}_n(\mathcal{R})$ with $\text{wt}(E) < D$, we have $\langle u | E | v \rangle=0$ for all $|u \rangle,|v \rangle \in \mathcal{Q}$. (See, for example, \cite{DG04}.) A quantum code $\mathcal{Q} \subseteq \mathbb{C}^{q^n}$ of dimension $K$ is referred to as an $((n,K))_q$ quantum code; additionally if it has minimum distance $D$, then it is an $((n,K,D))_q$ quantum code. The subscript $q$ may be dropped if there is no ambiguity. 
 
A quantum code $\mathcal{Q}$ is said to be \emph{non-degenerate} if for any two arbitrary correctable errors $E_1, E_2$ with $E_1 \neq E_2,$ $ E_1 | u \rangle$ and $E_2 | v \rangle$ are linearly independent for any $| u \rangle, | v \rangle \in \mathcal{Q}.$ A quantum code $\mathcal{Q}$ is said to be \emph{degenerate} if it is not non-degenerate --- see, for example, \cite{calder}.

\begin{defin}\cite{luerssen} \begin{itemize}\item[(a)] A subgroup $\mathcal{S}$ of $\mathcal{P}_n(\mathcal{R})$ is called a stabilizer group if 
$\mathcal{S}$ is abelian and $\mathcal{S}  \cap \text{ker } \Psi= \{I_{q{^{n}}}\}.$
\item[(b)] A subspace $\mathcal{Q}$ of $\mathbb{C}^{q^n}$ is called a quantum stabilizer code if there exists a stabilizer group $\mathcal{S}$ such that
$$\mathcal{Q}=\mathcal{Q}(\mathcal{S})~:= \{| x\rangle \in \mathbb{C}^{q^n} ~:~ V | x\rangle =| x\rangle \text{~for all~} V \in \mathcal{S}\}.$$
\end{itemize}
\end{defin}

The following theorem tells us about the dimension of a quantum stabilizer code.
\begin{thm}\label{dim} \cite{luerssen} Let $\mathcal{S} \subseteq \mathcal{P}_n(\mathcal{R}) $ be a stabilizer group and $\mathcal{Q}(\mathcal{S}) \subseteq \mathbb{C}^{q^n}$ be the corresponding quantum stabilizer code. Then, the dimension of $\mathcal{Q}(\mathcal{S})$ is equal to $ q^n/|\mathcal{S}|.$
\end{thm}


It is well known that the set of undetectable errors for a quantum stabilizer code $\mathcal{Q}(\mathcal{S})$ are those which commute with all the elements of $\mathcal{S}$ but are not themselves elements of $\mathcal{S}$ (ignoring the overall phase factor of the error and of the elements of $\mathcal{S}$). Hence the error-correction properties of a quantum stabilizer code depend on the centralizer of the stabilizer group. 

\subsection{EAQECCs over local Frobenius rings}\label{sec:eaqecc}

Quantum stabilizer codes can be constructed only from abelian subgroups of $\mathcal{P}_n(\mathcal{R}).$  To construct QECCs from non-abelian subgroups of $\mathcal{P}_n(\mathcal{R})$, there is a framework of entanglement-assisted quantum error correcting codes (EAQECCs), which involves the use of maximally entangled states shared between the transmitter and the receiver. Brun et al.\ \cite{brun} first proposed this construction from non-abelian subgroups of $\mathcal{P}_n(\mathbb{Z}_2)$. They constructed quantum stabilizer codes from stabilizer groups obtained by extending the operators in the non-abelian subgroups into an appropriate higher dimensional space to form abelian subgroups. Maximally entangled qubit pairs, termed as \emph{ebits}, are added for the extension; one qubit from each pair is held by the transmitter and the other qubit is held by the receiver. The qubits held at the receiver end are assumed to be error-free.

The basic idea of Brun et al.\ can be extended to construct EAQECCs from non-abelian subgroups of $\mathcal{P}_n(\mathcal{R}).$ For this, we need a method to extend a non-abelian subgroup of $\mathcal{P}_n(\mathcal{R})$ into an appropriate higher dimensional space to form an abelian group.  To this end, we first note that if $\mathcal{S}$ is a subgroup of $\mathcal{P}_n(\mathcal{R}),$ then $\Psi(\mathcal{S}) \subseteq \mathcal{R}^{2n}$ is an additive subgroup of $\mathcal{R}^{2n}.$ Moreover, by Lemma \ref{comm}, we see that an operator $P=\omega^\ell X(a)Z(b) \in \mathcal{P}_n(\mathcal{R})$ commutes with elements of the subgroup $\mathcal{S}$ if and only if $\chi(b \cdot a'-b' \cdot a)=1$ for each $\omega^{\ell'}X(a')Z(b')\in \mathcal{S}.$ From this, we observe that if $V_1 \in \mathcal{P}_n(\mathcal{R})$ commutes with all elements of $\mathcal{S},$ then $\Psi(V_1) \in \Psi (\mathcal{S})^{\perp_\chi}.$ 

Thus, extending $\mathcal{S}$ to make it an abelian subgroup $\mathcal{S}'$ of $\mathcal{P}_{n+c}(\mathcal{R})$ for some $c$ is equivalent to extending $\mathcal{C} := \Psi(\mathcal{S})$ to $\mathcal{C}' := \Psi(\mathcal{S})' \subseteq  \mathcal{R}^{2(n+c)}$ such that $\mathcal{C}' \subseteq \mathcal{C}'^{\perp_\chi}.$ In other words, it is equivalent to constructing a $\chi$-self-orthogonal extension $\mathcal{C}' \subseteq  \mathcal{R}^{2(n+c)}$ of the additive code $\mathcal{C} \subseteq \mathcal{R}^{2n}.$ In Theorem~\ref{thm:2step}, we provide a method to construct such a $\chi$-self-orthogonal extension. Then, in Theorem~\ref{thm:eaqecc}, we give a construction of an $((n+c,q^{n+c}/|\mathcal{C'}|))$ quantum stabilizer code from $\mathcal{C}'$, which will be the desired EAQECC. As in the Brun et al.\ framework, $c$ extra pairs of maximally entangled qudits are involved in the construction. The transmitter and receiver each hold one qudit from each maximally entangled pair; the qudits held by the receiver are assumed to be error-free. This quantum code is referred to as an $((n,q^{n+c}/|\mathcal{C'}|;c))$ EAQECC over $\mathcal{R}$. If $c=0$, then this is simply an $((n,q^{n}/|\mathcal{C}|))$ quantum stabilizer code.

As the maximally entangled qudits are assumed to be maintained error-free at the receiver end, we note that if $E=\omega^{\ell} X(a,a')Z(b,b') \in \mathcal{P}_{n+c}(\mathcal{R})$, with $a,b \in \mathcal{R}^{n},$ $a',b' \in \mathcal{R}^{c}$, is an error operating on the transmitted codeword, then we must have $a'=b'=(0,0,\ldots,0)\in \mathcal{R}^{c}.$ Thus, only errors of the form $\omega^{\ell} X(a,0)Z(b,0) \in \mathcal{P}_{n+c}(\mathcal{R})$, with $a,b \in \mathcal{R}^{n}$, are assumed to occur in this model. We then say that an $((n,q^{n+c}/|\mathcal{C}'| ;c))$EAQECC has minimum distance $D$ if it can detect all errors of the form $\omega^{\ell} X(a,0)Z(b,0) \in \mathcal{P}_{n+c}(\mathcal{R})$, with $a,b \in \mathcal{R}^{n}$,  of weight at most $D-1$, but it cannot detect some error of this form of weight $D.$  If for an EAQECC $\mathcal{Q}$, all errors of the form $\omega^{\ell} X(a,0)Z(b,0) \in \mathcal{P}_{n+c}(\mathcal{R})$, with $a,b \in \mathcal{R}^{n}$, are detectable, then the minimum distance of the code is defined as the largest integer $D$ such that for each non-identity operator $E= X(a,0)Z(b,0) \in \mathcal{P}_{n+c}(\mathcal{R})$, with $\text{wt}(E) < D$, we have $\langle u| E | v \rangle=0$ for all $|u \rangle,|v \rangle \in \mathcal{Q}$. An $((n,q^{n+c}/|\mathcal{C}'| ;c))$ EAQECC with minimum distance $D$ is referred to as an $((n,q^{n+c}/|\mathcal{C}'|,D ;c))$ EAQECC over $\mathcal{R}$; again, if $c=0$, then this is simply an $((n,q^{n}/|\mathcal{C}|,D))$ quantum stabilizer code.
 An EAQECC  with minimum distance $D$ can correct any error of the form $\omega^{\ell} X(a,0)Z(b,0) \in \mathcal{P}_{n+c}(\mathcal{R})$,  with $a,b \in \mathcal{R}^{n}$,  of weight at most  $\left\lceil \frac{D-1}{2}\right\rceil.$

\subsection{A summary of our main results}
In this paper, we provide a framework to construct EAQECCs from classical additive codes over a finite  commutative local Frobenius ring $\mathcal{R}$. Recall that $|\mathcal{R}| = q = p^a$ and $\text{char}\mathcal{R} = p^b$. Our main result  is as follows:

  \begin{thm}\label{dt1} Let $\mathcal{C} \subseteq \mathcal{R}^{2n}$ be an additive code,  i.e.,  $\mathcal{C}$ is a module over $\mathbb{Z}_{p^b}.$  From $\mathcal{C}$, we can construct an $((n,K, D; c))$ EAQECC over $\mathcal{R},$ where the number of pairs of maximally entangled qudits needed is $c = \frac12 \rank(\mathcal{C}/(\mathcal{C} \cap \mathcal{C}^{\perp_{{\chi}}}))$, the minimum distance is
\begin{equation*}
D~=~\left\{\begin{array}{ll}
d_{\mathrm{s}}(\mathcal{C}^{\perp_\chi}) & \text{~~if~~} \mathcal{C}^{\perp_\chi} \subseteq \mathcal{C} \\
d_{\mathrm{s}}(\mathcal{C}^{\perp_\chi}\setminus \mathcal{C}) & \text{~~otherwise} \, ,
\end{array}\right. 
\end{equation*}
and the dimension $K$ is bounded as $q^{n+c}/ (|\mathcal{C}| \, {p^{\sum_{t=1}^{b-1} (b-t)\rho_t}}) \leq  K \leq q^{n+c}/|\mathcal{C}|$, the $\rho_t$'s being numbers determined by a certain chain of subcodes of $\mathcal{C}$. If either \emph{(a)}~$\mathcal{C}$ is free or \emph{(b)}~$\mathcal{C}/(\mathcal{C} \cap \mathcal{C}^{\perp_{\chi}})$ is a free module over $\mathbb{Z}_{p^b}$, then $K = q^{n+c}/|\mathcal{C}|$. In the case of \emph{(b)}, we additionally have $c=\frac12 \bigl[\rank(\mathcal{C}) - \rank(\mathcal{C} \cap \mathcal{C}^{\perp_\chi})\bigr]$.
 \end{thm}

Theorem~\ref{dt1} is a direct consequence of Theorems~\ref{thm:2step} and \ref{thm:eaqecc} proved in Section~\ref{sec3}. The precise expression for the numbers $\rho_t$ can be found in the restatement of Theorem~\ref{dt1} at the end of that section.

\medskip
 
In the second part of the paper (Sections \ref{sec4}--\ref{sec5}), we explicitly derive the minimum number of pre-shared pairs of maximally entangled qudits required to construct an EAQECC from a submodule over the integer ring $\Z_{p^a}$, and more generally, the minimum number of pre-shared pairs of maximally entangled qudits required to construct an EAQECC from an additive code over a Galois ring.  We state our results in the two theorems below. 

\begin{thm}\label{dt2} Let  $\mathcal{C} \subseteq \mathbb{Z}_{p^a}^{2n}$ be a submodule. From $\mathcal{C}$, we can construct an $((n,K, D; c))$ EAQECC over $\Z_{p^a}$, where the minimum number, $c$, of pairs of maximally entangled qudits needed for the construction is equal to $\frac12 \rank(\mathcal{C}/(\mathcal{C} \cap \mathcal{C}^{\perp_{\mathrm{s}}}))$,  the minimum distance is 
\begin{equation*}
D~=~\left\{\begin{array}{ll}
d_{\mathrm{s}}(\mathcal{C}^{\perp_{\mathrm{s}}}) & \text{~~if~~} \mathcal{C}^{\perp_{\mathrm{s}}} \subseteq \mathcal{C} \\  
d_{\mathrm{s}}(\mathcal{C}^{\perp_{\mathrm{s}}}\setminus \mathcal{C} )  & \text{~~otherwise} \, , \end{array}\right. 
\end{equation*}
and the dimension $K$ is bounded as $p^{a(n+c)}/ (|\mathcal{C}| \, {p^{\sum_{t=1}^{a-1} (a-t)\rho_t}}) \leq  K \leq p^{a(n+c)}/|\mathcal{C}|$, the $\rho_t$'s being numbers determined by a certain chain of subcodes of $\mathcal{C}$. If either \emph{(a)}~$\mathcal{C}$ is free or \emph{(b)}~$\mathcal{C}/(\mathcal{C} \cap \mathcal{C}^{\perp_{\mathrm{s}}})$ is a free module over $\mathbb{Z}_{p^b}$, then $K = p^{a(n+c)}/|\mathcal{C}|$. In the case of \emph{(b)}, we additionally have $c=\frac12 \bigl[\rank(\mathcal{C}) - \rank(\mathcal{C} \cap \mathcal{C}^{\perp_{\mathrm{s}}})\bigr]$.
 \end{thm}

\medskip

Theorem~\ref{dt2} follows by putting together the results of Theorems~\ref{thm:2step} and \ref{thm:eaqecc}, and Theorem~\ref{4thm}, proved in Sections~\ref{sec3} and \ref{sec4}, respectively. The precise expression for the numbers $\rho_t$ can be found in the restatement of the theorem at the end of Section~\ref{sec4}.

\begin{thm}\label{dt3} Let $\mathcal{C} \subseteq \text{GR}(p^b,m)^{2n}$ be an additive code over the Galois ring $\text{GR}(p^b,m)$.   From $\mathcal{C}$, we can construct an  $((n,K,D;c))$ EAQECC over $\text{GR}(p^b,m),$ where the minimum number, $c$, of pairs of maximally entangled qudits needed for the construction is equal to  $\left\lceil \frac1{2m}\rank(\mathcal{C}/(\mathcal{C} \cap \mathcal{C}^{\perp_{\Tr}}))  \right\rceil$, the minimum distance is 
\begin{equation*}
D~=~\left\{\begin{array}{ll}
d_{\mathrm{s}}(\mathcal{C}^{\perp_{\Tr}}) & \text{~~if~~} \mathcal{C}^{\perp_{\Tr}} \subseteq \mathcal{C} \\
d_{\mathrm{s}}(\mathcal{C}^{\perp_{\Tr}}\setminus \mathcal{C} )  & \text{~~otherwise}\, , 
\end{array}\right.
\end{equation*}
and the dimension $K$ is bounded as $p^{bm(n+c)}/ (|\mathcal{C}| \, {p^{\sum_{t=1}^{b-1} (b-t)\rho_t}}) \leq  K \leq p^{bm(n+c)}/|\mathcal{C}|$, the $\rho_t$'s being numbers determined by a certain chain of subcodes of $\mathcal{C}$. If either \emph{(a)}~$\mathcal{C}$ is free or \emph{(b)}~$\mathcal{C}/(\mathcal{C} \cap \mathcal{C}^{\perp_{\Tr}})$ is a free module over $\mathbb{Z}_{p^b}$, then $K = p^{bm(n+c)}/|\mathcal{C}|$. In the case of \emph{(b)}, we additionally have $c=\left\lceil \frac1{2m}[\rank(\mathcal{C}) - \rank(\mathcal{C} \cap \mathcal{C}^{\perp_{\Tr}})]\right\rceil$.
\end{thm}

In the statement of the theorem above, $\mathcal{C}^{\perp_{\Tr}}=\{v \in \text{GR}(p^b,m)^{2n} :{\Tr}( \langle v | c\rangle_{\mathrm{s}}) = 0 \; \forall \, c \in \mathcal{C}\}$ is the trace-symplectic dual of $\mathcal{C}$, defined with respect to the generalized trace map ${\Tr}: \text{GR}(p^b,m) \to \mathbb{Z}_{p^b}$. Section~\ref{sec5} contains an exposition of the machinery of the map $\Tr$, and the theorem is a consequence of the results proved in that section. Again, the numbers $\rho_t$ are expressed precisely in the restatement of Theorem~\ref{dt3} at the end of that section.

Theorem~\ref{dt3} significantly extends the results of Galindo et al.\ \cite{galindo} and Nadkarni and Garani \cite{nadkarni} obtained for EAQECCs derived from codes over finite fields. For ready reference, we state below the result for finite fields obtained as a corollary of the theorem. This result can also be inferred from the work of Nadkarni and Garani \cite{nadkarni}.

\begin{cor}
\label{cor:Fpm}
 Let $\mathcal{C} \subseteq \mathbb{F}_{p^m}^{2n}$ be an additive code over the finite field $\F_{p^m}$. From $\mathcal{C}$, we can construct an  $((n,p^{m(n+c)}/{|\mathcal{C}|} ,D;c))$ EAQECC over $\mathbb{F}_{p^m},$ where the minimum number, $c$, of pairs of maximally entangled qudits needed for the construction is equal to  
$ \left\lceil  \frac{1}{2m} (\dim_{\F_p}(\mathcal{C}) - \dim_{\F_p}(\mathcal{C} \cap \mathcal{C}^{\perp_{\tr}})) \right\rceil$, and
\begin{equation*}
D~=~\left\{\begin{array}{ll}
d_{\mathrm{s}}(\mathcal{C}^{\perp_{\tr}}) & \text{~~if~~}  \mathcal{C}^{\perp_{\tr}} \subseteq \mathcal{C}\\ 
d_{\mathrm{s}}(\mathcal{C}^{\perp_{\tr}}\setminus \mathcal{C} )  & \text{~~otherwise}\,.
\end{array}\right. 
\end{equation*}
Here, $\mathcal{C}^{\perp_{\tr}} := \{v \in \mathbb{F}_{p^m} ^{2n} : \tr( \langle v | c\rangle_{\mathrm{s}}) = 0 \; \forall \, c \in \mathcal{C}\}$ is the trace-symplectic dual of $\mathcal{C}$ defined with respect to the trace map $\tr: \F_{p^m} \to \F_p$ given by $\tr(z) = z+z^p+z^{p^2}+\cdots+z^{p^{m-1}}$.
\end{cor}

  \section{Constructing EAQECCs from  Additive Codes over Local Frobenius Rings}\label{sec3}
  
In this section, we provide the details of our method for constructing EAQECCs from additive codes over a finite commutative local Frobenius ring $\mathcal{R}$ with generating character $\chi$. We again recall that $|\mathcal{R}| = q = p^a$ and $\text{char}\mathcal{R} = p^b$. At the heart of the construction is a mechanism to obtain a $\chi$-self-orthogonal extension $\mathcal{C}'$ of an additive code $\mathcal{C}$. This requires some preliminary development.

\subsection{Standard-Form Generating Sets of Additive Codes}
We start with some definitions.
\begin{defin} \label{def:sympl} A subset $\{a_{11},a_{12},a_{21},a_{22},\ldots,a_{e1},a_{e2}\}$
 of $\mathcal{R}^{2n}$ is said to be a \emph{symplectic subset} of $\mathcal{R}^{2n}$ if
  $\chi(\langle a_{i1} \mid a_{j1} \rangle_{\mathrm{s}}) = \chi(\langle a_{i2} \mid a_{j2} \rangle_{\mathrm{s}}) = \chi(\langle a_{i1} \mid a_{k 2} \rangle_{\mathrm{s}}) = 1$ and  $\chi(\langle a_{i1} \mid a_{i2} \rangle_{\mathrm{s}})\neq1$ for all $i ,j, k \in \{1,2,\ldots,e\}$ with $i \neq k.$
 \end{defin}

\begin{defin} Let $\mathcal{C} \subseteq \mathcal{R}^{2n}$ be an additive code, i.e.,  $\mathcal{C} \subseteq \mathcal{R}^{2n}$ is a module over $\mathbb{Z}_{p^b}.$ Further, let $\mathcal{G}$ be a generating set of  $\mathcal{C} \subseteq \mathcal{R}^{2n}$ as a $\mathbb{Z}_{p^b}$-module.
\begin{itemize}
\item A generator $g \in \mathcal{G}$ is called an \emph{isotropic generator} if  $\chi(\langle g | h \rangle_{\mathrm{s}})=1$ for all $h \in \mathcal{G}.$
\item Two generators $g,g' \in \mathcal{G}$ are called a \emph{hyperbolic pair} if $\chi(\langle g | g' \rangle_{\mathrm{s}})\neq1$, $\chi(\langle g| h \rangle_{\mathrm{s}})= 1$ for all $h \in \mathcal{G} \setminus \{g'\}$, and $\chi(\langle g'| h)\rangle_{\mathrm{s}})=1$ for all $h \in \mathcal{G} \setminus \{g\}$.
\end{itemize}  
\end{defin}
 
Note that a generator is isotropic if and only if  it belongs to $\mathcal{C}^{\perp_\chi}$. Thus, an additive code $\mathcal{C} \subseteq \mathcal{R}^{2n}$ with a generating set $\mathcal{G}$ is $\chi$-self-orthogonal if and only if all the generators in $\mathcal{G}$ are isotropic. 

\begin{defin} 
A generating set $\mathcal{G}$ of an additive code $\mathcal{C} \subseteq \mathcal{R}^{2n}$ is said to be in \emph{standard form} if it consists solely of isotropic generators and hyperbolic pairs.
\label{def:stdform}
\end{defin}

It is far from being clear at this stage whether generating sets in standard form exist for additive codes over rings. This fact has been established for additive codes over fields \cite{brun}, \cite{EAQECC_bookchapter}, \cite{galindo}, \cite{luo}, \cite{nadkarni} but not yet, to the best of our knowledge, for codes over finite rings beyond fields (see Remark~\ref{rem:lee_klapp}). We will extend this fundamental result to additive codes over any finite local Frobenius ring in Theorem~\ref{thm:stdform} a little later in this section. But first, we give an idea of why the notion of a standard-form generating set is useful.

 \begin{prop} \label{prop:ext}
 Let $\mathcal{C} \subseteq \mathcal{R}^{2n}$ be an additive code with a generating set in standard form, in which there are exactly $e$ hyperbolic pairs, $u_{i1},u_{i2}$, $i = 1,2,\ldots,e$. Then, $\mathcal{C}$ has a $\chi$-self-orthogonal extension $\mathcal{C}'  \subseteq \mathcal{R}^{2(n+c)}$ if and only if there is a symplectic subset $\{a_{11},a_{12},a_{21},a_{22},\ldots,a_{e1},a_{e2}\} \subset \mathcal{R}^{2c}$ of cardinality $2e$ such that $\chi(\langle u_{i1} \mid u_{i2} \rangle_{\mathrm{s}}) = \chi(\langle a_{i1} \mid a_{i2} \rangle_{\mathrm{s}})$ for $i=1,2,\ldots,e$.
 \end{prop}

 \begin{proof} Let $\mathcal{G} = \{u_{11},u_{12},\ldots,u_{e1},u_{e2},z_1,\ldots,z_d\}$ be a standard-form generating set of $\mathcal{C}$, with $z_1,\ldots,z_d$ being the isotropic generators. As the generators live in $\mathcal{R}^{2n}$, we can write them as
 $$
 u_{i1} = (v_i,w_i) \text{ and }  u_{i2} = (x_i,y_i) \text{ for } i = 1,2,\ldots,e,\ \text{ and }  z_j = (v_{e+j},w_{e+j}) \text{ for } j=1,2\ldots,d,
 $$
 where each of the components $u_i,v_i,x_i,y_i$ lies in $\mathcal{R}^n$. 
 
 Now, suppose that $\{a_{11},a_{12},a_{21},a_{22},\ldots,a_{e1},a_{e2}\} \subset \mathcal{R}^{2c}$ is a symplectic subset such that $\chi(\langle u_{i1} \mid u_{i2} \rangle_{\mathrm{s}}) = \chi(\langle a_{i1} \mid a_{i2} \rangle_{\mathrm{s}})$ for $i=1,2,\ldots,e$. Let $a_{i1} = (b_i,c_i)$ and $a_{i2} = (r_i,s_i)$, the components being from $\mathcal{R}^c$. We then extend the components $v_i,w_i,x_i,y_i$ of the generators in $\mathcal{G}$ to $v_i',w_i',x_i',y_i' \in \mathcal{R}^{n+c}$ as follows:
 \begin{gather}
 v_i' = (v_i,-b_i), \ w_i' = (w_i,c_i), \ x_i' = (x_i,-r_i),\ y_i' = (y_i,s_i) \ \text{ for } i = 1,2,\ldots,e, \label{eq:ext.1} \\
\text{and } v_i' = (v_i,0), w_i' = (w_i,0) \text{ for } i = e+1,\ldots,e+d, \label{eq:ext.2}
 \end{gather}
 where $0$ is the zero element of $\mathcal{R}^c$. Finally, set 
 \begin{equation}
 u_{i1}' := (v_i',w_i') \text{ and }  u'_{i2} = (x'_i,y'_i) \text{ for } i = 1,2,\ldots,e,\ \text{ and }  z'_j = (v'_{e+j},w'_{e+j}) \text{ for } j=1,2\ldots,d.
 \label{eq:ext.3}
 \end{equation}
 Then, $\mathcal{G}' = \{u'_{11},u'_{12},\ldots,u'_{e1},u'_{e2},z'_1,\ldots,z'_d\}$ generates an additive code $\mathcal{C}' \subseteq \mathcal{R}^{2(n+c)}$, and the generators in $\mathcal{G}'$ are all isotropic. For instance, $\langle u'_{i1} \mid u'_{i2}\rangle_s = w_i'x_i'-v_i'y_i' = w_ix_i-c_ir_i - (v_iy_i - b_is_i) = \langle u_{i1} \mid u_{i2}\rangle_s - \langle a_{i1} \mid a_{i2}\rangle_s$, so that 
 $$
 \chi(\langle u'_{i1} \mid u'_{i2}\rangle_s) = \chi(\langle u_{i1} \mid u_{i2}\rangle_s) \cdot \bigl(\chi(\langle a_{i1} \mid a_{i2}\rangle_s)\bigr)^{-1} = 1.
 $$ 
 It follows that $\mathcal{C}'$ is a $\chi$-self-orthogonal extension of $\mathcal{C}$.

Conversely, if $\mathcal{C}'$ is a $\chi$-self-orthogonal extension of $\mathcal{C}$, then it has codewords $u'_{11},u'_{12},\ldots,u'_{e1},u'_{e2}$ that are extensions of the generators $u_{11},u_{12},\ldots,u_{e1},u_{e2}$ that form the hyperbolic pairs in $\mathcal{G}$. That is, we can write
 $$
 u'_{i1} = \bigl((v_i,\hat{v}_i),(w_i,\hat{w}_i)\bigr) \text{ and }  u'_{i2} = \bigl((x_i,\hat{x}_i),(y_i,\hat{y}_i)\bigr) \text{ for } i = 1,2,\ldots,e,
 $$
for some $\hat{u}_i,\hat{v}_i,\hat{x}_i,\hat{y}_i \in \mathcal{R}^c$. Since $\mathcal{C}'$ is $\chi$-self-orthogonal, we have $\chi(\langle u'_{ik} \mid u'_{j\ell}\rangle_{\mathrm{s}}) = 1$ for all $i,j \in \{1,2,\ldots,e\}$ and $k, \ell \in \{1,2\}$. Hence, setting
$$
a_{i1} = (-\hat{v}_i,\hat{w}_i) \text{ and } a_{i2} = (-\hat{x}_i,\hat{y}_i) \text{ for } i = 1,2,\ldots,e,
$$
it is easy to verify that  $\{a_{11},a_{12},a_{21},a_{22},\ldots,a_{e1},a_{e2}\} \subset \mathcal{R}^{2c}$ is a symplectic subset such that $\chi(\langle u_{i1} \mid u_{i2} \rangle_{\mathrm{s}}) = \chi(\langle a_{i1} \mid a_{i2} \rangle_{\mathrm{s}})$ for $i=1,2,\ldots,e$. 
 \end{proof}

The proof of Proposition~\ref{prop:ext} shows that a $\chi$-self-orthogonal extension $\mathcal{C}' \subseteq \mathcal{R}^{2(n+c)}$ of an additive code $\mathcal{C} \subseteq \mathcal{R}^{2n}$ can be constructed in two steps:
\begin{itemize}
\item[(1)] Find a generating set $\mathcal{G}$ in standard form of the additive code $\mathcal{C} \subseteq \mathcal{R}^{2n}$.
\item[(2)] Find a symplectic subset of $\mathcal{R}^{2c}$, for some suitable choice of $c$, satisfying the property required by Proposition~\ref{prop:ext}. The desired code $\mathcal{C}'$ is generated by the set $\mathcal{G}'$ obtained by extending the generators in $\mathcal{G}$ by $2c$ coordinates as prescribed in \eqref{eq:ext.1}--\eqref{eq:ext.3}. 
\end{itemize}   
Indeed, this prescription for a construction of a $\chi$-self-orthogonal extension is important enough for our development that we will give it a name: the Two-Step Construction. 

In Theorems~\ref{thm:stdform} and \ref{thm:2step} below, we will show that it is always possible to construct a $\chi$-self-orthogonal extension of an additive code $\mathcal{C} \subseteq \mathcal{R}^{2n}$ by following the Two-Step Construction. The first of these theorems, stated next, shows that any additive code $\mathcal{C} \subseteq \mathcal{R}^{2n}$ has a (minimal) generating set that is in standard form. The proof is by construction of such a set.

\begin{thm}\label{thm:stdform} Let $\mathcal{C} \subseteq \mathcal{R}^{2n}$ be an additive code. There exists a symplectic subset $\mathcal{T} \subseteq \mathcal{C} \setminus \mathcal{C}^{\perp_\chi}$ of cardinality $|\mathcal{T}| = \rank(\mathcal{C}/(\mathcal{C} \cap \mathcal{C}^{\perp_\chi}))$, such that the following statements hold for \emph{any} minimal generating set $\mathcal{S}$ of the additive code $\mathcal{C} \cap \mathcal{C}^{\perp_\chi}$:
\begin{itemize}
\item[(a)] $\mathcal{S} \cup \mathcal{T}$ is a generating set of $\mathcal{C}$ in standard form, with $\mathcal{S}$ being the set of isotropic generators and $\mathcal{T}$ being the set of hyperbolic pairs\footnote{We use ``the set of hyperbolic pairs'' as shorthand for ``the set of generators that form hyperbolic pairs''. Thus, the cardinality of a set of hyperbolic pairs is equal to the number of generators that the set contains, which is actually \emph{twice} the number of hyperbolic pairs formed by these generators.\label{fn:hyppairs}}.
\item[(b)] There is a subset $\mathcal{S}_0 \subseteq \mathcal{S}$ such that $\mathcal{S}_0 \cup \mathcal{T}$ is a minimal generating set of $\mathcal{C}$.
\item[(c)] If $\mathcal{C}/(\mathcal{C} \cap \mathcal{C}^{\perp_\chi})$ is free (as a module over $\mathbb{Z}_{p^b}$), then $\mathcal{S} \cup \mathcal{T}$ is itself a minimal generating set of $\mathcal{C}$.
\end{itemize}

\end{thm}  
\begin{proof} 
Let $\pi : \mathcal{C}  \to \mathcal{C}/(\mathcal{C} \cap \mathcal{C}^{\perp_\chi})$ be the canonical projection map that takes $(v,w) \in \mathcal{C}$ to the coset $(v,w) + (\mathcal{C} \cap \mathcal{C}^{\perp_\chi})$. Let $\mathcal{T}_0=\{(a_1,b_1),(a_2,b_2),\ldots,(a_f,b_f)\} \subseteq \mathcal{C} \setminus \mathcal{C}^{\perp_\chi}$ be such that $\pi(\mathcal{T}_0) := \{\pi((a_1,b_1)),\pi((a_2,b_2)),\ldots,\pi((a_f,b_f))\}$ is a minimal generating set of $\mathcal{C}/(\mathcal{C} \cap \mathcal{C}^{\perp_\chi})$ as a $\mathbb{Z}_{p^b}$-module. Thus, $|\mathcal{T}_0| = |\pi(\mathcal{T}_0)| = \rank(\mathcal{C}/(\mathcal{C} \cap \mathcal{C}^{\perp_\chi}))$.

Now,  let $\mathcal{S}$ be \emph{any} minimal generating set of the additive code $\mathcal{C} \cap \mathcal{C}^{\perp_\chi}  \subseteq \mathcal{R}^{2n}$ as a $\mathbb{Z}_{p^b}$-module. We assert that  $\mathcal{S} \cup \mathcal{T}_0$ generates $\mathcal{C}$.
To see this, note that for any $(v,w) \in \mathcal{C},$ we must have 
$\pi(v,w)=z_1\pi((a_1,b_1))+z_2\pi((a_2,b_2))+\ldots+z_f\pi((a_f,b_f))$ 
for some $z_i \in \mathbb{Z}_{p^b}.$  Since $z_1 \pi((a_1,b_1)) + z_2\pi((a_2,b_2)) + \cdots + z_f \pi((a_f,b_f)) = \pi (z_1 (a_1,b_1) + z_2 (a_2,b_2) + \cdots + z_f (a_f,b_f))$, we have $(v,w)-z_1(a_1,b_1)-z_2(a_2,b_2)-\ldots-z_f(a_f,b_f) \in \mathcal{C} \cap \mathcal{C}^{\perp_\chi},$ which implies that $(v,w) \in \langle \mathcal{S} \cup \mathcal{T}_0 \rangle$. This completes the proof of the assertion.


(a)\ Since the generators in $\mathcal{S}$ belong to $\mathcal{C}^{\perp_\chi}$, they are all isotropic. We will not tamper with $\mathcal{S}$; instead, we will transform $\mathcal{T}_0$ into a symplectic subset. For $1 \leq i,j \leq f,$ $\chi(b_i \cdot a_j-b_j \cdot a_i)$ is a $p^b$-th root of unity, so let $\chi(b_i \cdot a_j-b_j \cdot a_i)=\zeta^{\ell_{i,j}},$ where $\zeta = \exp(\frac{2\pi i}{p^b})$ and $0\leq \ell_{i,j} < p^b.$ Further, for $1 \leq i,j \leq f,$ we note that $\chi(b_j \cdot a_i-b_i \cdot a_j)=\chi(b_i \cdot a_j-b_j \cdot a_i)^{-1},$ which gives $\ell_{i,j} \equiv -\ell_{j,i} \pmod{p^b}$. For each $i \in \{1,2,\ldots,f\}$, we must have $\ell_{i,j} \ne 0$ for some $j \in \{1,\ldots,f\}$; otherwise, $(a_i,b_i)$ is in $\mathcal{C} \cap \mathcal{C}^{\perp_\chi}$, so that $\pi((a_i,b_i))=\mathcal{C} \cap \mathcal{C}^{\perp_\chi},$ contradicting the minimality of $\pi(\mathcal{T}_0)$.

Let $\ell_{\hat{\imath},\hat{\jmath}}$ with $\hat\imath,\hat\jmath \in \{1,2, \ldots,f\}$ be such that
$\gcd(\ell_{\hat\imath,\hat\jmath},p^b)=\min\{\gcd(\ell_{i,j},p^b) : \ell_{i,j}\neq 0 \text{~and~} 1 \leq i,j \leq f \}.$ Clearly, $\gcd(\ell_{\hat\imath,\hat\jmath},p^b)$ divides $\gcd(\ell_{i,j},p^b)$ for $1 \leq i,j \leq f.$
 As $\ell_{\hat\imath,\hat\jmath}\neq 0,$  we have  $\chi(b_{\hat\imath} \cdot a_{\hat\jmath} - b_{\hat\jmath} \cdot a_{\hat\imath}) \neq 1.$   Swap $(a_1,b_1)$ and $(a_2,b_2)$  with $(a_{\hat\imath},b_{\hat\imath})$ and $(a_{\hat\jmath},b_{\hat\jmath}),$ respectively.

For $3 \leq i \leq f,$   replace $ (a_i,b_i)$ with 
 $$(a_i',b_i')=(a_i,b_i) +u_i \, (a_1,b_1) + v_i \, (a_2,b_2),$$ 
 where $u_i$ and $v_i$ are solutions of the linear equations 
$$\ell_{1,2}u_i\equiv\ell_{2,i}(\text{mod~}p^b) \ \ \ \text{ and } \ \ \ \ell_{1,2}v_i\equiv-\ell_{1,i}(\text{mod~}p^b).$$
Such $u_i$ and $v_i$ always exist, since $\gcd(\ell_{1,2},p^b)$ divides $\ell_{2,i}$ and $\ell_{1,i}$. By doing this, we get a new set 
$$\mathcal{T}_1 = \{(a_1,b_1),(a_2,b_2),(a_3',b_3'),\ldots,(a_f',b_f')\}$$
such that $\chi(b_2 \cdot a_1-b_1 \cdot a_2)\neq1$, $\chi(b_j' \cdot a_1-b_1 \cdot a_j')=1$, and $\chi(b_j' \cdot a_2-b_2 \cdot a_j')=1$  for $j \in \{3,4,\ldots,f\}$. In other words, $(a_1,b_1)$ and $(a_2,b_2)$ form a hyperbolic pair. Since $\mathcal{T}_0$ is recoverable from $\mathcal{T}_1$, we see that $\mathcal{S} \cup \mathcal{T}_1$ also generates $\mathcal{C}$. Moreover, $\pi(\mathcal{T}_0)$ is recoverable from $\pi(\mathcal{T}_1) := \{\pi((a_1,b_1)),\pi((a_2,b_2)),\pi((a_3',b_3')),\ldots,\pi((a_f',b_f'))$, so $\pi(\mathcal{T}_1)$ is also a minimal generating set of $\mathcal{C}/(\mathcal{C} \cap \mathcal{C}^{\perp_\chi})$. In particular, $|\mathcal{T}_1| = |\pi(\mathcal{T}_1)| = \rank(\mathcal{C}/(\mathcal{C} \cap \mathcal{C}^{\perp_\chi}))$.

By repeatedly applying the above process to the generators $(a_3',b_3'),\ldots,(a_f',b_f'),$ we will eventually obtain a symplectic subset 
 $$\mathcal{T} = \{(v_1,w_1),(v_2,w_2),\ldots,(v_{c},w_{c}),(x_1,y_1),(x_2,y_2),\ldots, (x_{c},y_{c})\}$$ 
such that $\mathcal{S} \cup \mathcal{T}$ generates $\mathcal{C}$, the generators $(v_i,w_i)$ and $(x_i,y_i)$, $i=1,2,\ldots,c$, are hyperbolic pairs, and $\pi(\mathcal{T})$ is a minimal generating set of $\mathcal{C}/(\mathcal{C} \cap \mathcal{C}^{\perp_\chi})$. In particular, $|\mathcal{T}| = |\pi(\mathcal{T})| = \rank(\mathcal{C}/(\mathcal{C} \cap \mathcal{C}^{\perp_\chi}))$.

(b)\ Note that no generator in $\mathcal{T}$ can be obtained by taking $\mathbb{Z}_{p^b}$-linear combinations of other generators   in $\mathcal{S} \cup \mathcal{T}$. Indeed, if, say, $(v_1,w_1)$ were expressible as a linear combination of other generators in $\mathcal{S} \cup \mathcal{T}$, we would end up with $\chi(\langle (v_1,w_1) | (x_1,y_1) \rangle_{\mathrm{s}}) = 1$, which would contradict the fact that $(v_1,w_1)$ and $(x_1,y_1)$ form a hyperbolic pair. However, it is possible that some of the generators in $\mathcal{S}$ can be obtained by taking linear combinations of other generators in $\mathcal{S} \cup \mathcal{T}$. Iteratively removing such generators from $\mathcal{S}$,  we are left with an $\mathcal{S}_0 \subseteq \mathcal{S}$ such that $\mathcal{S}_0 \cup \mathcal{T}$ is a minimal generating set of $\mathcal{C}$. Observe that $\mathcal{S}_0 \cup \mathcal{T}$ is also a generating set in standard form.

(c)\ Suppose that $\mathcal{C}/(\mathcal{C} \cap \mathcal{C}^{\perp_\chi})$ is free as a module over $\mathbb{Z}_{p^b}$. By the argument in (b) above, it suffices to show that no generator in $\mathcal{S}$ can be obtained as a linear combination of other generators in $\mathcal{S} \cup \mathcal{T}$. To prove this, we use the fact that $\{\pi((v_1,w_1)),\pi((v_2,w_2)),$ $\ldots, \pi((v_{c},w_{c})),$ $\pi((x_1,y_1)),\pi((x_2,y_2)),\ldots, \pi((x_{c},y_{c}))\}$ is a linearly independent set over $\mathbb{Z}_{p^b}$ --- this holds because, by Proposition \ref{prop:free},  any minimal generating set of the free $\mathbb{Z}_{p^b}$-module $\mathcal{C}/(\mathcal{C} \cap \mathcal{C}^{\perp_\chi})$ is linearly independent.


Now, suppose that some $(v,w) \in \mathcal{S}$ can be obtained as a linear combination of other generators in $\mathcal{S} \cup \mathcal{T}$, i.e., 
\begin{equation}\label{eq:iso}
(v,w)=a_1(v_1,w_1)+a_2(v_2,w_2)+\cdots+a_c(v_c,w_c)+b_1(x_1,y_1)+b_2(x_2,y_2)+\cdots+b_c(x_c,y_c)+(f,g) 
\end{equation} 
for some $a_i,b_i \in \mathbb{Z}_{p^b}$, and some $(f,g) \in \mathcal{C} \cap \mathcal{C}^{\perp_\chi}$ that is a linear combination of elements in  $\mathcal{S} \setminus (v,w)$. Applying the projection map $\pi$ to \eqref{eq:iso}, we get 
$0=a_1\pi(v_1,w_1)+a_2\pi(v_2,w_2)+\cdots+a_c\pi(v_c,w_c)+b_1\pi(x_1,y_1)+b_2\pi(x_2,y_2)+\cdots+b_c\pi(x_c,y_c).$
Since  $\{\pi((v_1,w_1)),\pi((v_2,w_2)),$ $\ldots, \pi((v_{c},w_{c})),$ $\pi((x_1,y_1)),\pi((x_2,y_2)),\ldots, \pi((x_{c},y_{c}))\}$ is a linearly independent set over $\mathbb{Z}_{p^b},$ we have $a_i=b_i=0$ for all $i$. Plugging this back into \eqref{eq:iso}, we find that $(v,w) \in \mathcal{S}$ is expressible as a linear combination of elements in $\mathcal{S} \setminus (v,w)$, which contradicts the minimality of $\mathcal{S}$. We conclude that no generator in $\mathcal{S}$ can be obtained as a linear combination of other generators in $\mathcal{S} \cup \mathcal{T}$, and hence, that $\mathcal{S} \cup \mathcal{T}$ is a minimal generating set of $\mathcal{C}$.
\end{proof}

In the following example, we illustrate the steps described in the proof of Theorem~\ref{thm:stdform} to find, for an additive code, a minimal generating set in standard form.
\begin{ex}  Let $\mathcal{C}$ be an additive code over $\mathbb{Z}_{16}$ of length 8 with a minimal generating set 
$\{(a_1,b_1),(a_2,b_2), $ $(a_3,b_3),$ $(a_4,b_4),(a_5,b_5)\},$
 where $(a_1,b_1)=(1,3,0,0,0,0,0,0),$  $(a_2,b_2)=(0,0,0,0,1,1,0,0),$ $(a_3,b_3)=(0,0,1,0,2,  $ $2,$ $0,0),$ $(a_4,b_4)=(2,6,0,0,0,0,8,0)$ and $(a_5,b_5)=(0,0,0,1,0,0,0,0).$ The map $\chi: \mathbb{Z}_{16} \to \mathbb{C}^*$ defined by $\chi(z) = \zeta^z$, with $\zeta=\exp\bigl(\frac{\pi i}{8}\bigr)$, is a generating character of the ring $\mathbb{Z}_{16}$. 
 
 We first need to find minimal generating sets of $\mathcal{C}/(\mathcal{C} \cap \mathcal{C}^{\perp_\chi})$ and $\mathcal{C} \cap \mathcal{C}^{\perp_\chi}$ as $\mathbb{Z}_{16}$-modules. We begin by observing that $ \langle 4(a_1,b_1),4(a_2, $ $b_2), 2(a_3,b_3),$ $(a_5,b_5) \rangle \subseteq \mathcal{C} \cap \mathcal{C}^{\perp_\chi}$. We will argue that the reverse inclusion also holds. To see this, consider any $\alpha_1(a_1,b_1)+\alpha_2(a_2,b_2)+\alpha_3(a_3,b_3)+\alpha_4(a_4,b_4)+\alpha_5(a_5,b_5) \in \mathcal{C} \cap \mathcal{C}^{\perp_\chi}$. As $(a_5,b_5) \in  \mathcal{C} \cap \mathcal{C}^{\perp_\chi},$ we have $(c_1,d_1):=\alpha_1(a_1,b_1)+\alpha_2(a_2,b_2)+\alpha_3(a_3,b_3)+\alpha_4(a_4,b_4) \in \mathcal{C} \cap \mathcal{C}^{\perp_\chi}$.  This implies that $\chi(\langle (c_1,d_1) | (a_1,b_1) \rangle_{\mathrm{s}}) =\chi(\langle (c_1,d_1) | (a_2,b_2) \rangle_{\mathrm{s}}) =\chi(\langle (c_1,d_1) | (a_3,b_3) \rangle_{\mathrm{s}})=\chi(\langle (c_1,d_1) | (a_4,b_4) \rangle_{\mathrm{s}})=1$, which  gives 
 $4\alpha_2+8\alpha_3=0,$ $4\alpha_1+8\alpha_4=0,$ $8\alpha_1+8\alpha_4=0$ and $8\alpha_2+8\alpha_3=0.$ Thus, $\alpha_1,\alpha_2 \in \langle 4 \rangle$ and $\alpha_3,\alpha_4 \in \langle 2 \rangle.$  Consequently, $$\mathcal{C} \cap \mathcal{C}^{\perp_\chi} \subseteq \langle 4(a_1,b_1),4(a_2,b_2), 2(a_3,b_3),2(a_4,b_4),(a_5,b_5) \rangle= \langle 4(a_1,b_1),4(a_2,b_2), 2(a_3,b_3),(a_5,b_5) \rangle$$ 
the equality above arising from the fact that $2(a_4,b_4) = 4(a_1,b_1)$. We have thus shown that $ \mathcal{C} \cap \mathcal{C}^{\perp_\chi} = \langle 4(a_1,b_1),4(a_2,$ $b_2), $ $2(a_3,b_3),(a_5,b_5) \rangle$, meaning that $\{4(a_1,b_1),4(a_2,b_2), $ $2(a_3,b_3),(a_5,b_5) \}$
 is a generating set of $\mathcal{C} \cap \mathcal{C}^{\perp_\chi}$ as a $\mathbb{Z}_{16}$-module. It is straightforward to check that this is in fact a minimal generating set of $\mathcal{C} \cap \mathcal{C}^{\perp_\chi}$.
 
  Now, let $\pi : \mathcal{C}  \to \mathcal{C}/(\mathcal{C} \cap \mathcal{C}^{\perp_\chi})$ be the canonical projection map that takes $(v,w) \in \mathcal{C}$ to the coset $(v,w) + (\mathcal{C} \cap \mathcal{C}^{\perp_\chi})$. It can be verified that $\{\pi((a_1,b_1)),\pi((a_2,b_2)), \pi((a_3,b_3)),$ $\pi((a_4,b_4)) \}$ is a minimal generating set of $\mathcal{C}/(\mathcal{C} \cap \mathcal{C}^{\perp_\chi})$ as a  $\mathbb{Z}_{16}$-module. So, set (as in the proof of Theorem~\ref{thm:stdform}), $\mathcal{T}_0:=\{(a_1,b_1),(a_2,b_2),(a_3,b_3),(a_4,b_4)\}$.  
 With $\mathcal{S}:=\{4(a_1,b_1),4(a_2,b_2), $ $2(a_3,b_3),(a_5,b_5)\}$, we note (as argued in the proof of the theorem) that $\mathcal{S} \cup \mathcal{T}_0$ is a generating set of $\mathcal{C}$. 
 
 Following the proof of Theorem \ref{thm:stdform}(a), we will transform $\mathcal{T}_0$ into a symplectic subset. For $1 \leq i,j \leq 4,$ let $\chi(b_i \cdot a_j-b_j \cdot a_i)=\zeta^{\ell_{i,j}}.$ One easily computes $\ell_{1,2}=12,$ $\ell_{1,3}=8,$  $\ell_{1,4}=0,$ $\ell_{2,3}=0,$ $\ell_{2,4}=8$ and $\ell_{3,4}=8.$ Note that $\gcd(\ell_{1,2},16)=\min\{\gcd(\ell_{i,j},16) : \ell_{i,j}\neq 0 \text{~and~} 1 \leq i,j \leq 4 \}=4.$ For $i =3,4,$ set
 $$(a_i',b_i')=(a_i,b_i) +u_i \, (a_1,b_1) + v_i \, (a_2,b_2),$$ 
 where $u_i$ and $v_i$ are solutions of the linear equations 
$$12u_i\equiv\ell_{2,i}(\text{mod~}16) \ \ \ \text{ and } \ \ \ 12v_i\equiv-\ell_{1,i}(\text{mod~}16).$$ Solving these equations, we get $u_3=0,v_3=2$, $u_4=2$ and $v_4=0$, so that $(a'_3,b'_3)=(0,0,1,0,4,4,0,0)$ and $(a'_4,b'_4)=(4,12,0,0,0,0,8,0)$. For $i = 3,4$, replacing $(a_i,b_i)$ in $\mathcal{T}_0$ with $(a_i',b_i')$, we get a new set $\mathcal{T}_1 :=\{(a_1,b_1),(a_2,b_2),(a'_3,b'_3),(a'_4,b'_4)\}$. This $\mathcal{T}_1$ is a symplectic subset, and $\mathcal{S} \cup \mathcal{T}_1$ is a generating set of $\mathcal{C}$ in standard form. From this, we readily obtain that $\{(a_5,b_5)\}  \cup \mathcal{T}_1$ is a minimal generating set of $\mathcal{C}$ in standard form.  \qed \end{ex}
  

The following result is an easy consequence of Theorem~\ref{thm:stdform}.

\begin{cor}
For an additive code $\mathcal{C} \subseteq \mathcal{R}^{2n}$, we have $\rank(\mathcal{C}/(\mathcal{C} \cap \mathcal{C}^{\perp_\chi})) \ge \rank(\mathcal{C}) - \rank(\mathcal{C} \cap \mathcal{C}^{\perp_\chi})$. Equality holds if $\mathcal{C}/(\mathcal{C} \cap \mathcal{C}^{\perp_\chi})$ is free as a module over $\mathbb{Z}_{p^b}$.
\label{cor:ranks}
\end{cor}
\begin{proof}
Let $\mathcal{S}$ and $\mathcal{T}$ be as in the statement of Theorem~\ref{thm:stdform}. Since $\mathcal{S}$ is a minimal generating set of $\mathcal{C} \cap \mathcal{C}^{\perp_\chi}$, we have $|\mathcal{S}| = \rank(\mathcal{C} \cap \mathcal{C}^{\perp_\chi})$. Hence, by part (a) of the theorem, we have
$$
\rank(\mathcal{C}) \le |\mathcal{S} \cup \mathcal{T}| = |\mathcal{S}| + |\mathcal{T}| = \rank(\mathcal{C} \cap \mathcal{C}^{\perp_\chi}) + \rank(\mathcal{C}/(\mathcal{C} \cap \mathcal{C}^{\perp_\chi})).
$$
If $\mathcal{C}/(\mathcal{C} \cap \mathcal{C}^{\perp_\chi})$ is free, then by part (c) of the theorem, the first inequality above is in fact an equality. 
\end{proof}

In the remainder of this subsection, we establish some useful structural properties of \emph{any} standard-form generating set of an additive code $\mathcal{C} \subseteq \mathcal{R}^{2n}$, which we will use in the sequel.

Recall that for each $r \in \mathcal{R}$, we have $\chi(r)\in \langle \zeta \rangle=\{1,\zeta,\zeta^2,\ldots,\zeta^{p^b-1}\}$, where $\zeta = \exp\bigl(\frac{2 \pi i}{p^b}\bigr)$ is a primitive $p^b$-th root of unity. The subgroups of $\langle \zeta \rangle$ are $H_t := \langle \zeta^{p^{b-t}} \rangle = \{1,\zeta^{p^{b-t}},\zeta^{2p^{b-t}},\ldots,\zeta^{(p^{t}-1)p^{b-t}}\}$, for $t = 0,1,\ldots,b$.  The subgroup $H_0$ is simply $\langle 1 \rangle = \{1\}$, and the indexing is chosen so that $H_0 \le H_1 \leq \cdots \le H_b = \langle \zeta \rangle$. For a subset $\mathcal{C} \subseteq \mathcal{R}^{2n}$ and $t = 0,1,\ldots,b$, define
\begin{equation}
\mathcal{C}^{\perp_{\chi,t}} = \{(v',w') \in \mathcal{R}^{2n}: \chi(\langle(v,w) \mid (v',w')\rangle_{\mathrm{s}}) \in H_t \text{ for all } (v,w) \in \mathcal{C}\}.
\label{def:perp_t}
\end{equation}
Note that $\mathcal{C}^{\perp_{\chi,t}}$ is an additive code. Clearly, we have the nested sequence $\mathcal{C}^{\perp_\chi} = \mathcal{C}^{\perp_{\chi,0}} \subseteq \mathcal{C}^{\perp_{\chi,1}} \subseteq \cdots \subseteq \mathcal{C}^{\perp_{\chi,b}} = \mathcal{R}^{2n}$.

\begin{prop}\label{prop:count} 
Let $\mathcal{C} \subseteq \mathcal{R}^{2n}$ be an additive code, and $\mathcal{G}$ any generating set of $\mathcal{C}$ in standard form. Let $\mathcal{S}$ and $\mathcal{T}$ denote the sets of isotropic generators and hyperbolic pairs, respectively, in $\mathcal{G}$. The following then hold:
\begin{itemize}
\item[(a)] For $t = 0,1,\ldots,b$, we have $\left|\mathcal{T} \cap \mathcal{C}^{\perp_{\chi,t}}\right| = \rank(\mathcal{C} / (\mathcal{C} \cap \mathcal{C}^{\perp_{\chi}})) - \rank(\mathcal{C} / (\mathcal{C} \cap \mathcal{C}^{\perp_{\chi,t}}))$. In particular\footnote{See Footnote~\ref{fn:hyppairs}.}, $\left|\mathcal{T}\right| = \rank(\mathcal{C} / (\mathcal{C} \cap \mathcal{C}^{\perp_{\chi}}))$.
\item[(b)] There is a subset $\mathcal{S}_0 \subseteq \mathcal{S}$ such that $\mathcal{S}_0 \cup \mathcal{T}$ is a minimal generating set of $\mathcal{C}$.
\end{itemize}
\end{prop}
\begin{proof}
(a)\ Let $t$ be any integer in $\{0,1,\ldots,b\}$, and set $H = H_t$. Let $\mathcal{V} = \mathcal{T} \setminus \mathcal{C}^{\perp_{\chi,t}}$ be the set of hyperbolic generators\footnote{Of course, $\mathcal{V}$ also depends on $t$, but we suppress this dependence from the notation for simplicity.} that do \emph{not} belong to $\mathcal{C}^{\perp_{\chi,t}}$. We will first show that $|\mathcal{V}| = \rank(\mathcal{C} / (\mathcal{C} \cap \mathcal{C}^{\perp_{\chi,t}}))$.

Let $(v_i,w_i)$ and $(x_i,y_i)$, $i=1,2,\ldots,c$, be the hyperbolic pairs that make up $\mathcal{T}$. Observe that if $\chi( \langle (v_i,w_i) \mid (x_i,y_i) \rangle_{\mathrm{s}}) \in H$, then by the way hyperbolic pairs are defined, it follows that both $(v_i,w_i)$ and $(x_i,y_i)$ are in $\mathcal{C}^{\perp_{\chi,t}}$. So, the generators constituting such hyperbolic pairs cannot be in $\mathcal{V}$. On the other hand, if $\chi( \langle (v_i,w_i) \mid (x_i,y_i) \rangle_{\mathrm{s}}) \notin H$, then $(v_i,w_i)$ and $(x_i,y_i)$ are both in $\mathcal{V}$. Thus, we may assume that $\mathcal{V}$ consists of the hyperbolic pairs $(v_i,w_i)$ and $(x_i,y_i)$, $i=1,2,\ldots,h$, and hence, $\chi( \langle (v_i,w_i) \mid (x_i,y_i) \rangle_{\mathrm{s}}) \notin H$ for these pairs.

Let $\pi : \mathcal{C}  \to \mathcal{C}/(\mathcal{C} \cap \mathcal{C}^{\perp_{\chi,t}})$ be the canonical projection map that takes $(v,w) \in \mathcal{C}$ to the coset $[v,w] := (v,w) + (\mathcal{C} \cap \mathcal{C}^{\perp_{\chi,t}})$. Since $\pi$ maps any generator in $\mathcal{G} \setminus \mathcal{V}$ to $\mathcal{C} \cap \mathcal{C}^{\perp_{\chi,t}}$, i.e., to the coset $[0,0]$, it must be the case that $\pi(\mathcal{V}) := \{\pi((v_1,w_1)), \pi((x_1,y_1)), \ldots, \pi((v_h,w_h)), \pi((x_h,y_h))\}$ is a generating set for the quotient $\mathcal{C} / (\mathcal{C} \cap \mathcal{C}^{\perp_{\chi,t}})$, again viewed as a $\mathbb{Z}_{p^b}$-module. We claim that for no generator in $\mathcal{V}$ is it the case that its image under $\pi$ is expressible as a $\mathbb{Z}_{p^b}$-linear combination of the $\pi$-images of other generators in $\mathcal{V}$. It then follows that the restriction of $\pi$ to $\mathcal{V}$ is one-to-one, so that $|\mathcal{V}| = |\pi(\mathcal{V})|$, and that $\pi(\mathcal{V})$ is a minimal generating set of $\mathcal{C} / (\mathcal{C} \cap \mathcal{C}^{\perp_{\chi,t}})$, so that $|\pi(\mathcal{V})| = \rank(\mathcal{C} / (\mathcal{C} \cap \mathcal{C}^{\perp_{\chi,t}}))$. Thus, proving the claim suffices to prove that $|\mathcal{V}| = \rank(\mathcal{C} / (\mathcal{C} \cap \mathcal{C}^{\perp_{\chi,t}}))$.

We prove the claim by way of contradiction. Assume, to the contrary, that $\pi((v_1,w_1))$ is expressible as 
$$
\pi((v_1,w_1)) = \eta_1 \cdot \pi((x_1,y_1)) + \sum_{i=2}^h [\theta_i \cdot \pi((v_i,w_i)) + \eta_i \cdot \pi((x_i,y_i))]
$$
for some $\eta_i$'s and $\theta_i$'s from $\mathbb{Z}_{p^b}$. Then, $\pi\bigl((v_1,w_1) - \eta_1 \cdot (x_1,y_1) - \sum_{i=2}^h [\theta_i \cdot (v_i,w_i) + \eta_i \cdot (x_i,y_i)]\bigr) = [0,0]$, or equivalently, $(v_1,w_1) - \eta_1 \cdot (x_1,y_1) - \sum_{i=2}^h [\theta_i \cdot (v_i,w_i) + \eta_i \cdot (x_i,y_i)] =: (v',w')$ belongs to $\mathcal{C} \cap \mathcal{C}^{\perp_{\chi,t}}$.

Consider the quantity $\kappa := \chi( \langle (v',w') \mid (x_1,y_1) \rangle_{\mathrm{s}})$. Since $(v',w') \in \mathcal{C}^{\perp_{\chi,t}}$ and $(x_1,y_1) \in \mathcal{C}$, we must have $\kappa \in H$. On the other hand, 
since $(v_1,w_1)$ and $(x_1,y_1)$ form a hyperbolic pair,  $\kappa$ reduces to $\chi( \langle (v_1,w_1) \mid (x_1,y_1) \rangle_{\mathrm{s}})$, and moreover, since $(v_1,w_1)$ and $(x_1,y_1)$ are in $\mathcal{V}$, we have $\kappa = \chi( \langle (v_1,w_1) \mid (x_1,y_1) \rangle_{\mathrm{s}}) \notin H$. This is the desired contradiction that proves the claim, and hence, the fact that $|\mathcal{V}| = \rank(\mathcal{C} / (\mathcal{C} \cap \mathcal{C}^{\perp_{\chi,t}}))$.

Now, observe that in the special case of $t = 0$, we have $\mathcal{V}$ being all of $\mathcal{T}$, and consequently, $\left|\mathcal{T}\right| = \rank(\mathcal{C} / (\mathcal{C} \cap \mathcal{C}^{\perp_{\chi}}))$. More generally, for any $t$, we have that $\mathcal{T}$ is the disjoint union of $\mathcal{V}$ and $\mathcal{T} \cap \mathcal{C}^{\perp_{\chi,t}}$; hence, 
\begin{align*}
\left|\mathcal{T} \cap \mathcal{C}^{\perp_{\chi,t}}\right| \ &= \ \left|\mathcal{T}\right| - |\mathcal{V}| \\
&= \ \rank(\mathcal{C} / (\mathcal{C} \cap \mathcal{C}^{\perp_{\chi}})) - \rank(\mathcal{C} / (\mathcal{C} \cap \mathcal{C}^{\perp_{\chi,t}})),
\end{align*}
which completes the proof of part (a) of the proposition.

(b)\ The argument here is identical to that for part (b) of Theorem~\ref{thm:stdform}.
\end{proof}

We will make use of this result in the proof of Theorem~\ref{thm:2step} that is to follow.

\subsection{Constructing a $\chi$-Self-Orthogonal Extension} \label{sec:chi-self-ortho}
 
For the second step in the Two-Step Construction of a $\chi$-self-orthogonal extension of an additive code $\mathcal{C}$, we need a suitable symplectic subset. Such a subset can always be obtained from the standard-form generating set of $\mathcal{C}$ guaranteed by Theorem~\ref{thm:stdform}(a), as we explain in the proof of the first part of the theorem below.
  
  \begin{thm}\label{thm:2step}  For any additive code $\mathcal{C} \subseteq \mathcal{R}^{2n}$, the following statements hold:
  \begin{itemize}
  \item[(a)] It is always possible to construct a $\chi$-self-orthogonal extension of $\mathcal{C}$, with entanglement degree equal to $\frac12 \rank(\mathcal{C}/(\mathcal{C} \cap \mathcal{C}^{\perp_\chi}))$, using the Two-Step Construction. 
  \item[(b)]  For any $\mathcal{C}'$ that is a $\chi$-self-orthogonal extension of $\mathcal{C}$ obtained using the Two-Step Construction, we have $|\mathcal{C}|  \leq |\mathcal{C}'| \leq |\mathcal{C}| \cdot p^{\sum_{t=1}^{b-1} (b-t)  \rho_t}$, with $\rho_t:=\rank(\mathcal{C}/(\mathcal{C} \cap \mathcal{C}^{\perp_{\chi,t-1}}))-\rank(\mathcal{C}/(\mathcal{C} \cap \mathcal{C}^{\perp_{\chi,t}}))$. Additionally, if either $\mathcal{C}$ or $\mathcal{C}/(\mathcal{C} \cap \mathcal{C}^{\perp_\chi})$ is a free module over $\mathbb{Z}_{p^b},$ then $|\mathcal{C}'| =|\mathcal{C}|$.
  \end{itemize}
\end{thm}

\begin{proof}
(a)\ For Step~(1) of the Two-Step Construction, we appeal to Theorem~\ref{thm:stdform}(a) to get a standard-form generating set of $\mathcal{C}$. Let 
\begin{equation}
\mathcal{G}=\{(v_1,w_1),(v_2,w_2),\ldots, (v_{c+d},w_{c+d}),(x_1,y_1),(x_2,y_2),\ldots, (x_{c},y_{c})\}
\label{eq:stdformG}
\end{equation}
be this generating set, where, for $i = 1,2\ldots,c$, the generators $(v_i,w_i)$ and $(x_i,y_i)$ form hyperbolic pairs, and the generators $(v_i,w_i)$, $i = c+1,\ldots,c+d$, are isotropic. Moreover, $2c = \rank(\mathcal{C}/(\mathcal{C} \cap \mathcal{C}^{\perp_\chi}))$ by Proposition~\ref{prop:count}. 

For $i = 1,2,\ldots,c$, set $a_{i1} := (-\gamma_i e_i,0)$ and $a_{i2} := (0,e_i)$, where $\gamma_i = \langle(v_i,w_i) \mid (x_i,y_i)\rangle_{\mathrm{s}}$, $e_i = (0,\ldots,0,1,0,$ $\ldots,$ $0) \in \mathcal{R}^c$, the $1$ occurring in the $i$th coordinate, and $0$ is the zero element of $\mathcal{R}^c$. It is then easy to verify that $\{a_{11},a_{12},a_{21},a_{22},\ldots,a_{c1},a_{c2}\} \subset \mathcal{R}^{2c}$ is a symplectic subset such that $\chi(\langle (v_i,w_i) \mid (x_i,y_i) \rangle_{\mathrm{s}}) = \chi(\langle a_{i1} \mid a_{i2} \rangle_{\mathrm{s}})$ for $i=1,2,\ldots,c$, as required by Step~(2) of the Two-Step Construction.

Hence, by Proposition~\ref{prop:ext}, $\mathcal{C}$ has a $\chi$-self-orthogonal extension $\mathcal{C}' \subseteq \mathcal{R}^{2(n+c)}$. Indeed, $\mathcal{C}'$ is the additive code generated by $\mathcal{G}' = \{u'_{11},u'_{12},\ldots,u'_{c1},u'_{c2},z'_1,\ldots,z'_d\}$, with 
$$
 u_{i1}' := (v_i,-\gamma_i e_i,w_i,0) \text{ and }  u'_{i2} = (x_i,0,y_i,e_i)  \text{ for } i = 1,2,\ldots,c,
 $$
 and  $z'_j = (v_{c+j},0,w_{c+j},0) \text{ for } j=1,2\ldots,d$, where $0$ is the zero element of $\mathcal{R}^c$. 

(b)\  Let $\mathcal{C}'$ be any $\chi$-self-orthogonal extension of $\mathcal{C}$ obtained using the Two-Step Construction. Let $\mathcal{G}$ as in \eqref{eq:stdformG} be the standard-form generating set of $\mathcal{C}$ used in the construction; by virtue of Proposition~\ref{prop:count}(b), we may assume, without loss of generality, that $\mathcal{G}$ is a minimal generating set of $\mathcal{C}$. We will use $\mathcal{T}$ to denote the set of all hyperbolic pairs in $\mathcal{G}$, i.e., $\mathcal{T} = \{(v_1,w_1),(v_2,w_2),\ldots, (v_{c},w_{c}),(x_1,y_1),(x_2,y_2),\ldots, (x_{c},y_{c})\}$.

Now, let
$$
\mathcal{G}'=\{(v_1,m_1,w_1,n_1),\ldots,(v_{c},m_{c},w_{c},n_{c}),(v_{c+1},0,w_{c+1},0),\ldots, (v_{c+d},0,w_{c+d},0),(x_1,r_1,y_1,s_1),\ldots, (x_{c},r_c,y_{c},s_c)\}
$$
be the generating set of $\mathcal{C}'$ that extends $\mathcal{G}$ in the manner of \eqref{eq:ext.1}--\eqref{eq:ext.3}.
Further, let $\varphi$ denote the restriction of $\mathcal{C}'$ onto the coordinates of $\mathcal{C}$, i.e., $\varphi$ maps $(v,f,w,g) \in \mathcal{C}'$ to $(v,w) \in \mathcal{C}$. Finally, let $\mathcal{C}_0$ denote the subcode of $\mathcal{C}'$ consisting of codewords of the form $(0,f,0,g)$, for some $f,g \in \mathcal{R}^c$. Then, $\ker \varphi = \mathcal{C}_0$, and hence, $|\mathcal{C}'| = |\mathcal{C}| \cdot |\mathcal{C}_0|$. Since $(0,0,0,0) \in \mathcal{C}_0$, we have $|\mathcal{C}_0| \ge 1$, so that $|\mathcal{C}'| \ge |\mathcal{C}|$. We next work towards an upper bound on $|\mathcal{C}_0|$.


From the definitions of the generating sets $\mathcal{G}$ and $\mathcal{G}'$, we see that $|\mathcal{C}_0|$ equals the number of distinct $(f,g) \in \mathcal{R}^{2c}$ that can be obtained as 
\begin{equation}
\sum_{i=1}^c [\xi_{i,1} \cdot (m_i,n_i) + \xi_{i,2} \cdot (r_i,s_i)] + \sum_{j=1}^d \eta_j \cdot (0,0)
\label{eq:fg}
\end{equation}
for some  $\xi_{i,1},\xi_{i,2},\eta_j \in \mathbb{Z}_{p^b}$ such that 
\begin{equation}
\sum_{i=1}^c [\xi_{i,1} \cdot (v_i,w_i) + \xi_{i,2} \cdot (x_i,y_i)] + \sum_{j=1}^d \eta_j \cdot (v_{c+j},w_{c+j}) \ = \ (0,0).
\label{eq:C0}
\end{equation}
The $\eta_j$ terms do not contribute to the overall sum in \eqref{eq:fg}, so can be dropped from there (but not from \eqref{eq:C0}).

Fix some $\ell \in \{1,2,\ldots,c\}$. Taking the symplectic product with $(x_{\ell},y_{\ell})$ on both sides of \eqref{eq:C0}, and applying the generating character $\chi$, we obtain ${\chi(\gamma_\ell)}^{\xi_{\ell,1}} = 1$, where we recall that $\gamma_\ell = \langle(v_\ell,w_\ell) \mid (x_\ell,y_\ell)\rangle_{\mathrm{s}}$. Similarly, taking the symplectic product with $(v_{\ell},w_{\ell})$ on both sides of \eqref{eq:C0}, and applying the generating character $\chi$, we obtain ${\chi(\gamma_\ell)}^{\xi_{\ell,2}} = 1$. Note that since $(v_{\ell},w_{\ell})$ and $(x_{\ell},y_{\ell})$ form a hyperbolic pair, we have $\chi(\gamma_\ell) \neq 1$. Let $t$ be the least positive integer such that $\chi(\gamma_\ell) \in H_t := \langle \zeta^{p^{b-t}} \rangle$. This means that $\chi(\gamma_\ell)$ is of the form $\exp\bigl(\frac{2\pi i \, z}{p^b}\bigr)$ with $\gcd(z,p^b) = p^{b-t}$. Then, for ${\chi(\gamma_\ell)}^{\xi_{\ell,1}} = {\chi(\gamma_\ell)}^{\xi_{\ell,2}} = 1$, we must have $\xi_{\ell,1}$ and $\xi_{\ell,2}$ being multiples of $p^t$. In particular, if $t=b$, then $\xi_{\ell,1}=\xi_{\ell,2}=0$ (in $\mathbb{Z}_{p^b}$). From all of this, we infer that whenever \eqref{eq:C0} holds, the expression in \eqref{eq:fg} is in fact of the form 
\begin{equation*}
\sum_{t = 1}^{b-1} \sum_{\ell \in \mathcal{I}_t}  [\xi_{\ell,1} \cdot (m_\ell,n_\ell) + \xi_{\ell,2} \cdot (r_\ell,s_\ell)],
\end{equation*}
where $\mathcal{I}_t = \{i : \text{$t$ is the least positive integer such that $\chi(\gamma_i) \in H_t$}\}$, and $\xi_{\ell,1},\xi_{\ell,2} \in \mathbb{Z}_{p^b}$ are multiples of $p^t$. In particular, for $\ell \in \mathcal{I}_t$, we can have at most $p^{b-t}$ possible choices for each of $\xi_{\ell,1}$ and $\xi_{\ell,2}$. Consequently, 
$$
|\mathcal{C}_0| \le \prod_{t=1}^{b-1} \prod_{\ell \in \mathcal{I}_t} p^{2(b-t)} = p^{\sum_{t=1}^{b-1} 2(b-t)|\mathcal{I}_t|}.
$$

Now, observe that $\ell \in \mathcal{I}_t$ iff $(v_{\ell},w_{\ell})$ and $(x_\ell,y_\ell)$ are both in $\mathcal{C}^{\perp_{\chi,t}} \setminus \mathcal{C}^{\perp_{\chi,t-1}}$ --- this can be inferred from the definition of $\mathcal{C}^{\perp_{\chi,t}}$ in \eqref{def:perp_t}. Hence, $2|\mathcal{I}_t| = \left|\mathcal{T} \cap \mathcal{C}^{\perp_{\chi,t}} \right| - \left|\mathcal{T} \cap \mathcal{C}^{\perp_{\chi,t-1}} \right|$, which from Proposition~\ref{prop:count}, equals $\rho_t := \rank(\mathcal{C}/(\mathcal{C} \cap \mathcal{C}^{\perp_{\chi,t-1}}))-\rank(\mathcal{C}/(\mathcal{C} \cap \mathcal{C}^{\perp_{\chi,t}}))$. Thus, we conclude that $|\mathcal{C}_0| \le p^{\sum_{t=1}^{b-1} (b-t) \rho_t}$, and hence, $|\mathcal{C}'| = |\mathcal{C}|\cdot |\mathcal{C}_0|  \le |\mathcal{C}| \cdot p^{\sum_{t=1}^{b-1} (b-t) \rho_t}$.

Finally, let $\mathcal{C}/(\mathcal{C} \cap \mathcal{C}^{\perp_\chi})$ be free as a $\mathbb{Z}_{p^b}$-module. Further, let $\pi : \mathcal{C}  \to \mathcal{C}/(\mathcal{C} \cap \mathcal{C}^{\perp_\chi})$ be the canonical projection map. Clearly, $\pi(\mathcal{T}) := \{\pi((v_1,w_1)),\pi((v_2,w_2)),$ $\ldots, \pi((v_{c},w_{c})),\pi((x_1,y_1)),$ $\pi((x_2,y_2)),\ldots, \pi((x_{c},y_{c}))\}$ is a generating set of $\mathcal{C}/(\mathcal{C} \cap \mathcal{C}^{\perp_\chi})$ as a $\mathbb{Z}_{p^b}$-module. Since $2c= \rank(\mathcal{C}/(\mathcal{C} \cap \mathcal{C}^{\perp_\chi})),$ we see that $\pi(\mathcal{T})$ is in fact a minimal generating set of $\mathcal{C}/(\mathcal{C} \cap \mathcal{C}^{\perp_\chi})$. 
Recall  that, by Proposition \ref{prop:free}, for a finite  free  module over $\mathbb{Z}_{p^b}$, any minimal generating set  is linearly independent over $\mathbb{Z}_{p^b}.$ This implies that $\pi(\mathcal{T})$ is a linearly independent set over $\mathbb{Z}_{p^b}.$ Applying $\pi$ to both sides of \eqref{eq:C0}, we see that \eqref{eq:C0} holds only if all the coefficients $\xi_{i,1},\xi_{i,2}$ are equal to $0$. 

On the other hand, if $\mathcal{C}$ is free as a  $\mathbb{Z}_{p^b}$-module, then, using Proposition \ref{prop:free}, $\mathcal{G},$ being a minimal generating set, is linearly independent over $\mathbb{Z}_{p^b}.$  Thus \eqref{eq:C0} holds only if all the coefficients $\xi_{i,1},\xi_{i,2},\eta_j$ are equal to $0$. Thus, if either $\mathcal{C}$ or $\mathcal{C}/(\mathcal{C} \cap \mathcal{C}^{\perp_\chi})$ is a free module over $\mathbb{Z}_{p^b},$ then $\mathcal{C}_0$ consists of only the $(0,0,0,0)$ codeword, so that $|\mathcal{C}'|=|\mathcal{C}|\cdot |\mathcal{C}_0| = |\mathcal{C}|$.  This completes the proof of the theorem.
\end{proof}

%

\subsection{Constructing an EAQECC from a $\chi$-Self-Orthogonal Extension}
It only remains to provide a means of constructing an EAQECC from a $\chi$-self-orthogonal extension of $\mathcal{C}$ obtained via the Two-Step Construction. This is done in the theorem below, which also specifies the dimension and minimum distance of the resulting EAQECC.
 
\begin{thm}\label{thm:eaqecc}
Let $\mathcal{C} \subseteq \mathcal{R}^{2n}$ be an additive code over $\mathcal R$,  and let $\mathcal{C}' \subseteq \mathcal{R}^{2(n+c)}$ be a $\chi$-self-orthogonal extension of $\mathcal{C}$ with entanglement degree $c$, obtained using the Two-Step Construction. Then, there  exists an $((n,q^{n+c}/{|\mathcal{C}'|} ,D;c))$ EAQECC,  where
\begin{equation*}D~=~\left\{\begin{array}{ll}
 d_{\mathrm{s}}(\mathcal{C}^{\perp_\chi}) & \text{~~if~~} \mathcal{C}^{\perp_\chi}\subseteq \mathcal{C}\, ;\\
d_{\mathrm{s}}(\mathcal{C}^{\perp_\chi}\setminus \mathcal{C} )  & \text{~~otherwise} \,.
\end{array}\right. \vspace{-1mm}\end{equation*}
\end{thm}
\begin{proof} From the given $\chi$-self-orthogonal extension, $\mathcal{C'}$, of $\mathcal{C}$, we will first construct a stabilizer group $\mathcal{A} \le \mathcal{P}_{n+c}(\mathcal{R})$ of size equal to $|\mathcal{C}'|$. This is done by adapting the proof of the ``$\Leftarrow$'' part of Theorem~3.12 in \cite{luerssen}.

We start by defining
$$ \mathcal{W}:=\{ \omega^\ell X(v)Z(w) : 0 \leq \ell \leq N-1, (v,w) \in \mathcal{C}'\} \subseteq \mathcal{P}_{n+c}(\mathcal{R}).$$
Note that $\mathcal{W}$ is an abelian subgroup of $\mathcal{P}_{n+c}(\mathcal{R}).$ 
Let $\xi$ be a nontrivial character of $(\mathcal{W},\cdot)$ such that $\xi(\omega^\ell I)=\omega^\ell$ for $0 \leq \ell < N.$ (A character with this property does exist as any character of a subgroup $\langle \omega I \rangle$ of the finite abelian group $\mathcal{W}$  can be
extended to a character of $\mathcal{W}$.) 
Now, define
$$ \mathcal{A}:=\{ \xi(\chi(vw)X(-v)Z(-w))X(v)Z(w) : (v,w) \in \mathcal{C}'\} \subseteq \mathcal{P}_{n+c}(\mathcal{R}).$$
Following the proof of the ``$\Leftarrow$'' part of Theorem~3.12 in \cite{luerssen}, it can be verified that $\mathcal{A}$ is an abelian subgroup of $\mathcal{P}_{n+c}(\mathcal{R})$ such that $\mathcal{A} \cap \text{ker } \Psi= \{I_{q{^{(n+c)}}}\}$. Thus, $\mathcal{A}$ is a stabilizer group of size $|\mathcal{A}| = |\mathcal{C}'|$.


 
  By Theorem \ref{dim}, the quantum code $\mathcal{Q}(\mathcal{A}) \subseteq \mathbb{C}^{q^{n+c}}$ has dimension equal to $q^{n+c}/|\mathcal{A}|.$ This proves the existence of an $((n,q^{n+c}/|\mathcal{C}'|; c))$ EAQECC over $\mathcal{R}.$

We now proceed to determine the minimum distance $D$ of the EAQECC. Let $C(\mathcal{A})$  be a centralizer of $\mathcal{A}$ in $\mathcal{P}_{n+c}(\mathcal{R}),$ i.e., $C(\mathcal{A})$ is a set of all the elements of  $\mathcal{P}_{n+c}(\mathcal{R})$ which commute with all the elements of $\mathcal{A}.$ Further, let $E=\omega^{\ell} X(a,a')Z(b,b') \in \mathcal{P}_{n+c}(\mathcal{R})$ with $a,b \in \mathcal{R}^{n},$ $a',b' \in \mathcal{R}^{c}$ be an error occurred on the transmitted codeword. As the receiver-end qudits are assumed to be maintained error-free, we have  $E=\omega^{\ell} X(a,0)Z(b,0) \in \mathcal{P}_{n+c}(\mathcal{R}).$ Further, as the overall phase factor of the error does not matter, we can assume that $E= X(a,0)Z(b,0) \in \mathcal{P}_{n+c}(\mathcal{R})$; indeed, from now on, in this proof, we will ignore the overall phase factor of the elements of $C(\mathcal{A})$ and $\mathcal{A}.$ Note that an element $E= X(a,0)Z(b,0) \in C(\mathcal{A})$  if and only if $(a,b) \in \mathcal{C}^{\perp_\chi}.$

If $E \not\in C(\mathcal{A})$, then we have $EF\neq FE$ for some $F=\omega^{k} X(c,c')Z(d,d') \in \mathcal{A}.$ That is, we have $\chi(b\cdot c-a\cdot d)\neq 1.$ Using this, we observe that $$\langle u | E| v \rangle=\langle u | EF| v \rangle=\chi(b\cdot c-a\cdot d) \langle u | FE| v \rangle=\chi(b\cdot c-a\cdot d) \langle u | E| v \rangle.$$
As $\chi(b\cdot c-a\cdot d)\neq1,$ we get $\langle u | E| v \rangle=0.$ This implies that $E$ is detectable. 
Further, when $E \in\mathcal{A},$ we have $\langle u | E| v \rangle=\langle u |  v \rangle,$ hence $E$ is detectable. Now we assume that $E \in C(\mathcal{A}) \setminus \mathcal{A}$. Here we assert that $E$ is not detectable. We will prove this by contradiction. Suppose that $E$ is detectable. This implies that for each $|v \rangle \in \mathcal{Q}(\mathcal{A}),$ we have $E|v \rangle =\lambda_E |v \rangle$ some complex scalar $\lambda_E.$ Note that for the projection operator $P= \frac{1}{|\mathcal{A}|}\sum\limits_{F \in \mathcal{A}}^{} F,$  we have $EP=\lambda_EP,$ which implies that $\lambda_E \neq 0.$ Now consider an abelian subgroup $\mathcal{A}'$ of $\mathcal{P}_{n+c}(\mathcal{R})$ generated by $\lambda_E^{-1}E$ and all elements of $\mathcal{A}.$ By Theorem \ref{dim}, the quantum code $\mathcal{Q}(\mathcal{A}') \subseteq \mathbb{C}^{q^{n+c}}$ has dimension equal to $q^{n+c}/|\mathcal{A}'|<q^{n+c}/|\mathcal{A}|.$ This implies that not all elements of $\mathcal{Q}(\mathcal{A})$ remain invariant under $\lambda_E^{-1}E.$ This is a contradiction to the detectability of $E.$ This proves our assertion.

Thus an error $E= X(a,0)Z(b,0) \in \mathcal{P}_{n+c}(\mathcal{R})$ is not detectable if and only if $E \in C(\mathcal{A}) \setminus \mathcal{A}$,  i.e., an error $E= X(a,0)Z(b,0) \in \mathcal{P}_{n+c}(\mathcal{R})$ is not detectable if and only if $(a,b) \in \mathcal{C}^{\perp_\chi} \setminus \mathcal{C}.$  From this, we see that if $\mathcal{C}^{\perp_\chi} \setminus \mathcal{C} \not=  \emptyset,$ then $D=d_{\mathrm{s}}(\mathcal{C}^{\perp_\chi}\setminus \mathcal{C}).$ Now, suppose that $\mathcal{C}^{\perp_\chi} \subseteq \mathcal{C}.$  Then, there is no element in  $C(\mathcal{A}) \setminus \mathcal{A}$ of the form  $X(a,0)Z(b,0),$ hence all errors are detectable. Note that  $E=X(a,0)Z(b,0) \in \mathcal{P}_{n+c}(\mathcal{R})$ with $\text{wt}(E) < d_{\mathrm{s}}(\mathcal{C}^{\perp_\chi})$ does not belong to $C(\mathcal{A})$ as $(a,b) \not \in \mathcal{C}^{\perp_\chi} .$ This implies that for each non-identity operator $E=X(a,0)Z(b,0) \in \mathcal{P}_{n+c}(\mathcal{R}),$ with $\text{wt}(E) < d_{\mathrm{s}}(\mathcal{C}^{\perp_\chi}),$ we have $\langle u | E| v \rangle=0$ for each $| u \rangle, | v \rangle \in \mathcal{Q}(\mathcal{A})$.  Further,  for each $(a,b) \in \mathcal{C}^{\perp_\chi}\subseteq \mathcal{C},$ we see that $E =X(a,0)Z(b,0) \in \mathcal{A},$ which implies that $ \langle u | E| v \rangle=\langle u | v \rangle.$ Thus the minimum distance $D$ of the code is equal to $d_{\mathrm{s}}(\mathcal{C}^{\perp_\chi}).$ This completes the proof of the theorem.\end{proof}

Observe that Theorem~\ref{dt1}, which is restated below with the numbers $\rho_t$ precisely specified, follows readily from Theorems~\ref{thm:2step} and \ref{thm:eaqecc}.

\begin{thm}[Restatement of Theorem~\ref{dt1}]  \label{thm:dt1_restate} 
Let $\mathcal{C} \subseteq \mathcal{R}^{2n}$ be an additive code,  i.e.,  $\mathcal{C}$ is a module over $\mathbb{Z}_{p^b}.$  From $\mathcal{C}$, we can construct an $((n,K, D; c))$ EAQECC over $\mathcal{R},$ where the number of pairs of maximally entangled qudits needed is $c = \frac12 \rank(\mathcal{C}/(\mathcal{C} \cap \mathcal{C}^{\perp_{{\chi}}}))$, the minimum distance is
\begin{equation*}
D~=~\left\{\begin{array}{ll}
d_{\mathrm{s}}(\mathcal{C}^{\perp_\chi}) & \text{~~if~~} \mathcal{C}^{\perp_\chi} \subseteq \mathcal{C} \\
d_{\mathrm{s}}(\mathcal{C}^{\perp_\chi}\setminus \mathcal{C}) & \text{~~otherwise} \, ,
\end{array}\right. 
\end{equation*}
and the dimension $K$ is bounded as $q^{n+c}/ (|\mathcal{C}| \, {p^{\sum_{t=1}^{b-1} (b-t)\rho_t}}) \leq  K \leq q^{n+c}/|\mathcal{C}|$, with $\rho_t:=\rank(\mathcal{C}/(\mathcal{C} \cap \mathcal{C}^{\perp_{\chi,t-1}}))-\rank(\mathcal{C}/(\mathcal{C} \cap \mathcal{C}^{\perp_{\chi,t}}))$. If either \emph{(a)}~$\mathcal{C}$ is free or \emph{(b)}~$\mathcal{C}/(\mathcal{C} \cap \mathcal{C}^{\perp_{\chi}})$ is a free module over $\mathbb{Z}_{p^b}$, then $K = q^{n+c}/|\mathcal{C}|$. In the case of \emph{(b)}, we additionally have $c=\frac12 \bigl[\rank(\mathcal{C}) - \rank(\mathcal{C} \cap \mathcal{C}^{\perp_\chi})\bigr]$.
\end{thm}
 
\begin{rem}  \label{rem:lee_klapp} We point out here that Lee and Klappenecker \cite[Proposition 2]{lee} provided a method to construct an EAQECC from a free linear code over a finite commutative Frobenius ring. The construction proposed in \cite[Proposition 2]{lee} relies crucially on \cite[Theorem 5]{lee}, which is analogous to our Theorem~\ref{thm:stdform}. However, we found a gap in the proof of \cite[Theorem 5]{lee} that could not readily be filled, namely, that replacing $\mathbf{w}_k$ with $\mathbf{w}'_{k-2} = e_{k,i} \mathbf{w}_k - \cdots $ may not result in a basis of $R^{2n}$, as $e_{k,i}$ may not be a unit in the ring $R$. \qed
\end{rem}

As we will now argue, our results in this section in fact offer a mathematically rigorous means of constructing EAQECCs from linear codes (not necessarily free) over finite commutative (not necessarily local) Frobenius rings.
By the Chinese Remainder Theorem, a finite commutative Frobenius ring $R$ can be written as direct product of finite commutative local rings. By \cite[Remark 1.3]{wood}, a finite commutative ring $R \cong R_1\times R_2 \times \cdots \times R_t$ is Frobenius if and only if each $R_i$ is Frobenius.
Thus, for a finite commutative Frobenius ring $R,$ we have $R \cong R_1\times R_2 \times \cdots \times R_t,$ where each $R_i$ is a local Frobenius ring.
 By \cite[Theorem 3.1]{dougherty19},  if $\chi_{R_i}$ is a generating character for $R_i$, then a character $\chi$ of $R$ defined for  $a=(a_1,a_2,\ldots,a_t) \in R$ by $$\chi(a)=\prod_{i=1}^{t}\chi_{R_i}(a_i) $$ is a generating character for $R$. 
 
It is straightforward to see that any linear code, $C$, of length $2n$ over $R$ is of the form $C_1\times C_2 \times \cdots \times C_t$, where each $C_i$ is a linear code of length $2n$ over $R_i$. Now, by Theorem~\ref{thm:stdform}, each $C_i$ has a generating set, $\mathcal{G}_i$, in standard form. 
 Define $\mathcal{G} \subset C$ to be the set consisting of all $(c_1,c_2,\ldots,c_t) \in C_1\times C_2 \times \cdots \times C_t$ such that for some $j \in [t]$, we have $c_j \in \mathcal{G}_j$ while $c_i = 0$ for all $i \ne j$.
 It is clear that $\mathcal{G}$ is a generating set of $C$. Moreover, it is an easy exercise to verify that $\mathcal{G}$ is in standard form with respect to $\chi({\langle \cdot | \cdot \rangle}_{\mathrm{s}})$, where $\chi$ acts on the symplectic inner product on $R^{2n}$ defined as follows: for $u = (u_1,u_2,\ldots,u_t)$ and $v = (v_1,v_2,\ldots,v_t) \in R_1^{2n} \times R_2^{2n} \times \cdots \times R_t^{2n} \cong R^{2n}$, 
 $$
{\langle u \mid v \rangle}_{\mathrm{s}} = \left({\langle u_1 \mid v_1 \rangle}_{\mathrm{s}(1)},{\langle u_2 \mid v_2 \rangle}_{\mathrm{s}(2)},\ldots,{\langle u_t \mid v_t \rangle}_{\mathrm{s}(t)}\right).
 $$
 Here, ${\langle u_j \mid v_j \rangle}_{\mathrm{s}(j)}$ is the usual symplectic inner product on $R_j^{2n}$ defined as in Definition~\ref{def:sip}. Once we have a standard-form generating set for an additive code $C$ over $R$, the entire machinery developed in this section for constructing EAQECCs over local Frobenius rings can be extended to obtain from $C$ an EAQECC over the ring $R$. We omit the formal details of such a construction, merely noting that we now have a valid method for constructing an EAQECC from any linear code over any finite commutative Frobenius ring. 
 
 We can in fact squeeze out a bit more from the decomposition of $R$ into local rings. If the characteristics of the local rings $R_i$ above are pairwise coprime, then any additive (not necessarily linear) code of length $2n$ over $R$ is of the form $C_1\times C_2 \times \cdots \times C_t,$ where each $C_i$ is an additive code of length $2n$ over $R_i$. Thus, the argument given above will again work to give us a construction of EAQECCs over $R$ from additive codes in $R^{2n}$. However, one point that we need to note here is that we do not, in general, have a notion of rank for an additive code $C$ over a finite commutative Frobenius ring $R$. As a stand-in for rank, we use the minimum number of generators of $C$ as an additive subgroup of $R^{2n}$. Then, the number of pairs of maximally entangled qudits needed to construct an EAQECC from an additive code $C$ over $R$ can be described in terms of the minimum number of generators of $C/(C \cap C^{\perp_\chi})$  as an additive code.

\bigskip

Our aim next is to determine the minimum number, $c_{\min}$, of pairs of maximally entangled qudits needed to construct an EAQECC from an additive code $\mathcal{C}$ over the Galois ring, $\text{GR}(p^b,m)$, of characteristic $p^b$ and cardinality $p^{mb}$. To bound $c_{\min}$ from below, we argue as follows: From Theorem~\ref{thm:stdform}, $\mathcal{C}$ has a generating set with $c = \frac12 \rank(\mathcal{C}/(\mathcal{C} \cap \mathcal{C}^{\perp_\chi}))$ hyperbolic pairs. If $\mathcal{C}'$ is a minimal $\chi$-self-orthogonal extension of $\mathcal{C}$, with entanglement degree $c_{\min}$, then, by Proposition~\ref{prop:ext}, there is a symplectic subset of $\mathcal{R}^{2c_{\min}}$ of cardinality $2c$. Hence, $2c$ can be bounded above by the size of the largest symplectic subset constructible on $\mathcal{R}^{2c_{\min}}$, which we obtain in terms of $c_{\min}$. This yields a lower bound on $c_{\min}$, as desired. A matching upper bound is obtained by constructing a symplectic subset as stipulated in Proposition~\ref{prop:ext}. This program is carried out for the case of integer rings $\mathbb{Z}_{p^a}$ in the next section, while the treatment for general Galois rings is presented in Section~\ref{sec5}.

\section{EAQECCs from Additive Codes over the Ring $\mathbb{Z}_{p^a}$} \label{sec4}

The map $\chi: \mathbb{Z}_{p^a} \to \mathbb{C}^*$ defined by $\chi(z) = \zeta^z$, with $\zeta=\exp\bigl(\frac{2 \pi i}{p^a}\bigr)$, is a generating character of the ring $\mathbb{Z}_{p^a}$. In particular, $\chi(z) = 1$ iff $z = 0$, and $\chi(z) \in \langle \zeta^{p^t} \rangle $ iff $z \equiv 0 \pmod{p^t}$. Hence, for an additive (linear) code $\mathcal{C} \subseteq \mathbb{Z}_{p^a}^{2n}$, we have $\mathcal{C}^{\perp_\chi}=\mathcal{C}^{\perp_{\mathrm{s}}}$ and $\mathcal{C}^{\perp_{\chi,t}}=\mathcal{C}^{\perp_{\mathrm{s},t}}$, where for $t = 0,1,\ldots,b$, we define
\begin{equation}
\label{def:Cperp_st} 
\mathcal{C}^{\perp_{\mathrm{s},t}} \ = \  \{(v',w') \in \mathbb{Z}_{p^a}^{2n} ~:~ \langle(v',w') \mid (v,w) \rangle_{\mathrm{s}} \equiv 0 \!\!\!\! \pmod{p^{a-t}} \text{ for all } (v,w) \in \mathcal{C}\}.
\end{equation}
Also, the definition of a symplectic subset (Definition~\ref{def:sympl}) of $\mathbb{Z}_{p^a}^{2n}$ reduces to the following:
\begin{defin} \label{def:sympl2} A subset $\{a_{11},a_{12},a_{21},a_{22},\ldots,a_{e1},a_{e2}\}$
 of $\mathbb{Z}_{p^a}^{2n}$ is said to be a symplectic subset if
  $\langle a_{i1} \mid a_{j1} \rangle_{\mathrm{s}} = \langle a_{i2} \mid a_{j2} \rangle_{\mathrm{s}} = \langle a_{i1} \mid a_{k 2} \rangle_{\mathrm{s}} = 0$ and  $\langle a_{i1} \mid a_{i2} \rangle_{\mathrm{s}} \neq 0$ for all $i ,j, k \in \{1,2,\ldots,e\}$ with $i \neq k.$
 \end{defin}

We will first derive an upper bound on the cardinality of a symplectic subset of $\mathbb{Z}_{p^a}^{2n}$, and then, using this bound, we will derive an explicit form of the minimum number of pairs of maximally entangled qudits needed for an EAQECC constructed from an additive (linear) code over $\mathbb{Z}_{p^a}.$   Our upper bound is in fact proved for the more general notion of a quasi-symplectic subset of ${\mathbb{Z}^{2n}_{p^a}}$, defined next. Here, for $z \in \mathbb{Z}_{p^a}$, we use the notation $\overline{z}$ to denote the residue $z \!\! \mod p$. 
\begin{defin} A subset $\{a_{11},a_{12},a_{21},a_{22},\ldots,a_{e1},a_{e2}\}$
 of ${\mathbb{Z}^{2n}_{p^a}}$ is said to be a quasi-symplectic subset if
 \begin{enumerate}
 \item[(a)] $\langle a_{i1}, a_{j1} \rangle_{\mathrm{s}}=\langle a_{i1}, a_{k 2} \rangle_{\mathrm{s}}=0$ for all $i ,j, k  \in \{1,2,\ldots,e\}$ with $i \neq k.$
 \item[(b)] there exists a subset $\mathcal{J}$ of $\{1,2,\ldots,e \}$ such that $\langle a_{i1}, a_{i2} \rangle_{\mathrm{s}}\neq0$ for all $i \not\in \mathcal{J}$, and $\{\overline{a}_{j1} \}_{j \in \mathcal{J}}$ is a linearly independent set over $\mathbb{Z}_p.$
 \end{enumerate}
\end{defin}

Note that a symplectic subset satisfies the definition of a quasi-symplectic subset by setting $\mathcal{J} = \emptyset$. In the following theorem, we provide an upper bound on the size of a quasi-symplectic subset of ${\mathbb{Z}^{2n}_{p^a}}.$
\begin{thm} \label{tquasi} If $\{b_{11},b_{12},b_{21},b_{22},\ldots,b_{e1},b_{e2}\}$ is a quasi-symplectic subset of ${\mathbb{Z}^{2n}_{p^a}},$ then $e \leq n.$ 
\end{thm}
\begin{proof} We will apply induction on $a \geq 1$ to prove this result. For $a=1,$ suppose that $\{b_{11},b_{12},b_{21},b_{22},\ldots,b_{e1},b_{e2}\}$ is a quasi-symplectic subset of ${\mathbb{Z}^{2n}_{p}}$. By definition of quasi-symplectic subset, we have $\langle b_{i1}, b_{j1} \rangle_{\mathrm{s}}=0$ and $\langle b_{i1}, b_{k 2} \rangle_{\mathrm{s}}=0$ for all $i ,j, k  \in \{1,2,\ldots,e\}$ with $i \neq k$, and there exists a subset $\mathcal{J}$ of $\{1,2,\ldots,e \}$ such that $\langle b_{t1}, b_{t2} \rangle_{\mathrm{s}}\neq0$ for all $ t \not\in \mathcal{J}$ and $\{b_{v1} \}_{v \in \mathcal{J}}$ is a linearly independent set over $\mathbb{Z}_p$. Here, we assert that $\{b_{11},b_{21},\ldots,b_{e1}\}$ is a linearly independent set over $\mathbb{Z}_p.$ To see this, suppose that $\sum\limits_{\ell=1}^{e}r_\ell b_{\ell 1} =0$ with $r_\ell \in \mathbb{Z}_p.$ From this, using the fact that $\langle b_{k1}, b_{j2}\rangle_{\mathrm{s}} =0$ and $\langle b_{j1}, b_{j2}\rangle_{\mathrm{s}} \neq 0$ for $k \in \{1,2,\ldots,e\}$, $j \in \mathcal{J}^\mathsf{c}$ with $k \neq j$, we get $r_t\langle b_{t1}, b_{t2}\rangle_{\mathrm{s}} =0$ for $t \in \mathcal{J}^\mathsf{c}.$ This implies that $r_t=0$ for each $t \in \mathcal{J}^\mathsf{c}.$ So we have $\sum\limits_{ \ell \in \mathcal{J}}^{}  r_\ell b_{\ell 1} =0.$ As $\{b_{\ell 1}\}_{\ell \in \mathcal{J}}$ is a linearly independent set over $\mathbb{Z}_p,$ we get $r_\ell =0$ for each $\ell \in \mathcal{J}.$ This proves the assertion. Now consider a linear code $\mathcal{C}$ over $\mathbb{Z}_p$ generated by  $\{b_{11},b_{21},\ldots,b_{e1}\}.$ Clearly, $\mathcal{C} \subseteq \mathcal{C}^{\perp_{\mathrm{s}}}.$ From this, using the fact that $\mathcal{C}^{\perp_\chi}=\mathcal{C}^{\perp_{\mathrm{s}}}$ and Lemma \ref{card}, we get $|{\mathbb{Z}_p}^{2n}|=|\mathcal{C}||\mathcal{C}^{\perp_{\mathrm{s}}}|\geq |{\mathbb{Z}_p}|^{2e}.$ This implies that $e \leq n,$ which completes the proof for $a=1.$ 
 
Now we assume that $a\geq 2.$ As $\{b_{11},b_{12},b_{21},b_{22},\ldots,b_{e1},b_{e2}\}$
is a quasi-symplectic subset of ${\mathbb{Z}^{2n}_{p^a}},$  we have $\langle b_{i1}, b_{j1} \rangle_{\mathrm{s}}=0$ and $\langle b_{i1}, b_{k 2} \rangle_{\mathrm{s}}=0$ for all $i ,j, k  \in \{1,2,\ldots,e\}$ with $i \neq k$ and there exists a subset $\mathcal{J}$ of $\{1,2,\ldots,e \}$ such that $\langle b_{t1}, b_{t2} \rangle_{\mathrm{s}}\neq0$ for all $ t \not\in \mathcal{J}$ and $\{\overline{b}_{v1} \}_{v \in \mathcal{J}}$ is a linearly independent set over $\mathbb{Z}_p.$ 
 Note that by adding indices from $\mathcal{J}^\mathsf{c}$ to $\mathcal{J}$ if necessary and interchanging $b_{j1}$ and $b_{j2}$ if needed, we can assume that for each $\ell \in \mathcal{J}^\mathsf{c},$
 both $\overline{b}_{j1}$ and $\overline{b}_{j2}$ are linearly dependent on the set $\{\overline{b}_{\ell 1}\}_{\ell \in \mathcal{J}}.$ That is, for each $\ell \in \mathcal{J}^\mathsf{c},$ we have

$$\overline{b}_{\ell 1}=\sum\limits_{ j \in \mathcal{J}}^{}  u_j \overline{b}_{j 1} \text{~and~} \overline{b}_{\ell 2}=\sum\limits_{ j \in \mathcal{J}}^{}  v_j \overline{b}_{j 1} \text{~for ~} u_j, v_j \in \mathbb{Z}_p.$$

This implies that for each $\ell \in \mathcal{J}^\mathsf{c},$ we have

\begin{equation}\label{e1}
b_{\ell 1}=\sum\limits_{ j \in \mathcal{J}}^{}  u_j b_{j 1}+pc_{\ell 1} \text{~and~} {b}_{\ell 2}=\sum\limits_{ j \in \mathcal{J}}^{}  v_j b_{j 1}+pc_{\ell 2} \text{~with~} u_j, v_j \in \mathbb{Z}_p \text{~and~} c_{\ell 1}, c_{\ell 2} \in \mathbb{Z}_{p^{a-1}}^{2n} \end{equation}

Recall that for $t, u \in \mathcal{J}$ and $h,i,j,k,\ell \in \mathcal{J}^\mathsf{c}$ with $i \neq k,$ 
we have $\langle b_{t1}, b_{h1}\rangle_{\mathrm{s}} =0,$  $\langle b_{t1}, b_{h2}\rangle_{\mathrm{s}} =0,$ $\langle b_{h1}, b_{u2}\rangle_{\mathrm{s}}=0,$ $\langle b_{i1}, b_{j1}\rangle_{\mathrm{s}}=0,$ $\langle b_{i1}, b_{k2}\rangle_{\mathrm{s}} =0$ and $ \langle b_{\ell1}, b_{\ell2}\rangle_{\mathrm{s}} \neq 0$ in $\mathbb{Z}_{p^{a}}.$

 From this and using \eqref{e1}, we observe that for $t, u \in \mathcal{J}$ and $h,i,j,k,\ell \in \mathcal{J}^\mathsf{c}$ with $i \neq k,$ we have
 \begin{eqnarray}\label{e2}
 &&    \langle b_{t1}, b_{h1}\rangle_{\mathrm{s}} =p\langle b_{t1},  c_{h1}\rangle_{\mathrm{s}} =0,    \langle b_{t1}, b_{h2}\rangle_{\mathrm{s}} =p\langle b_{t1},  c_{h2}\rangle_{\mathrm{s}} =0, \langle b_{h1}, b_{u2}\rangle_{\mathrm{s}} =p\langle c_{h1},  b_{u2}\rangle_{\mathrm{s}} = 0, 
 \nonumber\\ &&
\langle b_{i1}, b_{j1}\rangle_{\mathrm{s}} =p^2\langle c_{i1},  c_{j1}\rangle_{\mathrm{s}} = 0,   \langle b_{i1}, b_{k2}\rangle_{\mathrm{s}} =p^2\langle c_{i1},  c_{k2}\rangle_{\mathrm{s}} = 0
  \text{~and ~} 
 \langle b_{\ell1}, b_{\ell2}\rangle_{\mathrm{s}} =p^2\langle c_{\ell1},  c_{\ell2}\rangle_{\mathrm{s}} \neq 0  \text{~in~} \mathbb{Z}_{p^{a}}. ~ \end{eqnarray}

For $a=2,$ by using  \eqref{e2}, we see that $\mathcal{J}^\mathsf{c}=\phi.$ This implies that $\{\overline{a}_{11},\overline{a}_{21},\ldots,\overline{a}_{e1}\}$ is linearly independent over $\mathbb{Z}_{p}.$  Thus  a linear code $\mathcal{C}$ over  $\mathbb{Z}_p$ generated by  $\{b_{11},b_{21},\ldots,b_{e1}\}$ satisfies $\mathcal{C} \subseteq \mathcal{C}^{\perp_{\mathrm{s}}}.$ From this and using the fact that $\mathcal{C}^{\perp_\chi}=\mathcal{C}^{\perp_{\mathrm{s}}}$ and by Lemma \ref{card},   we get $|{\mathbb{Z}_p}^{2n}|=|\mathcal{C}||\mathcal{C}^{\perp_{\mathrm{s}}}|\geq |{\mathbb{Z}_p}|^{2e}.$ This implies that $e \leq n.$  Hence the result holds for $a=2.$   Now we assume that $k\geq 3$  be a fix integer and result holds for integer $a=k-2.$
  We will prove the result for $a=k.$ For that, by using \eqref{e2}, we get $\{ \{\tilde{b}_{j1},\tilde{b}_{j2}\}_{j \in \mathcal{J}},$   $ \{\tilde{c}_{\ell1},\tilde{c}_{\ell2}\}_{\ell \in \mathcal{J}^\mathsf{c}}\}$ is quasi symplectic subset of $\mathbb{Z}_{p^{k-2}}^n,$ where $\tilde{c}\equiv c  (\text{mod~}p^{k-2}).$ So by induction, we must have $e \leq n,$ which implies that the result holds for $a=k.$ This completes the proof of the theorem. 
\end{proof}

Since a symplectic subset of $\mathbb{Z}_{p^a}^{2n}$ is also a quasi-symplectic subset, we have the following corollary.

\begin{cor}\label{cor2} If $\{a_{11},a_{12},a_{21},a_{22},\ldots,a_{e1},a_{e2}\}$ is  a symplectic subset of ${\mathbb{Z}^{2n}_{p^a}},$ then $e \leq n.$ \end{cor}

From this, we obtain an explicit formula for the minimum entanglement degree of any $\chi$-self-orthogonal extension of a submodule $\mathcal{C} \subseteq \mathbb{Z}_{p^a}^{2n}.$

\begin{thm}\label{4thm} Any minimal $\chi$-self-orthogonal extension of a submodule $\mathcal{C} \subseteq \mathbb{Z}_{p^a}^{2n}$ has entanglement degree equal to $\frac12 \rank(\mathcal{C}/(\mathcal{C} \cap \mathcal{C}^{\perp_{\mathrm{s}}})).$
\end{thm}
\begin{proof} 
Let $c_{\min}$ be the entanglement degree of a minimal $\chi$-self-orthogonal extension, $\mathcal{C}'$, of $\mathcal{C}$. By Theorem~\ref{thm:2step}, we have $c_{\min} \le  \frac12 \rank(\mathcal{C}/(\mathcal{C} \cap \mathcal{C}^{\perp_{\mathrm{s}}})) $. On the other hand, from Theorem~\ref{thm:stdform}, $\mathcal{C}$ has a generating set with $c = \frac12 \rank(\mathcal{C}/(\mathcal{C} \cap \mathcal{C}^{\perp_{\mathrm{s}}}))$ hyperbolic pairs. Then, by Proposition~\ref{prop:ext}, there is a symplectic subset of $\mathbb{Z}_{p^a}^{2c_{\min}}$ of cardinality $2c$. Hence, by Corollary~\ref{cor2}, we have $c \le c_{\min}$, as desired.
\end{proof}

Observe that Theorem~\ref{dt2}, which we restate below with the numbers $\rho_t$ exactly specified,  follows by combining the results of Theorems~\ref{thm:2step}, \ref{thm:eaqecc} and \ref{4thm}, using the fact that $\mathcal{C}^{\perp_\chi}=\mathcal{C}^{\perp_{\mathrm{s}}}$ and  $\mathcal{C}^{\perp_{\chi,t}}=\mathcal{C}^{\perp_{\mathrm{s},t}}$, as defined in \eqref{def:Cperp_st}.

\begin{thm}[Restatement of Theorem~\ref{dt2}] Let  $\mathcal{C} \subseteq \mathbb{Z}_{p^a}^{2n}$ be a submodule. From $\mathcal{C}$, we can construct an $((n,K, D; c))$ EAQECC over $\Z_{p^a}$, where the minimum number, $c$, of pairs of maximally entangled qudits needed for the construction is equal to $\frac12 \rank(\mathcal{C}/(\mathcal{C} \cap \mathcal{C}^{\perp_{\mathrm{s}}}))$,  the minimum distance is 
\begin{equation*}
D~=~\left\{\begin{array}{ll}
d_{\mathrm{s}}(\mathcal{C}^{\perp_{\mathrm{s}}}) & \text{~~if~~} \mathcal{C}^{\perp_{\mathrm{s}}} \subseteq \mathcal{C} \\  
d_{\mathrm{s}}(\mathcal{C}^{\perp_{\mathrm{s}}}\setminus \mathcal{C} )  & \text{~~otherwise} \, , \end{array}\right. 
\end{equation*}
and the dimension $K$ is bounded as $p^{a(n+c)}/ (|\mathcal{C}| \, {p^{\sum_{t=1}^{a-1} (a-t)\rho_t}}) \leq  K \leq p^{a(n+c)}/|\mathcal{C}|$, with $\rho_t:=\rank(\mathcal{C}/(\mathcal{C} \cap \mathcal{C}^{\perp_{\mathrm{s},t-1}}))-\rank(\mathcal{C}/(\mathcal{C} \cap \mathcal{C}^{\perp_{\mathrm{s},t}}))$.  If either \emph{(a)}~$\mathcal{C}$ is free or \emph{(b)}~$\mathcal{C}/(\mathcal{C} \cap \mathcal{C}^{\perp_{\mathrm{s}}})$ is a free module over $\mathbb{Z}_{p^b}$, then $K = p^{a(n+c)}/|\mathcal{C}|$. In the case of \emph{(b)}, we additionally have $c=\frac12 \bigl[\rank(\mathcal{C}) - \rank(\mathcal{C} \cap \mathcal{C}^{\perp_{\mathrm{s}}})\bigr]$.
 \end{thm}

\section{EAQECCs from Codes over Galois Rings}\label{sec5}
In this section, we will first derive an upper bound on the cardinality of a symplectic subset of $\text{GR}(p^b,m)^{2n}$ and then using this bound, we will derive an explicit formula for the minimum number of pairs of maximally entangled qudits needed for an EAQECC constructed from an additive code over $\text{GR}(p^b,m)$.  To do this, we start by reviewing some basic properties of Galois rings.

A finite commutative ring with unity is called a Galois ring if all its non-units (or equivalently, all its zero divisors  including $0)$ form an ideal generated by some prime number. Any finite field of characteristic $p$ is a Galois ring as all the non-units form a zero ideal which generated by $p$. Another example of a Galois ring is the ring $\mathbb{Z}_{p^a}$ of integers modulo a prime power $p^a$, in which all the non-units form an ideal generated by $p$.

Let $R$ be a Galois ring, in which the set of all zero-divisors of including $0$ form an ideal generated by a prime number $p$. Then the ideal generated by $p$ is the only maximal ideal of ${R}$ and the characteristic of ${R}$ is $p^b$ for some positive integer $b$. The residue field of ${R}$ is given by ${R} / \langle p \rangle \simeq  \mathbb{F}_{p^m}$ for some positive integer $m$ and  $|{R}|=p^{mb}.$ Moreover, the ring $R$ is isomorphic to the quotient ring $\mathbb{Z}_{p^b}[x]/ \langle h(x) \rangle,$ where  $h(x)\in \mathbb{Z}_{p^b}[x]$ is a monic polynomial of degree $m$ such that its reduction $(h(x) \text{ mod } p) \in \mathbb{F}_{p}[x]$ is irreducible and primitive over $\mathbb{F}_{p}.$ Thus any two Galois rings of  the same characteristic and the same cardinality are isomorphic (see \cite[Chp. 14]{wan}). 

Through out this section, let $\text{GR}(p^b,m)$ denote the Galois ring of characteristic $p^b$ and cardinality $p^{mb}.$  Further, let $h(x)\in \mathbb{Z}_{p^b}[x]$ be a monic polynomial of degree $m$ such that its reduction $(h(x) \text{ mod } p) \in \mathbb{F}_{p}[x]$ is irreducible and primitive over $\mathbb{F}_{p}$ and $h(x)$ divides $x^{p^m-1}-1$ in $\mathbb{Z}_{p^b}[x].$  Then we have  $\text{GR}(p^b,m) \cong \mathbb{Z}_{p^b}[x]/ \langle h(x) \rangle.$ If $\theta=x+\langle h(x) \rangle,$ then clearly $h(\theta)=0$ and every element of $\text{GR}(p^b,m) \cong \mathbb{Z}_{p^b}[x]/ \langle h(x) \rangle$ has a unique representation as $$r=r_0+r_1\theta+\cdots+r_{m-1}\theta^{m-1} \text{~with~} r_0,r_1,\ldots,r_{m-1} \in \mathbb{Z}_{p^b}.$$
This implies that $\text{GR}(p^b,m)$ is a free module of rank $m$ over  $\mathbb{Z}_{p^b}$ with $\{1,\theta,\theta^2,\ldots,\theta^{m-1}\}$ as a basis.
 Moreover,   there exists an element $\beta$  in $\text{GR}(p^b,m)$ having the multiplicative order as  $p^{m}-1,$ which is a root of $h(x).$ The set $\mathcal{T}= \{ 0,1, \beta,\ldots,\beta^{p^{m}-2} \}$ is called the Teichm\"{u}ller set of $\text{GR}(p^b,m)$ and each element $z \in \text{GR}(p^b,m)$ can be uniquely expressed as $z=z_0+ z_1 p+ z_2 p^2+\cdots + z_{b-1} p^{b-1},$ where $z_0,z_1,\ldots,z_{b-1} \in \mathcal{T}$ (see \cite[Chp. 14]{wan}).

The generalized Frobenius map $f$ on the Galois ring $\text{GR}(p^b,m)$ is defined by $$f(z)=z^p_0+ z^p_1 p+ z^p_2 p^2+\cdots + z^p_{b-1} p^{b-1}$$ for each $z=z_0+ z_1 p+ z_2 p^2+\cdots + z_{b-1} p^{b-1} \in \text{GR}(p^b,m),$ where $z_0,z_1,\ldots,z_{b-1} \in \mathcal{T}.$  

For $b = 1,$ the map $f$ reduces to the usual Frobenius automorphism on $\mathbb{F}_{p^m}$ defined by $f(z)=z^p$ for all $z \in \mathbb{F}_{p^m}.$ 
The generalized trace map ${\Tr}$ from $\text{GR}(p^b,m)$ to $\mathbb{Z}_{p^b}$ is defined by 
$${\Tr}(z)=z+f(z)+f^2(z)+\cdots+f^{m-1}(z) \text{ for each } z \in \text{GR}(p^b,m).$$
For $b=1$, we have $\text{GR}(p,m)\simeq  \mathbb{F}_{p^m}$; thus the generalized trace map ${\Tr}$  reduces to the usual trace map  ${\tr}: \mathbb{F}_{p^m}  \to \mathbb{F}_{p}$  defined by ${\tr}(z) =z+z^p+z^{p^2}+\cdots+z^{p^{m-1}}$ for all $z \in \mathbb{F}_{p^m}.$ 
\begin{prop}\cite{sison}  The generalized trace map ${\Tr}:\text{GR}(p^b,m) \to  \mathbb{Z}_{p^b}$ has the following properties.
\begin{itemize}
\item ${\Tr}(u+v)={\Tr}(u)+{\Tr}(v)$ for each $u,v \in \text{GR}(p^b,m).$
\item  ${\Tr}(e u)=e{\Tr}(u)$ for each $u \in \text{GR}(p^b,m)$ and $e \in \mathbb{Z}_{p^b}.$
\item  ${\Tr}$ is surjective and $\text{GR}(p^b,m)/ \ker( {\Tr} )=\mathbb{Z}_{p^b}.$
\end{itemize}
\end{prop}

\begin{prop}\label{prop:charGR}\cite{shuqin} The map $\chi: \text{GR}(p^b,m) \to \mathbb{C}^*$ defined by
$\chi(r)=\zeta^{{\Tr}(r)}$, with $\zeta = \exp\bigl(\frac{2\pi i}{p^b}\bigr)$, is a generating character of  $\text{GR}(p^b,m).$
\end{prop}

\begin{defin} 
\label{def:trace_duals}
For an additive code $\mathcal{C}$ of $\text{GR}(p^b,m)^{2n},$  we define
\begin{itemize}
\item the trace-symplectic dual of $\mathcal{C}$ as
$$\mathcal{C}^{\perp_{\Tr}} = \{v \in \text{GR}(p^b,m)^{2n} : \Tr( \langle v | c\rangle_{\mathrm{s}}) = 0 \text{~for all~} c \in \mathcal{C}\},$$
\item and for $t = 0,1,\ldots,b$, the trace-symplectic $t$-dual of $\mathcal{C}$ as 
$$\mathcal{C}^{\perp_{\Tr,t}} = \{v \in \text{GR}(p^b,m)^{2n} ~:~ \Tr( \langle v | c\rangle_{\mathrm{s}}) \equiv 0 \!\!\! \pmod{p^{b-t}} \text{~for all~} c \in \mathcal{C}\}.$$
\end{itemize}
\end{defin}
For $b=1,$ we have $\text{GR}(p,m)\simeq  \mathbb{F}_{p^m},$ thus $\mathcal{C}^{\perp_{\Tr}}$ is equal to 
$\mathcal{C}^{\perp_{\tr}} =\{v \in \mathbb{F}_{p^m}^{2n} ~:~ \tr(\langle v | c\rangle_{\mathrm{s}}) = 0 \text{ for all } c \in \mathcal{C}\}$.

\begin{prop}\label{prop:dualGR}  For an additive code $\mathcal{C}$ of $\text{GR}(p^b,m)^{2n},$ we have $\mathcal{C}^{\perp_\chi}=\mathcal{C}^{\perp_{\Tr}}$, and $\mathcal{C}^{\perp_{\chi,t}}=\mathcal{C}^{\perp_{\Tr,t}}$ for $t = 0,1,\ldots,b$.
\end{prop}
\begin{proof}Using  Proposition \ref{prop:charGR}, we observe that $\chi(r)=1$ if and only if ${\Tr}(r)=0,$ and $\chi(r) \in \langle \zeta^{p^{b-t}} \rangle $ if and only if ${\Tr}(r) \equiv 0 \pmod{p^{b-t}}$.  This implies that $\mathcal{C}^{\perp_\chi}=\mathcal{C}^{\perp_{\Tr}}$, and $\mathcal{C}^{\perp_{\chi,t}}=\mathcal{C}^{\perp_{\Tr,t}}.$
\end{proof} 
The definition of a symplectic subset (Definition~\ref{def:sympl}) of $\text{GR}(p^b,m)^{2n}$  reduces to the following:
\begin{defin} \label{def:sympl2} A subset $\{a_{11},a_{12},a_{21},a_{22},\ldots,a_{e1},a_{e2}\}$
 of $\text{GR}(p^b,m)^{2n},$ is said to be a symplectic subset if
  $\Tr( \langle a_{i1} \mid a_{j1} \rangle_{\mathrm{s}}) =\Tr(  \langle a_{i2} \mid a_{j2} \rangle_{\mathrm{s}} )= \Tr( \langle a_{i1} \mid a_{k 2} \rangle_{\mathrm{s}} )= 0$ and  $\Tr( \langle a_{i1} \mid a_{i2} \rangle_{\mathrm{s}} )\neq 0$ for all $i ,j, k \in \{1,2,\ldots,e\}$ with $i \neq k.$
 \end{defin}

\begin{defin}\cite{sison} Two bases $\{\beta_1,\beta_2,\ldots,\beta_m\}$ and  $\{\gamma_1,\gamma_2,\ldots,\gamma_m\}$ of $\text{GR}(p^b,m)$ as a free module over $\mathbb{Z}_{p^b}$ are said to be dual if ${\Tr}(\beta_i\gamma_j)=\delta_{ij},$ where $\delta_{ij}$ denotes the Kronecker delta.\end{defin}
\begin{thm}\cite{sison} Every basis $\{\beta_1,\beta_2,\ldots,\beta_m\}$  of $\text{GR}(p^b,m)$ as a free module over $\mathbb{Z}_{p^b}$ has a unique dual basis.
\end{thm}

 Next we provide an upper bound on the size of a symplectic subset of $\text{GR}(p^b,m)^{2n}.$
\begin{prop}\label{prop:symGR}  If $\{a_{11},a_{12},a_{21},a_{22},\ldots,a_{e1},a_{e2}\} \subseteq \text{GR}(p^b,m)^{2n}$ is a symplectic subset of $\text{GR}(p^b,m)^{2n},$ then $e \leq nm.$ Moreover, for any given $z_1,z_2,\ldots,z_e \in \mathbb{Z}_{p^b}$ with $1 \leq e \leq nm,$ there exists a symplectic subset $\{a_{11},a_{12},a_{21},a_{22},$ $\ldots,a_{e1},a_{e2}\} \subseteq \text{GR}(p^b,m)^{2n}$ such that $ \chi(\langle a_{j1} |  a_{j2}\rangle_{\mathrm{s}})= e^{i \frac{2\pi }{p^b} z_j  }$ for $1\leq j \leq e.$
\end{prop}
\begin{proof} To prove this result, we first note that $\chi(r)=e^{i \frac{2\pi}{p^b}  {\Tr}(r)}$ for each $r \in \text{GR}(p^b,m),$ which implies that $\chi(r)=1$ if and only if ${\Tr}(r)=0.$ Let $\gamma_1,\gamma_2,\ldots,\gamma_m$ be a  basis of $\text{GR}(p^b,m)$ as a free module over $\mathbb{Z}_{p^b}$ and $\beta_1,\beta_2,\ldots,\beta_m$ be the dual basis of $\gamma_1,\gamma_2,\ldots,\gamma_m.$ Any element $r \in \text{GR}(p^b,m)$ can be uniquely written as $$r=x_1\gamma_1+x_2\gamma_2+\cdots+x_{m}\gamma_{m}=y_1\beta_1+y_2\beta_2+\cdots+y_{m}\beta_{m} \text{~with~} x_i, y_j \in \mathbb{Z}_{p^b}.$$

 Now define two maps $\phi_\gamma: \text{GR}(p^b,m) \to \mathbb{Z}_{p^b}^m$ and
$\phi_\beta: \text{GR}(p^b,m) \to \mathbb{Z}_{p^b}^m$ as follows:
$$r=x_1\gamma_1+x_2\gamma_2+\cdots+x_{m}\gamma_{m}   \stackrel{\phi_\gamma}{\longmapsto}  (x_1,x_2,\ldots,x_m)$$ 
and $$r=y_1\beta_1+y_2\beta_2+\cdots+y_{m}\beta_{m}\stackrel{\phi_\beta}{\longmapsto}  (y_1,y_2,\ldots,y_m).$$ 
Define a map $\phi: \text{GR}(p^b,m) ^{2n} \to \mathbb{Z}_{p^b}^{2nm}$ as follows:
$$(d_1,d_2,\ldots,d_n,e_1,e_2,\ldots,e_n) \stackrel{\phi}{\longmapsto} (\phi_\gamma(d_1),\phi_\gamma(d_2),\ldots,\phi_\gamma(d_n),\phi_\beta(e_1),\phi_\beta(e_2),\ldots,\phi_\beta(e_n)).$$
It can be checked that $${\Tr}\bigl({\langle (d_1,d_2,\ldots,d_n,e_1,e_2,\ldots,e_n) \mid (d'_1,d'_2,\ldots,d'_n,e'_1,e'_2,\ldots,e'_n) \rangle}_{\mathrm{s}}\bigr)~~~~~~~~~~~~~~~~~~~~~~~~~~~~~~~~~~~~~~~~~~~~~~~~~~~~$$ 
$$={\langle  \phi((d_1,d_2,\ldots,d_n,e_1,e_2,\ldots,e_n)) \mid \phi((d'_1,d'_2,\ldots,d'_n,e'_1,e'_2,\ldots,e'_n)) \rangle}_{\mathrm{s}}~~~~~~~~~~~~~~~~~~~~~ \, $$
$$={\langle (\phi_\gamma(d_1),\ldots,\phi_\gamma(d_n),\phi_\beta(e_1),\ldots,\phi_\beta(e_n)) \mid (\phi_\gamma(d'_1),\ldots,\phi_\gamma(d'_n),\phi_\beta(e'_1),\ldots,\phi_\beta(e'_n)) \rangle}_{\mathrm{s}} \, . $$

From this, we see that if $\{a_{11},a_{12},a_{21},a_{22},\ldots,a_{e1},a_{e2}\} \subseteq \text{GR}(p^b,m)^{2n}$ is a symplectic subset of $\text{GR}(p^b,m)^{2n},$ then $\{\phi(a_{11}),\phi(a_{12}),\phi(a_{21}),\phi(a_{22}),\ldots,\phi(a_{e1}),\phi(a_{e2})\} \subseteq \mathbb{Z}_{p^b}^{2mn}$ is a symplectic subset of $\mathbb{Z}_{p^b}^{2mn}$. Hence, by Corollary \ref{cor2}, we have $e \leq nm.$ 

To prove the next part,  let $t $ be a non-negative integer such that $tm < e \leq (t+1)m$. Then, for $j = km+ \ell$, with $0 \le k \le t-1$ and $1 \le \ell \le m$, set 
$$
a_{j1} =  (\underbrace{ 0,0,\ldots,0}_{n+k},\beta_{\ell},\underbrace{ 0,0,\ldots,0}_{n-k-1})\ \text{ and } \ a_{j2} =(\underbrace{ 0,0,\ldots,0}_{k},z_{km+\ell}\gamma_{\ell},\underbrace{ 0,0,\ldots,0}_{2n-k-1})  , 
$$
and for $j = tm + \ell$, with $1 \le \ell \le e-tm$, set
$$
a_{j1} =(\underbrace{ 0,0,\ldots,0}_{n+t},\beta_{\ell},\underbrace{ 0,0,\ldots,0}_{n-t-1})  \ \text{ and } \ a_{j2} = (\underbrace{ 0,0,\ldots,0}_{t},z_{tm+\ell}\gamma_{\ell},\underbrace{ 0,0,\ldots,0}_{2n-t-1})  .
$$ 
Then, $\{a_{11},a_{12},a_{21},a_{22},$ $\ldots,a_{e1},a_{e2}\}$ is a symplectic subset of $\text{GR}(p^b,m)^{2n}$ with the required property.
\end{proof}

We can now obtain an explicit formula for the minimum entanglement degree of any $\chi$-self-orthogonal extension of an additive code $\mathcal{C} \subseteq \text{GR}(p^b,m)^{2n}.$ 

\begin{thm}\label{EAQZ} Any minimal $\chi$-self-orthogonal extension of an additive code $\mathcal{C} \subseteq \text{GR}(p^b,m)^{2n}$ has entanglement degree equal to $\left\lceil { \frac{1}{2m} \rank(\mathcal{C}/(\mathcal{C} \cap \mathcal{C}^{\perp_{\Tr}}))  }\right\rceil.$ 
\end{thm}
\begin{proof} 
Let $c := \frac12 \rank(\mathcal{C}/(\mathcal{C} \cap \mathcal{C}^{\perp_{{\chi}}}))  $, which, by Proposition~\ref{prop:dualGR}, equals 
$\frac{1}{2}  \rank(\mathcal{C}/(\mathcal{C} \cap \mathcal{C}^{\perp_{\Tr}}))  $. Let $c_{\min}$ be the entanglement degree of a minimal $\chi$-self-orthogonal extension, $\mathcal{C}'$, of $\mathcal{C}$. We want to show that $c_{\min} = \lceil \frac1m c\rceil$.

From Theorem~\ref{thm:stdform}, $\mathcal{C}$ has a generating set with $c$ hyperbolic pairs. Then, by Proposition~\ref{prop:ext}, there is a symplectic subset of $\mathbb{Z}_{p^a}^{2c_{\min}}$ of cardinality $2c$, so that by Proposition~\ref{prop:symGR}, we have $c \le mc_{\min}$. Thus, $c_{\min} \ge \lceil \frac1m c\rceil$.

For the inequality in the opposite direction, set $\nu := \lceil \frac1m c\rceil$, so that $c \le m\nu$. Let $(v_j,w_j)$ and $(x_j,y_j)$, $j = 1,2,\ldots,c$, be the hyperbolic pairs in the generating set of $\mathcal{C}$, and let $z_j = {\Tr}(\langle (v_j,w_j) \mid (x_j,y_j)\rangle_{\mathrm{s}})$, $j=1,2,\ldots,c$. By Proposition~\ref{prop:symGR}, there exists a symplectic subset $\{a_{11},a_{12},a_{21},a_{22},$ $\ldots,a_{c1},a_{c2}\} \subseteq \text{GR}(p^b,m)^{2\nu}$ such that $ \chi(\langle a_{j1} |  a_{j2}\rangle_{\mathrm{s}}) = e^{i \frac{2\pi }{p^b} z_j  }$ for $1\leq j \leq c$. Since $e^{i \frac{2\pi }{p^b} z_j} = \chi(\langle (v_j,w_j) \mid (x_j,y_j)\rangle_{\mathrm{s}})$, we have a symplectic subset satisfying the condition of Proposition~\ref{prop:ext}. Hence, there exists a $\chi$-self-orthogonal extension $\mathcal{C}'$ with entanglement degree $\nu$. Therefore, $c_{\min} \le \nu = \lceil \frac1m c\rceil$, completing the proof of the theorem.
\end{proof}

Theorem~\ref{dt3} is an immediate consequence of Theorems~\ref{thm:2step}, \ref{thm:eaqecc},  and \ref{EAQZ}, again using the fact (Proposition \ref{prop:dualGR}) that $\mathcal{C}^{\perp_\chi} = \mathcal{C}^{\perp_{\Tr}}$ and $\mathcal{C}^{\perp_{\chi,t}}=\mathcal{C}^{\perp_{\Tr,t}}$. The theorem is restated below with the numbers $\rho_t$ made explicit.

\begin{thm}[Restatement of Theorem~\ref{dt3}] 
\label{thm:dt3_restatement} Let $\mathcal{C} \subseteq \text{GR}(p^b,m)^{2n}$ be an additive code over the Galois ring $\text{GR}(p^b,m)$.   From $\mathcal{C}$, we can construct an  $((n,K,D;c))$ EAQECC over $\text{GR}(p^b,m),$ where the minimum number, $c$, of pairs of maximally entangled qudits needed for the construction is equal to  $\left\lceil \frac1{2m}\rank(\mathcal{C}/(\mathcal{C} \cap \mathcal{C}^{\perp_{\Tr}}))  \right\rceil$, the minimum distance is 
\begin{equation*}
D~=~\left\{\begin{array}{ll}
d_{\mathrm{s}}(\mathcal{C}^{\perp_{\Tr}}) & \text{~~if~~} \mathcal{C}^{\perp_{\Tr}} \subseteq \mathcal{C} \\
d_{\mathrm{s}}(\mathcal{C}^{\perp_{\Tr}}\setminus \mathcal{C} )  & \text{~~otherwise}\,,
\end{array}\right.
\end{equation*}
and the dimension $K$ is bounded as $p^{bm(n+c)}/ (|\mathcal{C}| \, {p^{\sum_{t=1}^{b-1} (b-t)\rho_t}}) \leq  K \leq p^{bm(n+c)}/|\mathcal{C}|$, with $\rho_t := \rank(\mathcal{C}/(\mathcal{C} \cap \mathcal{C}^{\perp_{\Tr,t-1}}))-\rank(\mathcal{C}/(\mathcal{C} \cap \mathcal{C}^{\perp_{\Tr,t}}))$. If either \emph{(a)}~$\mathcal{C}$ is free or \emph{(b)}~$\mathcal{C}/(\mathcal{C} \cap \mathcal{C}^{\perp_{\Tr}})$ is a free module over $\mathbb{Z}_{p^b}$, then $K = p^{bm(n+c)}/|\mathcal{C}|$. In the case of \emph{(b)}, we additionally have $c=\left\lceil \frac1{2m}[\rank(\mathcal{C}) - \rank(\mathcal{C} \cap \mathcal{C}^{\perp_{\Tr}})]\right\rceil$.
\end{thm}

Corollary~\ref{cor:Fpm} follows from the observations that, for finite fields $\F_{p^m}$, the generalized trace map $\Tr$ reduces to the usual trace map $\tr: \F_{p^m} \to \F_p$, and $\mathcal{C},$ $\mathcal{C} \cap \mathcal{C}^{\perp_{\tr}}$ and $\mathcal{C}/(\mathcal{C} \cap \mathcal{C}^{\perp_{\tr}})$ are vector spaces over $\F_p$, so that, being free modules, their rank equals their dimension over $\F_p$. 

As an application of Theorem~\ref{thm:dt3_restatement}, we can have a CSS-like construction of EAQECCs starting from two additive codes $\mathcal{C}_1$ and $\mathcal{C}_2$ of length $n$ over $\text{GR}(p^b,m)$. This is done by applying Theorem~\ref{thm:dt3_restatement} to the code $\mathcal{C} = \mathcal{C}_1 \oplus \mathcal{C}_2$. We provide the details of this construction in Appendix~\ref{app:CSS}.

Finally, instead of starting from an additive code of length $2n$ over $\text{GR}(p^b,m)$, the construction in Theorem~\ref{thm:dt3_restatement} can be equivalently, and somewhat more conveniently, described as starting from an additive code of length $n$ over $\text{GR}(p^b,2m)$. This requires the extension to Galois rings of the machinery of the trace-alternating form that exists over finite fields \cite{calder}, \cite[Section~IV-B]{ketkar}. We describe this extension in Appendix~\ref{app:trace-alt-form}.

\section{EAQECCs from Lengthened Codes}\label{sec:lengthening}

The theory developed in the preceding sections shows that an EAQECC can be constructed from an additive code $\mathcal{C}$ over the Galois ring $\text{GR}(p^b,m)$, starting from a standard-form generating set $\mathcal{G}$ for $\mathcal{C}$. Any such generating set contains $e = \frac12 \rank\bigl(\mathcal{C} / (\mathcal{C} \cap \mathcal{C}^{\perp_{\Tr}})\bigr)$ hyperbolic pairs, and the (minimum) number of maximally entangled qudit pairs needed for the EAQECC construction is $c = \lceil\frac{e}{m}\rceil$. In this section, we will show that if we are allowed to lengthen the code $\mathcal{C}$ by inserting additional coordinates, then by carefully selecting what goes into the extra coordinates, we can obtain standard-form generating sets for the longer codes that differ from $\mathcal{G}$ in the number of isotropic generators and hyperbolic pairs. The EAQECCs obtained from the longer codes will then have parameters that are different from those of the EAQECC obtained from $\mathcal{C}$. Thus, the mechanism of lengthening an additive code gives us a means of influencing the parameters and error-handling capabilities of the resulting EAQECCs.

We illustrate this principle by describing two methods for lengthening an additive code. In one method, we show that by inserting extra coordinates, we can reduce the number of hyperbolic pairs in a standard-form generating set by converting some of the hyperbolic pairs into isotropic generators. This brings about a reduction in the number of pairs of maximally entangled qudits needed in the entanglement-assisted quantum code, but this is usually at the expense of a loss in minimum distance; the dimension of the EAQECC remains unaffected. In the other method of lengthening, we again apply the idea of using the extra coordinates to convert hyperbolic pairs into isotropic generators. But this idea is applied not to a standard-form generating set of the code $\mathcal{C}$, but instead to a standard-form generating set of the code $\mathcal{C}^{\perp_{\Tr}}$. From the perspective of the code $\mathcal{C}$ and its generating set $\mathcal{G}$ (in standard form), this results in a standard-form generating set of the lengthened code with a larger number of isotropic generators than in $\mathcal{G}$, while keeping the number of hyperbolic pairs the same. The EAQECC obtained from the lengthened code then requires the same number of maximally entangled qudit pairs as that obtained from $\mathcal{C}$, but the minimum distance may now increase. Any increase in minimum distance, however, is accompanied by a reduction in dimension. 

The techniques for lengthening an additive code described in this section are inspired by the propagation rules for EAQECCs derived from linear codes over finite fields, proposed by Luo et al.\ \cite[Theorems 16 and 18]{luo1}. As acknowledged by Luo et al., the idea behind these techniques, which we have expressed in terms of converting hyperbolic pairs into isotropic generators, can be traced back to the work of Lison\v{e}k and Singh \cite{liso}.

We introduce a definition that will be useful for our purposes. A code $\mathcal{C}' \subseteq \text{GR}(p^b,m)^{2(n+1)}$ is defined to be a \emph{$1$-step extension} of a code $\mathcal{C} \subseteq \text{GR}(p^b,m)^{2n}$ if $\mathcal{C}$ can be obtained by puncturing $\mathcal{C}'$ at the coordinates at $n+1$ and $2(n+1)$.

\subsection{Reducing the Number of Maximally Entangled Qudit Pairs}\label{subsec:less}

In this section, we describe how additional coordinates can be used to reduce the number of hyperbolic pairs by converting some of them to isotropic generators.  The following proposition gives a method of going from an additive code $\mathcal{C}$ to a $\chi$-self-orthogonal extension of it via a sequence of $1$-step extensions, with the number of hyperbolic pairs being reduced at each step.

\begin{prop} \label{prop:extless}
 Let $\mathcal{C} \subseteq \text{GR}(p^b,m)^{2n}$ be an additive code with a minimal generating set in standard form, in which there are exactly $e$ hyperbolic pairs and $d$ isotropic generators. Set $c := \lceil \frac{e}{m}\rceil$. Then, there exists a sequence of additive codes $\mathcal{C}^{(\ell)} \subseteq \text{GR}(p^b,m)^{2(n+\ell)}$, $\ell = 0,1,\ldots,c$, such that $\mathcal{C}^{(0)} = \mathcal{C}$, and for $\ell=0,1,\ldots,c-1$,
 \begin{itemize}
  \item[(i)] $\mathcal{C}^{(\ell)}$ has a minimal generating set in standard form with exactly $e - \ell m$ hyperbolic pairs and $d+2\ell m$ isotropic generators;
  \item[(ii)] $\mathcal{C}^{(\ell+1)}$ is a $1$-step extension of $\mathcal{C}^{(\ell)}$, with $d_s\bigl({\mathcal{C}^{(\ell+1)}}^{\perp_{\Tr}} \setminus \mathcal{C}^{(\ell+1)}\bigr) \le d_s\bigl({\mathcal{C}^{(\ell)}}^{\perp_{\Tr}} \setminus \mathcal{C}^{(\ell)}\bigr)$;
  \item[(iii)] $\mathcal{C}^{(c)}$ is a minimal $\chi$-self-orthogonal extension, obtainable via the Two-Step Construction, of $\mathcal{C}^{(\ell)}$.
  \end{itemize}
 \end{prop}

\begin{proof} Let  $$\mathcal{G}=\{(v_1,w_1),(v_2,w_2),\ldots, (v_{e+d},w_{e+d}),(x_1,y_1),(x_2,y_2),\ldots, (x_{e},y_{e})\}$$ be a minimal generating set of $\mathcal{C}$ as a $\mathbb{Z}_{p^b}$-module such that, for $i = 1,2\ldots,e$, the generators $(v_i,w_i)$ and $(x_i,y_i)$ form hyperbolic pairs, and the generators $(v_i,w_i)$, $i = e+1,\ldots,e+d$, are isotropic. 
For $1 \leq i \leq e,$ $\chi(w_i \cdot x_i-v_i \cdot y_i)$ is a $p^b$-th root of unity and $(v_i,w_i)$ and $(x_i,y_i)$ form hyperbolic pairs, so let $\chi(w_i \cdot x_i-v_i \cdot y_i)=\zeta^{z_{i}},$ where $\zeta = \exp(\frac{2\pi i}{p^b})$ and $0< z_{i} < p^b.$ Let $\gamma_1,\gamma_2,\ldots,\gamma_m$ be a  basis of $\text{GR}(p^b,m)$ as a free module over $\mathbb{Z}_{p^b}$, and let $\beta_1,\beta_2,\ldots,\beta_m$ be the dual basis of $\gamma_1,\gamma_2,\ldots,\gamma_m.$  

Set $\mathcal{C}^{(0)} = \mathcal{C}$. In the following, we use the notation $0^{\mu}$, for a non-negative integer $\mu$, to denote a string of $\mu$ $0$s. For $\ell = 1,2,\ldots,c-1$, let $\mathcal{C}^{(\ell)} \subseteq \text{GR}(p^b,m)^{2(n+\ell)}$ be the additive code generated by $\mathcal{G}^{(\ell)} = \{h^{(\ell)}_{11},h^{(\ell)}_{12},\ldots, h^{(\ell)}_{e1},h^{(\ell)}_{e2},g^{(\ell)}_1,\ldots,g^{(\ell)}_d\}$, where
\begin{equation}
h^{(\ell)}_{i1} = (v_i,0^{\ell},w_i,0^k,-\beta_j,0^{\ell-1-k}) \ \ \text{ and } \ \ h^{(\ell)}_{i2} = (x_i,0^{k}, z_i \gamma_j,0^{\ell-1-k}, y_i,0^{\ell})
\label{eq:hyp2iso}
\end{equation}
for $i = km+j$ with $0 \le k \le \ell-1$ and $1 \le j \le m$; $h^{(\ell)}_{i1} = (v_i,0^{\ell},w_i,0^{\ell})$ and $h^{(\ell)}_{i2} = (x_i,0^{\ell},y_i,0^{\ell})$ for $\ell m + 1 \le i \le e$; and $g^{(\ell)}_j = (v_{e+j},0^{\ell},w_{e+j},0^{\ell})$ for $1 \le j \le d$. Note that the generators in $\mathcal{G}^{(\ell)}$ are obtained from the generators in $\mathcal{G}$ by adding extra coordinates. Therefore, the fact that $\mathcal{G}$ is a minimal generating set of $\mathcal{C}$ implies that $\mathcal{G}^{(\ell)}$ is a minimal generating set of $\mathcal{C}^{(\ell)}$. It is straightforward to verify that the generators $h^{(\ell)}_{i1}$ and $h^{(\ell)}_{i2},$ $i=\ell m+1,\ldots,e$, form hyperbolic pairs, and that the generators $h^{(\ell)}_{i1}$ and $h^{(\ell)}_{i2},$ $i=1,\ldots,\ell m,$ and $g^{(\ell)}_j$, $j = 1,\ldots,d$, are isotropic. 
This verifies item (i) in the statement of the proposition. Note that, for $i=1,\ldots,\ell m$, the isotropic generators $h^{(\ell)}_{i1}$ and $h^{(\ell)}_{i2}$ are obtained from the hyperbolic pairs $(v_i,w_i)$ and $(x_i,y_i)$ in $\mathcal{G}$.

The last code, $\mathcal{C}^{(c)}$, in the sequence is taken to be the additive code of length $2(n+c)$ generated by $\mathcal{G}^{(c)} = \{h^{(c)}_{11},h^{(c)}_{12},\ldots, h^{(c)}_{e1},h^{(c)}_{e2},g^{(c)}_1,\ldots,g^{(c)}_d\}$, where
$$
h^{(c)}_{i1} = (v_i,0^{c},w_i,0^k,-\beta_j,0^{c-1-k}) \ \ \text{ and } \ \ h^{(c)}_{i2} = (x_i,0^{k}, z_i \gamma_j,0^{c-1-k}, y_i,0^{c})
$$
for $i = km+j$ with $0 \le k \le c-2$ and $1 \le j \le m$; 
$$
h^{(c)}_{i1} = (v_i,0^{c},w_i,0^{c-1},-\beta_j) \ \ \text{ and } \ \ h^{(c)}_{i2} = (x_i,0^{c-1}, z_i \gamma_j,y_i,0^{c})
$$
for $i = (c-1)m+j$ with $1 \le j \le e-(c-1)m$; and $g^{(c)}_j = (v_{e+j},0^{c},w_{e+j},0^{c})$ for $1 \le j \le d$.
It is an easy exercise to verify that the generators in $\mathcal{G}^{(c)}$ are all isotropic; hence, $\mathcal{C}^{(c)}$ is $\chi$-self-orthogonal. 

For $\ell = 0,1,\ldots,c-1$, $\mathcal{C}^{(\ell+1)}$ is a $1$-step extension of $\mathcal{C}^{(\ell)}$: a simple means of verifying this is to check that each generator $h_{i1}^{(\ell)}$ (resp.\ $h_{i2}^{(\ell)}$) is obtained by puncturing the generator $h_{i1}^{(\ell+1)}$ (resp.\ $h_{i2}^{(\ell+1)}$) at the coordinates $\ell+1$ and $2(\ell+1)$; an analogous statement clearly holds for the generators $g_j^{(\ell)}$ as well. To verify that $d_s\bigl({\mathcal{C}^{(\ell+1)}}^{\perp_{\Tr}} \setminus \mathcal{C}^{(\ell+1)}\bigr) \le d_s\bigl({\mathcal{C}^{(\ell)}}^{\perp_{\Tr}} \setminus \mathcal{C}^{(\ell)}\bigr)$, it suffices to observe that $\{(r,0,s,0) : (r,s) \in {\mathcal{C}^{(\ell)}}^{\perp_{\Tr}}  \setminus \mathcal{C}^{(\ell)}\} \subseteq {\mathcal{C}^{(\ell+1)}}^{\perp_{\Tr}} \setminus \mathcal{C}^{(\ell+1)}$. Thus, item (ii) in the statement of the proposition holds.

Finally, to verify item (iii), it is clear that $\mathcal{C}^{(c)} \subseteq \mathcal{R}^{2(n+c)}$ is a $\chi$-self-orthogonal extension of $\mathcal{C}^{(\ell)} \subseteq \text{GR}(p^b,m)^{2{n+\ell}}$, obtainable via the Two-Step Construction. The entanglement degree of $\mathcal{C}^{(c)}$ over $\mathcal{C}^{(\ell)}$ is $c-\ell$. To show the minimality of the extension, we appeal to Theorem~\ref{EAQZ} and Proposition~\ref{prop:count}(a), which give us that the minimum entanglement degree of any $\chi$-self-orthogonal extension of $\mathcal{C}^{(\ell)}$ is equal to $\left\lceil { \frac{1}{2m} \rank(\mathcal{C}^{(\ell)}/(\mathcal{C}^{(\ell)} \cap{\mathcal{C}^{(\ell)}}^{\perp_{\Tr}}))  }\right\rceil=\left\lceil  \frac{e-\ell m}{m}\right\rceil=c-\ell$. 
\end{proof}

From the sequence of additive codes in Proposition~\ref{prop:extless}, we obtain a sequence of EAQECCs that progressively trade off maximally entangled qudit pairs pre-shared between the transmitter and receiver, for qudits held solely by the transmitter. As the next theorem shows, these EAQECCs all have the same dimension, but potentially lose in minimum distance as the number of pre-shared qudits is reduced.

\begin{thm}\label{thm:less} 
Let $\mathcal{C} \subseteq \text{GR}(p^b,m)^{2n}$ be an additive code such that $\rank(\mathcal{C}  \cap  \mathcal{C}^{\perp_{\Tr}}) < \min\{\rank(\mathcal{C}),\rank(\mathcal{C}^{\perp_{\Tr}})\}$. Set $c : = \left\lceil\frac{1}{2m}\rank(\mathcal{C}/(\mathcal{C} \cap \mathcal{C}^{\perp_{\Tr}}))\right\rceil$. Let $\mathcal{C}^{(\ell)}$, $\ell = 0,1,\ldots,c$, be a sequence of additive codes constructed as in Proposition~\ref{prop:extless}. Set $D^{(\ell)} := d_{\mathrm{s}}\bigl({\mathcal{C}^{(\ell)}}^{\perp_{\Tr}} \setminus \mathcal{C}^{(\ell)}\bigr)$ and $K := p^{bm(n+c)}/|\mathcal{C}^{(c)}|$. Then, for $\ell = 0,1,\ldots,c$, there exists an $((n+\ell,K,D^{(\ell)};c-\ell))$ EAQECC over $\text{GR}(p^b,m)$. Moreover, $D^{(0)} \ge D^{(1)} \ge \cdots \ge D^{(c)}$.
\end{thm}

\begin{proof} By Theorem \ref{thm:stdform}(b), there is a minimal generating set, $\mathcal{G}$, of $\mathcal{C}$ in standard form. By Proposition~\ref{prop:count}(a) and Proposition~\ref{prop:dualGR}, the number, $e$, of hyperbolic pairs in $\mathcal{G}$ is $\frac12 \rank(\mathcal{C}/(\mathcal{C} \cap \mathcal{C}^{\perp_{\Tr}}))$.   As $\rank(\mathcal{C}  \cap  \mathcal{C}^{\perp_{\Tr}}) < \rank(\mathcal{C}),$ we have $\mathcal{C}   \not\subseteq \mathcal{C}^{\perp_{\Tr}},$ so that $e > 0$, and hence, $c = \lceil \frac{e}{m} \rceil \ge 1$. There exists a sequence of additive codes $\mathcal{C}^{(\ell)} \subseteq \text{GR}(p^b,m)^{2(n+\ell)}$, $\ell = 0,1,\ldots,c$, with the properties listed in the statement of Proposition~\ref{prop:extless}. For each $\ell \in 0,1,\ldots,c$, we apply Theorem \ref{thm:eaqecc} to $C^{(\ell)}$ and its $\chi$-self-orthogonal extension $\mathcal{C}^{(c)}$ to obtain an $((n+{\ell},p^{mb(n+c)}/{|\mathcal{C}^{(c)}|},\widetilde{D}^{(\ell)};c-\ell))$ EAQECC, with  \begin{equation*}
\widetilde{D}^{(\ell)}\ = \ \left\{\begin{array}{ll}
d_{\mathrm{s}}\bigl({\mathcal{C}^{(\ell)}}^{\perp_{\Tr}}\bigr) & \text{~~if~~} {\mathcal{C}^{(\ell)}}^{\perp_{\Tr}} \subseteq \mathcal{C}^{(\ell)} \\
d_{\mathrm{s}}\bigl({\mathcal{C}^{(\ell)}}^{\perp_{\Tr}} \setminus \mathcal{C}^{(\ell)}\bigr)  & \text{~~otherwise}.\,
\end{array}\right.
\end{equation*}

It only remains to show that $\widetilde{D}^{(\ell)} = D^{(\ell)} := d_{\mathrm{s}}\bigl({\mathcal{C}^{(\ell)}}^{\perp_{\Tr}} \setminus \mathcal{C}^{(\ell)}\bigr)$, or equivalently, that ${\mathcal{C}^{(\ell)}}^{\perp_{\Tr}} \not \subseteq \mathcal{C}^{(\ell)}$, for $\ell = 0,1,\ldots,c$.
By assumption, $\rank(\mathcal{C}  \cap  \mathcal{C}^{\perp_{\Tr}}) <  \rank(\mathcal{C}^{\perp_{\Tr}})$, so we have $\mathcal{C}^{\perp_{\Tr}} \not\subseteq \mathcal{C}$. Since $\mathcal{C}^{(0)} = \mathcal{C}$, we have ${\mathcal{C}^{(0)}}^{\perp_{\Tr}} \not \subseteq \mathcal{C}^{(0)}$, and hence, $D^{(0)} = d_{\mathrm{s}}\bigl({\mathcal{C}^{(0)}}^{\perp_{\Tr}} \setminus \mathcal{C}^{(0)}\bigr) < \infty$. Then, from property (ii) in Proposition~\ref{prop:extless}, it follows that for $\ell = 0,1,\ldots,c-1$, we have $d_{\mathrm{s}}\bigl({\mathcal{C}^{(\ell)}}^{\perp_{\Tr}} \setminus \mathcal{C}^{(\ell)}\bigr) < \infty$. Thus, recalling that, by definition, $d_s(A) = \infty$ for $A = \emptyset$, we have ${\mathcal{C}^{(\ell)}}^{\perp_{\Tr}} \not \subseteq \mathcal{C}^{(\ell)}$ for $\ell = 0,1,\ldots,c$. 

Finally, the chain of inequalities $D^{(0)} \ge D^{(1)} \ge \cdots \ge D^{(c)}$ follows again from property (ii) in Proposition~\ref{prop:extless}. 
\end{proof}


The following example illustrates the construction of the EAQECCs guaranteed by Theorem~\ref{thm:less}. In particular, it shows that it is possible for the inequality $D^{(\ell)} \ge D^{(\ell+1)}$ to hold with equality, indicating that it is sometimes possible to reduce the number of maximally entangled qudit pairs without incurring a penalty in terms of minimum distance. 

\begin{ex}Let $\mathcal{C}$ be the additive code of length $8$ over $\mathbb{Z}_4 \, (=\text{GR}(2^2,1))$ with a minimal generating set 
$\{(v_1,w_1),(v_2,w_2), $ $(v_3,w_3),$ $(x_1,y_1),(x_2,y_2)\},$  where
$(v_1,w_1)= (0,0,0,0,0,0,2,0)$, $(x_1,y_1)= (0,0,1,0,0,0,0,0),$ $(v_2,w_2)= (0,0,0,0,1,1,0,0)$, $(x_2,y_2)= (1,1,0,0,0,0,0,0)$ and $(v_3,w_3)=(0,0,0,1,0,0,0,0)$.
 It is easy to check that for $i = 1,2$, the generators $(v_i,w_i)$ and $(x_i,y_i)$ form hyperbolic pairs, and the generator $(v_3,w_3)$ is isotropic. Moreover, it can be verified that $\{(1,3,0,0,0,0,0,0),$ $(0,0,2,0,0,0,0,0),$ $(0,0,0,1,0,0,0,0),$ $(0,0,0,0,1,3,0,0)\}$ is a minimal generating set of  $\mathcal{C}^{\perp_{\mathrm{s}}}$.

Applying the construction in Proposition~\ref{prop:extless} gives us additive codes $\mathcal{C}^{(1)}$ and $\mathcal{C}^{(2)}$ over $\mathbb{Z}_4$ generated by $\{(v_1,0,w_1,-1),(v_2,0,w_2,0),(v_3,0,w_3,0),(x_1,2,y_1,0),(x_2,0,y_2,0)\}$ and  $\{(v_1,0,0,w_1,-1,0),(v_2,0,0,w_2,0,-1),(v_3,0,0,w_3,0,0),(x_1,2,0,y_1,0,0),(x_2,0,2,y_2,0,0)\},$ respectively. The code $\mathcal{C}^{(2)}\subseteq\mathbb{Z}_4^{2(4+2)}$ is a minimal $\chi$-self-orthogonal extension, obtainable via the Two-Step Construction, of $\mathcal{C}$ and $\mathcal{C}^{(1)}$. We have verified (using the Magma algebra system \cite{bosma}) that $|\mathcal{C}^{(2)}| = 4^5$, $d_{\mathrm{s}}(\mathcal{C}^{\perp_{\mathrm{s}}}\setminus \mathcal{C})=2$, $d_{\mathrm{s}}(\mathcal{C}^{(1)^{\perp_{\mathrm{s}}}}\setminus \mathcal{C}^{(1)})=2$, and  $d_{\mathrm{s}}(\mathcal{C}^{(2)^{\perp_{\mathrm{s}}}}\setminus \mathcal{C}^{(2)})=1$.  Thus, by Theorem \ref{thm:eaqecc}, there exist a $((4,4,2;2))$ EAQECC and a $((5,4,2;1))$ EAQECC over $\mathbb{Z}_4$. The code $\mathcal{C}^{(2)}$ itself yields a $((6,4,1))$ quantum stabilizer code over $\mathbb{Z}_4$.  Note that the EAQECCs constructed from $\mathcal{C}$ and $\mathcal{C}^{(1)}$ have the same minimum distance ($D^{(0)} = D^{(1)} = 2)$, while the quantum stabilizer code constructed from $\mathcal{C}^{(2)}$ has a strictly smaller minimum distance. \qed
\end{ex}

\bigskip

The next example brings out a different aspect of the effect that this lengthening method can have on the error-handling capability of the resulting EAQECCs. It considers a situation when we start with an additive code $\mathcal{C}$ over $\mathbb{Z}_4$  with a minimal generating set $\mathcal{G}$ in standard form that contains (at least) two hyperbolic pairs. A $1$-step extension of $\mathcal{C}$ can be obtained, as in Proposition~\ref{prop:extless}, by converting one of the hyperbolic pairs into isotropic generators. The example shows that the choice of the hyperbolic pair of generators to be made isotropic can make a significant difference.

\begin{ex} Let $\mathcal{C}$ be the additive code over $\mathbb{Z}_4$ of length 10 with a minimal generating set 
$\{(v_1,w_1),(v_2,w_2), $ $(v_3,w_3),$ $(v_4,w_4),(x_1,y_1),(x_2,y_2)\},$
 where $(v_1,w_1)=(1,1,0,0,0,1,0,0,1,0),$ $(x_1,y_1)=(1,0,0,1,0,0,1,0,1,0),$ $(v_2,w_2)=(0,0,0,0,1,0,0,0,0,0),$ $(x_2,y_2)=(0,0,0,0,0,0,0,0,0,1),$ $(v_3,w_3)=(1,1,1,3,0,0,0,0,0,0),$ and $(v_4,w_4)=(0,0,0,0,0,3,1,1,1,0)$. The generators $(v_i,w_i)$ and $(x_i,y_i)$ form hyperbolic pairs for $i=1,2$, and the generators $(v_3,w_3)$ and $(v_4,w_4)$ are isotropic. It can be verified that $\{(v_3,w_3),(v_4,w_4),(e_1,f_1),(e_2,f_2)\}$ with $(e_1,f_1)=(3,0,3,0,0,0,3,1,0,0)$ and $(e_2,f_2)=(0,3,1,0,0,3,1,0,0,0),$ is a minimal generating set of  $\mathcal{C}^{\perp_{\mathrm{s}}}$. Moreover, using the Magma algebra system \cite{bosma}, we get that $d_{\mathrm{s}}(\mathcal{C}^{\perp_{\mathrm{s}}}) = d_{\mathrm{s}}(\mathcal{C}^{\perp_{\mathrm{s}}}\setminus \mathcal{C})=3$.
 
 Following the construction of Proposition~\ref{prop:extless}, we can obtain two $1$-step extensions of $\mathcal{C}$, which we will denote by $\mathcal{C}^{(1)}$ and $\widetilde{\mathcal{C}}^{(1)}$. The first of these is obtained by converting the hyperbolic pair $(v_1,w_1)$ and $(x_1,y_1)$ into isotropic generators as prescribed in \eqref{eq:hyp2iso}. Thus, $\mathcal{C}^{(1)}$ has the generating set 
 $$
 \{(v_1,0,w_1,-1),(v_2,0,w_2,0),(v_3,0,w_3,0),(v_4,0,w_4,0),(x_1,-1,y_1,0),(x_2,0,y_2,0)\}.
 $$ 
 The code $\widetilde{\mathcal{C}}^{(1)}$, on the other hand, is obtained by converting the hyperbolic pair $(v_2,w_2)$ and $(x_2,y_2)$ into isotropic generators; thus, it has the generating set 
 $$
 \{(v_1,0,w_1,0),(v_2,0,w_2,3),(v_3,0,w_3,0),(v_4,0,w_4,0),(x_1,0,y_1,0),(x_2,1,y_2,0)\}.
 $$
 
The codes $\mathcal{C}$, $\mathcal{C}^{(1)}$ and $\mathcal{C}$, $\widetilde{\mathcal{C}}^{(1)}$ have common minimal $\chi$-self-orthogonal extensions $\mathcal{C}^{(2)}, \widetilde{\mathcal{C}}^{(2)}  \subseteq\mathbb{Z}_4^{2(5+2)}$ generated by 

$$
 \{(v_1,0,0,w_1,-1,0),(v_2,0,0,w_2,0,3),(v_3,0,0,w_3,0,0),(v_4,0,0,w_4,0,0),(x_1,-1,0,y_1,0,0),(x_2,0,1,y_2,0,0)\}
 $$  
 and 
 $$
 \{(v_1,0,0,w_1,0,-1),(v_2,0,0,w_2,3,0),(v_3,0,0,w_3,0,0),(v_4,0,0,w_4,0,0),(x_1,0,-1,y_1,0,0),(x_2,1,0,y_2,0,0)\},$$ respectively.  
 
 By Theorem~\ref{thm:eaqecc}, the codes $\mathcal{C}$, $\mathcal{C}^{(1)}$ and $\widetilde{\mathcal{C}}^{(1)}$ give rise to EAQECCs $\mathcal{Q}$, $\mathcal{Q}^{(1)}$ and $\widetilde{\mathcal{Q}}^{(1)}$, respectively, over $\mathbb{Z}_4$. The parameters of these EAQECCs can be computed with the aid of Magma \cite{bosma} as $((5,4,3;2))$, $((6,4,3;1))$ and $((6,4,3;1))$, respectively. While $\mathcal{Q}^{(1)}$ and $\widetilde{\mathcal{Q}}^{(1)}$ have the same parameters, they differ significantly in how they handle errors. Note that both these EAQECCs have minimum distance $3$, so they can correct all weight-$1$ errors in $\mathcal{P}_7(\mathbb{Z}_4)$ of the form $X(a,0)Z(b,0)$, with $a,b \in \mathbb{Z}_4^6$.

We claim that $\mathcal{Q}^{(1)}$ is non-degenerate in the following sense: for any two weight-$1$ errors $E_1=X(a_1,0)Z(b_1,0)$ and $E_2=X(a_2,0)Z(b_2,0)$, with $a_1,b_1,a_2,b_2 \in \mathbb{Z}_4^6$, we have that $ E_1 | u \rangle$ and $E_2 | v \rangle$ are linearly independent for any $| u \rangle, | v \rangle \in \mathcal{Q}^{(1)}$. To see this, we first note that by Lemma~\ref{comm}, $E_1^{\dag}E_2 = \omega^{\ell}X(-a_1+a_2,0)Z(-b_1+b_2,0)$ for some $\ell$, so that the weight of $E_1^{\dag}E_2$ is at most 2. This implies that $w_{\mathrm{s}}((-a_1+a_2,0,-b_1+b_2,0)) \leq 2.$ However, using the Magma algebra system \cite{bosma}, we get that $d_{\mathrm{s}}({\mathcal{C}^{(1)}}^{\perp_{\mathrm{s}}})=3$, which implies that $(-a_1+a_2,0,-b_1+b_2,0)\not\in{\mathcal{C}^{(1)}}^{\perp_{\mathrm{s}}}$. Thus, $E_1^{\dag}E_2$ does not belong to the centralizer of the stabilizer of $\mathcal{Q}^{(1)}$ in $\mathcal{P}_{7}(\mathbb{Z}_4)$. Hence, following the proof of Theorem \ref{thm:eaqecc}, we have $\langle u | E_1^{\dag}E_2| v \rangle=0$ for any $| u \rangle,| v \rangle \in \mathcal{Q}.$ This implies that $E_1| u \rangle$ and $E_2| v \rangle$  are linearly independent  for any $| u \rangle,| v \rangle \in \mathcal{Q}^{(1)}$, thus proving the claim. We remark here that a similar argument also shows that the EAQECC $\mathcal{Q}$ is also non-degenerate in an analogous sense.
 
On the other hand, $\widetilde{\mathcal{Q}}^{(1)}$ is degenerate, as we now show. Following the proof of Theorem \ref{thm:eaqecc}, we observe that for some $\ell \in \mathbb{Z}$, $\omega^\ell X(v_2,0,0)Z(w_2,3,0)$  ($= \omega^\ell X(0,0,0,0,1,0,0)Z(0,0,0,0,0,3,0)$ is an element of the stabilizer of $\widetilde{\mathcal{Q}}^{(1)}$. This implies that the weight-$1$ errors $\omega^\ell Z(0,0,0,0,0,3,0)$ and $X(0,0,0,0,3,0,0)$ ($=X(0,0,0,0,1,0,0)^{-1}$) produce the same result when applied to any element of $\widetilde{\mathcal{Q}}^{(1)}$, i.e.,  $\omega^\ell Z(0,0,0,0,0,3,0)| u \rangle$ and $X(0,0,0,0,3,0,0)| u \rangle$  are linearly dependent  for any $| u \rangle \in \widetilde{\mathcal{Q}}^{(1)}$. Therefore, by definition, $\widetilde{\mathcal{Q}}^{(1)}$ is a degenerate code. \qed
\end{ex}

\subsection{Retaining the Number of Maximally Entangled Qudit Pairs} \label{subsec:same}

The second method for lengthening an additive code is encapsulated in the proposition below.

\begin{prop} \label{prop:extsame}
Let $\mathcal{C} \subseteq \text{GR}(p^b,m)^{2n}$ be an additive code such that the additive code $\mathcal{C}^{\perp_{\Tr}}/(\mathcal{C}  \cap  \mathcal{C}^{\perp_{\Tr}})$ is free as a $\mathbb{Z}_{p^b}$-module and has  rank at least  $2m$. Further, let any minimal generating set of $\mathcal{C}$ in standard form have exactly $e$ hyperbolic pairs and $d$ isotropic generators. Set $c := \lceil \frac{e}{m}\rceil$.   Then, there exists a $1$-step extension, $\mathcal{M} \subseteq \text{GR}(p^b,m)^{2(n+1)}$, of $\mathcal{C}^{\perp_{\Tr}}$ such that the following hold:
\begin{itemize}
\item[(i)]  $\mathcal{M}^{\perp_{\Tr}}$ has a minimal generating set in standard form with exactly $e$ hyperbolic pairs and $d+2m$ isotropic generators, and $\mathcal{M}^{\perp_{\Tr}}$ contains a $1$-step extension of $\mathcal{C}$.
  \item[(ii)] $d_s\bigl({\mathcal{M}} \setminus \mathcal{M}^{\perp_{\Tr}}\bigr) \ge d_s\bigl({\mathcal{C}}^{\perp_{\Tr}} \setminus \mathcal{C}\bigr)$
   if $\rank(\mathcal{C}^{\perp_{\Tr}}/(\mathcal{C}^{\perp_{\Tr}} \cap \mathcal{C})) > 2m$, and 
   $d_s\bigl({\mathcal{M}} \bigr) \geq  d_s\bigl({\mathcal{C}}^{\perp_{\Tr}} \bigr)$ if $\rank(\mathcal{C}^{\perp_{\Tr}}/(\mathcal{C}^{\perp_{\Tr}} \cap \mathcal{C})) = 2m$;
  \item[(iii)] There exist minimal $\chi$-self-orthogonal extensions  $\mathcal{C}'\subseteq \text{GR}(p^b,m)^{2(n+c)}$  and $\mathcal{M}'\subseteq \text{GR}(p^b,m)^{2(n+c+1)}$ of $\mathcal{C}$ and $\mathcal{M}^{\perp_{\Tr}}$, respectively,  obtained using the Two-Step Construction, with $|\mathcal{M}'|=|\mathcal{C}'| \,p^{2mb}.$ 
 \end{itemize}
\end{prop}

It is the code $\mathcal{M}^{\perp_{\Tr}}$ in the statement of the proposition that we view as the desired lengthening of the additive code $\mathcal{C}$. We sketch here the construction of the code $\mathcal{M}$. The details of the verification of its properties (i), (ii) and (iii) are left to the proof given in Appendix~\ref{app:extsame_proof}.

Let $\mathcal{H}=\{(e_1,f_1),(e_2,f_2),\ldots, $ $(e_{t+s},f_{t+s}),$ $(g_1,h_1),(g_2,h_2),\ldots, (g_{t},h_{t})\}$ be any minimal generating set in standard form  of $\mathcal{C}^{\perp_{\Tr}}$  as a $\mathbb{Z}_{p^b}$-module,  such that, for $i = 1,2,\ldots,t$, the generators $(e_i,f_i)$ and $(g_i,h_i)$ form hyperbolic pairs, and the generators $(e_i,f_i)$, $i = t+1,\ldots,t+s$, are isotropic. From the condition that $\rank(\mathcal{C}^{\perp_{\Tr}}/(\mathcal{C}^{\perp_{\Tr}} \cap \mathcal{C})) \geq 2m$, we obtain, via Proposition~\ref{prop:count}(a), $t\geq m$. We follow the recipe for the $1$-step extension in the proof of Proposition~\ref{prop:extless} that converts the first $m$ hyperbolic pairs in $\mathcal{H}$ into isotropic generators, as described by \eqref{eq:hyp2iso} with $\ell = 1$. The resulting $1$-step extension of $\mathcal{C}^{\perp_{\Tr}}$ is precisely the additive code $\mathcal{M}$ in the statement of Proposition~\ref{prop:extsame}.

From the above proposition, we can readily derive a propagation rule for EAQECCs over Galois rings along the lines of that in \cite[Theorem~16]{luo1} --- see also Corollary~\ref{cor:same} in Appendix~\ref{app:trace-alt-form}. 

\begin{thm}\label{thm:same} 
 Let $\mathcal{C} \subseteq \text{GR}(p^b,m)^{2n}$ be an additive code such that the additive code $\mathcal{C}^{\perp_{\Tr}}/(\mathcal{C}  \cap  \mathcal{C}^{\perp_{\Tr}})$ is free as a $\mathbb{Z}_{p^b}$-module of rank at least  $2m.$ Set $c : = \left\lceil\frac{1}{2m}\rank(\mathcal{C}/(\mathcal{C} \cap \mathcal{C}^{\perp_{\Tr}}))\right\rceil$ and $D := d_{\mathrm{s}}\bigl({\mathcal{C}}^{\perp_{\Tr}} \setminus \mathcal{C}\bigr)$.  Then, the existence of an $((n,K,D;c))$ EAQECC over $\text{GR}(p^b,m)$ constructed from $\mathcal{C}$ using Theorem~\ref{dt3} implies the existence of an $((n+1,\frac{1}{p^{bm}}K,D';c))$ EAQECC  over  $\text{GR}(p^b,m)$ with 
 \begin{equation*}
D'~\geq ~\left\{\begin{array}{ll}
d_{\mathrm{s}}(\mathcal{C}^{\perp_{\Tr}})  & \text{~~if~~}  \text{rank} (\mathcal{C}^{\perp_{\Tr}}/(\mathcal{C}  \cap  \mathcal{C}^{\perp_{\Tr}})) =2m; \\
D   & \text{~~if~~} \text{rank} (\mathcal{C}^{\perp_{\Tr}}/(\mathcal{C}  \cap  \mathcal{C}^{\perp_{\Tr}})) >2m.\,
\end{array}\right.
\end{equation*}

\end{thm}

\begin{proof} By Theorem \ref{thm:stdform}(b), there is a minimal generating set, $\mathcal{G}$, of $\mathcal{C}$ in standard form. By Proposition~\ref{prop:count}(a) and Proposition~\ref{prop:dualGR}, the number, $e$, of hyperbolic pairs in $\mathcal{G}$ is $\frac12 \rank(\mathcal{C}/(\mathcal{C} \cap \mathcal{C}^{\perp_{\Tr}}))$. There exists an additive code $\mathcal{M}\subseteq \text{GR}(p^b,m)^{2(n+1)}$ with the properties listed in the statement of Proposition~\ref{prop:extsame}. We apply Theorem~\ref{thm:eaqecc} to $\mathcal{C}$ and $\mathcal{M}^{\perp_{\Tr}}$ using $\chi$-self-orthogonal extensions $\mathcal{C}'$ and $\mathcal{M}',$ respectively,  to obtain an $((n,p^{mbn}/{|\mathcal{C}'|},D;c))$  and an $((n+1,p^{mb(n+1)}/{|\mathcal{M}'|},D';c))$ EAQECC, with 
 \begin{equation*}
D~=~\left\{\begin{array}{ll}
d_{\mathrm{s}}(\mathcal{C}^{\perp_{\Tr}}) & \text{~~if~~} \mathcal{C}^{\perp_{\Tr}} \subseteq \mathcal{C} \\
d_{\mathrm{s}}(\mathcal{C}^{\perp_{\Tr}}\setminus \mathcal{C} )  & \text{~~otherwise}\,,
\end{array}\right.
\end{equation*} and   \begin{equation}\label{edis1}
D'~=~\left\{\begin{array}{ll}
d_{\mathrm{s}}(\mathcal{M}) & \text{~~if~~} \mathcal{M} \subseteq \mathcal{M}^{\perp_{\Tr}} \\
d_{\mathrm{s}}(\mathcal{M}\setminus \mathcal{M}^{\perp_{\Tr}} )  & \text{~~otherwise}.\,
\end{array}\right.
\end{equation}  Since $|\mathcal{M}'|=|\mathcal{C}'|p^{2mb},$  we have   $p^{mb(n+c+1)}/{|\mathcal{M}'|}=p^{mb(n+c)}/{(p^{mb}|\mathcal{C}'|)}.$
As $\text{rank}(\mathcal{C}^{\perp_{\Tr}}/(\mathcal{C}  \cap  \mathcal{C}^{\perp_{\Tr}})) \geq 2m,$ we have $\mathcal{C}^{\perp_{\Tr}} \not\subseteq \mathcal{C}.$ Thus $D=d_{\mathrm{s}}(\mathcal{C}^{\perp_{\Tr}}\setminus \mathcal{C} ).$ Note that $\mathcal{M} \subseteq \mathcal{M}^{\perp_{\Tr}}$ if and only if $\text{rank} (\mathcal{C}^{\perp_{\Tr}}/(\mathcal{C}  \cap  \mathcal{C}^{\perp_{\Tr}})) =2m.$ 
 From this, using Proposition~\ref{prop:extsame}(b) and \eqref{edis1}, we observe that if   $D' \geq d_{\mathrm{s}}(\mathcal{C}^{\perp_{\Tr}}) $ if $\text{rank} (\mathcal{C}^{\perp_{\Tr}}/(\mathcal{C}  \cap  \mathcal{C}^{\perp_{\Tr}})) =2m$ and $D' \geq D$ if $\text{rank} (\mathcal{C}^{\perp_{\Tr}}/(\mathcal{C}  \cap  \mathcal{C}^{\perp_{\Tr}})) >2m.$ This completes the proof of the theorem. \end{proof}

The corollary below specializes Theorem~\ref{thm:same} to the case of finite fields. 

\begin{cor} Let $\mathcal{C} \subseteq \mathbb{F}_{p^m}^{2n}$ be an additive code over the finite field $\F_{p^m}$ such that $(\dim_{\F_p}(\mathcal{C}^{\perp_{\tr}}) - \dim_{\F_p}(\mathcal{C} \cap \mathcal{C}^{\perp_{\tr}})) \geq 2m.$ Set $c : = \left\lceil  \frac{1}{2m} (\dim_{\F_p}(\mathcal{C}) - \dim_{\F_p}(\mathcal{C} \cap \mathcal{C}^{\perp_{\tr}})) \right\rceil $ and $D=d_{\mathrm{s}}\bigl({\mathcal{C}}^{\perp_{\tr}} \setminus \mathcal{C}\bigr)$.  Then the existence of an $((n,p^{m(n+c)}/{|\mathcal{C}|},D;c))$ EAQECC over $\mathbb{F}_{p^m}$ constructed from $\mathcal{C}$ using Corollary~\ref{cor:Fpm} implies the existence of an $((n+1,p^{m(n+c-1)}/{|\mathcal{C}|},D';c))$ EAQECC over $ \mathbb{F}_{p^m}$ with
  \begin{equation*}
D'~\geq ~\left\{\begin{array}{ll}
d_{\mathrm{s}}(\mathcal{C}^{\perp_{\tr}})  & \text{~~if~~} (\dim_{\F_p}(\mathcal{C}^{\perp_{\tr}}) - \dim_{\F_p}(\mathcal{C} \cap \mathcal{C}^{\perp_{\tr}}))  =2m; \\
D   & \text{~~if~~} (\dim_{\F_p}(\mathcal{C}^{\perp_{\tr}}) - \dim_{\F_p}(\mathcal{C} \cap \mathcal{C}^{\perp_{\tr}}))  >2m.\,
\end{array}\right.
\end{equation*}
 \end{cor}
\begin{proof} The proof follows from Theorem \ref{thm:same} and the observations that, for finite fields $\F_{p^m},$ $\mathcal{C},$ $\mathcal{C} \cap \mathcal{C}^{\perp_{\tr}},$ $\mathcal{C}/(\mathcal{C} \cap \mathcal{C}^{\perp_{\tr}})$ and $\mathcal{C}^{\perp_{\tr}}/(\mathcal{C} \cap \mathcal{C}^{\perp_{\tr}})$ are vector spaces over $\F_p$, thus are free modules.\end{proof}

The following example illustrates the propagation rule for EAQECCs given in Theorem~\ref{thm:same}. In particular, it shows that the new EAQECC obtained via the lengthening procedure can have a strictly larger minimum distance than the original EAQECC, while keeping the number of maximally entangled qudit pairs the same. However, there is a price to be paid in terms of reduction in dimension. 

\begin{ex} Let $\mathcal{C}$ be an additive code over $\mathbb{Z}_9 \ (=\text{GR}(3^2,1))$ of length 12 with a minimal generating set 
$\{(v_1,w_1),(v_2,w_2),(v_3,w_3),(v_4,w_4),(x_1,y_1),(x_2,y_2)\}$,
 where $(v_1,w_1)=(0,0,0,0,1,1,0,0,0,0,1,0)$, $(x_1,y_1)=(0,0,0,0,2,0,0,0,0,1,1,0)$, $(v_2,w_2)=(0,0,0,1,1,0,0,0,0,0,1,0)$, $(x_2,y_2)=(0,0,0,0,2,0,0,0,0,0,1,1)$, $(v_3,w_3)=(1,1,0,0,0,0,0,1,0,0,0,0),$ and $(v_4,w_4)=(0,2,0,0,0,0,1,1,0,0,0,0)$. For $i = 1,2$, the generators $(v_i,w_i)$ and $(x_i,y_i)$ form hyperbolic pairs, whereas the generators $(v_3,w_3)$ and $(v_4,w_4)$ are isotropic. It can be verified that $\mathcal{C}^{\perp_{\mathrm{s}}}$ has $\{(v_3,w_3),(v_4,w_4),(e_1,f_1),(g_1,h_1),(e_2,f_2),(g_2,h_2)\}$ as a minimal generating set, with $(e_1,f_1)=(0,0,1,0,0,0,0,0,0,0,0,0), \ (g_1,h_1)=(0,0,0,0,0,0,0,0,1,0,0,0), \ (e_2,f_2)=(0,0,0,1,1,1,0,0,0,0,1,0)$, and $(g_2,h_2)=(0,0,0,0,2,0,0,0,0,1,1,1)$.  
 Note that $\mathcal{C},$ $\mathcal{C}^{\perp_{\mathrm{s}}},$ $\mathcal{C}/\mathcal{C}  \cap  \mathcal{C}^{\perp_{{\mathrm{s}}}}$  and $\mathcal{C}^{\perp_{\mathrm{s}}}/\mathcal{C}  \cap  \mathcal{C}^{\perp_{{\mathrm{s}}}}$ are all free $\mathbb{Z}_9$-modules. 
Moreover, using Magma \cite{bosma}, we get that $d_{\mathrm{s}}(\mathcal{C}^{\perp_{\mathrm{s}}}\setminus \mathcal{C}))=1$. Thus, by Theorem~\ref{dt2}, we obtain from $\mathcal{C}$ a $((6,9^2,1;2))$ EAQECC over $\mathbb{Z}_9.$ 

 Following the construction of Proposition~\ref{prop:extsame}, we first obtain a $1$-step extension, $\mathcal{M}$, of $\mathcal{C}^{\perp_{\mathrm{s}}}$, by converting the hyperbolic pair $(e_2,f_2)$ and $(g_2,h_2)$ into isotropic generators.  Thus, $\mathcal{M}$ has the generating set  
$\{(e_1,0,f_1,8),(g_1,8,h_1,0),(v_3,0,w_3,0),(v_4,0,w_4,0),(e_2,0,f_2,0),(g_2,0,h_2,0)\}$. Then, as in the proof of Proposition~\ref{prop:extsame} in Appendix~\ref{app:extsame_proof}, $\mathcal{M}^{\perp_{\mathrm{s}}}$ has 
$$\mathcal{A}=\{(e_1,0,f_1,8),(g_1,8,h_1,0),(v_1,0,w_1,0),(v_2,0,w_2,0),(v_3,0,w_3,0),(v_4,0,w_4,0),(x_1,0,y_1,0),(x_2,0,y_2,0)\}$$
as a minimal generating set. Using Magma \cite{bosma}, we find that $d_{\mathrm{s}}(\mathcal{M}\setminus \mathcal{M}^{\perp_{\mathrm{s}}}))=3$. Note that $\mathcal{M}^{\perp_{\mathrm{s}}}$ is a free $\mathbb{Z}_9$-module 9 (as indeed is $\mathcal{M}^{\perp_{{\mathrm{s}}}}/\mathcal{M}  \cap  \mathcal{M}^{\perp_{{\mathrm{s}}}}$). So, via Theorem~\ref{dt2}, we obtain from $\mathcal{M}^{\perp_{\mathrm{s}}}$ a $((7,9,3;2))$ EAQECC over $\mathbb{Z}_9$. This has minimum distance strictly larger than that of the original EAQECC obtained from $\mathcal{C}$.

While constructing the $1$-step extension from $\mathcal{C}^{\perp_{\mathrm{s}}}$, we could instead have converted the hyperbolic pair $(e_2,f_2)$ and $(g_2,h_2)$ into isotropic generators. If we had done this and gone through the same lengthening procedure, we would have ended up with a $((7,9,1;2))$ EAQECC, which is not as good as the $((7,9,3;2))$ EAQECC above. This once again demonstrates that the choice of which hyperbolic pairs to convert into isotropic generators can make a significant difference.
 \end{ex}

\section{Conclusion} \label{sec6}


In this paper, we proposed a framework to construct, from first principles, EAQECCs from additive codes over finite commutative local Frobenius rings. Given an additive code $\mathcal{C}$ over such a ring $\mathcal{R}$, we provide a means of constructing an EAQECC that uses a certain number of pre-shared pairs of maximally entangled qudits, a number that is determined by the algebraic structure of the code $\mathcal{C}$. This yields an upper bound on the minimum number of pairs of maximally entangled qudits required to construct an EAQECC from $\mathcal{C}$. We show that this bound is in fact tight for EAQECCs constructed from additive codes over $\mathbb{Z}_{p^a}$, but it can be improved when we start from codes over more general Galois rings. We derive an explicit expression for the minimum number of pre-shared pairs of maximally entangled qudits required to construct an EAQECC from an additive code over a Galois ring. This result significantly extends known results for EAQECCs constructed from additive codes over finite fields. Finally, we presented two methods to lengthen an additive code so as to modify, with some degree of control, the parameters of the EAQECCs obtained from these codes.

An interesting direction of future work would be to extend our formula for the minimum number of pre-shared pairs of maximally entangled qudits to EAQECCs over general finite commutative local Frobenius rings (beyond Galois rings). It would also be useful to find constructions of EAQECCs over Frobenius rings with good parameters, for example, codes that saturate the generalized quantum Singleton bound applicable to EAQECCs \cite{GHW21}. Finally, it would be of considerable interest to extend the EAQECC framework, starting from Theorem~\ref{thm:stdform}, to the general setting of finite commutative Frobenius rings, beyond the cases discussed at the end of Section~\ref{sec3}, obtained via the Chinese Remainder Theorem.

\section*{Appendices}
\appendix
\section{A CSS-Like Construction of EAQECCs} \label{app:CSS}
In the section, we will provide a construction of EAQECCs of length $n$ over the Galois ring $\text{GR}(p^b,m)$ using two additive codes of length $n$ over $\text{GR}(p^b,m).$ This construction is a straightforward generalization of the well-known CSS construction, due to Calderbank and Shor \cite{calderbank} and Steane \cite{steane}.   

We start with some definitions. As usual, for $u = (u_1,\ldots,u_n), v = (v_1,\ldots,v_n)\in \text{GR}(p^b,m)^n$, $u \cdot v$ denotes the standard (Euclidean) dot product $\sum_{i=1}^nu_iv_i$. 
\begin{defin} For a subset $\mathcal{C} \subseteq \text{GR}(p^b,m)^{n}$, we define\footnote{To distinguish this from the notation used for the trace-symplectic dual, we use $\independent$ instead of $\perp$ to denote the dual object.}
\begin{itemize}
\item the trace-Euclidean dual, $\mathcal{C}^{\independent_{\Tr}}$, of $\mathcal{C}$ as
$$\mathcal{C}^{\independent_{\Tr}}=\{ v \in \text{GR}(p^b,m)^{n}  : \Tr(u \cdot v) = 0 \text{~for all~}  u \in \mathcal{C}\},$$
\item and for $t = 0,1,\ldots,b$, the trace-Euclidean $t$-dual of $\mathcal{C}$ as 
$$\mathcal{C}^{\independent_{\Tr,t}} = \{v  \in \text{GR}(p^b,2m)^{n} ~:~ \Tr(u \cdot v) \equiv 0 \!\!\!   \pmod{p^{b-t}} \text{~for all~}  u \in \mathcal{C}\}.$$
\end{itemize}
(Note that $\mathcal{C}^{\independent_{\Tr,0}} = \mathcal{C}^{\independent_{\Tr}}$.)
\end{defin}



We will be applying the construction in Theorem~\ref{thm:dt3_restatement} to the direct sum $\mathcal{C}_1 \oplus \mathcal{C}_2 := \{(u,v): u \in \mathcal{C}_1, v \in \mathcal{C}_2\}$. For this, the following lemma relating the trace-symplectic dual of the direct sum to the trace-Euclidean duals of the component codes will be useful.

\begin{lem}
\label{lem:oplus}
Let $\mathcal{C}_1$ and $\mathcal{C}_2$ be additive codes of length $n$ over $\text{GR}(p^b,m)$, and let $\mathcal{C} = \mathcal{C}_1 \oplus \mathcal{C}_2$. We then have, for $t = 0,1,\ldots,b$,
\begin{itemize}
\item[\emph{(a)}] $\mathcal{C}^{\perp_{\Tr,t}} = \mathcal{C}_2^{\independent_{\Tr,t}} \oplus \mathcal{C}_1^{\independent_{\Tr,t}}$ for $t = 0,1,\ldots,b$, and
\item[\emph{(b)}] $\mathcal{C}/(\mathcal{C} \cap \mathcal{C}^{\perp_{\Tr,t}}) \  \cong \ \mathcal{C}_1/ (\mathcal{C}_1 \cap\mathcal{C}_2^{\independent_{\Tr,t}})  \oplus \mathcal{C}_2/(\mathcal{C}_2 \cap \mathcal{C}_1^{\independent_{\Tr,t}})$.
\end{itemize}
In particular (the $t=0$ case), $\mathcal{C}^{\perp_{\Tr}} = \mathcal{C}_2^{\independent_{\Tr}} \oplus \mathcal{C}_1^{\independent_{\Tr}}$ and $\mathcal{C}/(\mathcal{C} \cap \mathcal{C}^{\perp_{\Tr}})  \cong \mathcal{C}_1/ (\mathcal{C}_1 \cap\mathcal{C}_2^{\independent_{\Tr}})  \oplus \mathcal{C}_2/(\mathcal{C}_2 \cap \mathcal{C}_1^{\independent_{\Tr}})$.
\end{lem}
\begin{proof}
(a)\ We first prove that $\mathcal{C}^{\perp_{\Tr}}= \mathcal{C}_2^{\independent_{\Tr}} \oplus \mathcal{C}_1^{\independent_{\Tr}}$. Note that for $(u,v) \in  \mathcal{C}_2^{\independent_{\Tr}} \oplus \mathcal{C}_1^{\independent_{\Tr}}$ and $(u',v') \in \mathcal{C}_1 \oplus \mathcal{C}_2,$ we have $\Tr(v \cdot u'-v' \cdot u)=\Tr(v \cdot u')-\Tr(v' \cdot u)=0.$ Thus, $\mathcal{C}_2^{\independent_{\Tr}} \oplus \mathcal{C}_1^{\independent_{\Tr}} \subseteq \mathcal{C}^{\perp_{\Tr}}$. On the other hand, by Lemma \ref{card}, we also have
\begin{equation*}
\bigl|\mathcal{C}^{\perp_{\Tr}}\bigr|=\frac{p^{2bmn}}{|\mathcal{C}|}=\frac{p^{bmn}}{|\mathcal{C}_1|} \times\frac{p^{bmn}}{|\mathcal{C}_2|}= \bigl|\mathcal{C}_1^{\independent_{\Tr}}\bigr|\bigl|\mathcal{C}_2^{\independent_{\Tr}}\bigr|.\end{equation*} Here, the last equality holds as $\mathcal{C}_i^{\independent_{\Tr}}$ is equal to the orthogonal of $\mathcal{C}_i$ defined via the generating character of $\text{GR}(p^b,m)$, and the equality $\bigl|\mathcal{C}_i\bigr| \bigl|\mathcal{C}_i^{\independent_{\Tr}}\bigr| = {\bigl| \text{GR}(p^b,m) \bigr|}^n$ holds for this orthogonal --- see, e.g., \cite[Chapter 3]{dougherty17}.
Since the expression on the right-hand side equals $\bigl| \mathcal{C}_2^{\independent_{\Tr}} \oplus \mathcal{C}_1^{\independent_{\Tr}} \bigr|$, the inclusion in $\mathcal{C}_2^{\independent_{\Tr}} \oplus \mathcal{C}_1^{\independent_{\Tr}} \subseteq \mathcal{C}^{\perp_{\Tr}}$ must in fact be an equality.

Next, we want to prove that $\mathcal{C}^{\perp_{\Tr,t}} = \mathcal{C}_2^{\independent_{\Tr,t}} \oplus \mathcal{C}_1^{\independent_{\Tr,t}}$ for $t = 1,\ldots,b$. As before, the inclusion $\mathcal{C}_2^{\independent_{\Tr,t}} \oplus \mathcal{C}_1^{\independent_{\Tr,t}} \subseteq \mathcal{C}^{\perp_{\Tr,t}}$ is proved as follows: for $(u,v) \in  \mathcal{C}_2^{\independent_{\Tr},t} \oplus \mathcal{C}_1^{\independent_{\Tr},t}$ and   $(u',v') \in \mathcal{C}_1 \oplus \mathcal{C}_2,$ we have $\Tr(v \cdot u'-v' \cdot u)=\Tr(v \cdot u')-\Tr(v' \cdot u)\equiv 0 \!\!  \pmod{p^{b-t}}$. For the reverse inclusion, consider any $(u,v) \in \mathcal{C}^{\perp_{\Tr,t}}$. We then have $\Tr(v \cdot u'-v' \cdot u)\equiv 0 \!\! \pmod{p^{b-t}}$ for all $(u',v') \in \mathcal{C}_1 \oplus \mathcal{C}_2.$ This implies that $\Tr(p^{t}v \cdot u'-v' \cdot p^t u) = p^t \Tr(v \cdot u'-v' \cdot u) \equiv 0 \!\! \pmod{p^{b}}$ for all  $(u',v') \in \mathcal{C}_1 \oplus \mathcal{C}_2$, from which we infer that $(p^t u,p^t v) \in \mathcal{C}^{\perp_{\Tr}} =  \mathcal{C}_2^{\independent_{\Tr}} \oplus \mathcal{C}_1^{\independent_{\Tr}}$. Thus, for any $(u,v) \in\mathcal{C}^{\perp_{\Tr,t}}$, we have $p^t u \in \mathcal{C}_2^{\independent_{\Tr}}$ and  $p^t v \in\mathcal{C}_1^{\independent_{\Tr}}$, which, by definition, implies that $u \in \mathcal{C}_2^{\independent_{\Tr},t}$ and $ v \in\mathcal{C}_1^{\independent_{\Tr},t}.$ Hence $\mathcal{C}^{\perp_{\Tr,t}} \subseteq \mathcal{C}_2^{\independent_{\Tr},t} \oplus \mathcal{C}_1^{\independent_{\Tr},t}$, as required. 

\medskip

(b)\ From part (a), we deduce that $\mathcal{C} \cap \mathcal{C}^{\perp_{{\Tr,t}}} = (\mathcal{C}_1 \oplus \mathcal{C}_2) \cap (\mathcal{C}_2^{\independent_{\Tr},t} \oplus \mathcal{C}_1^{\independent_{\Tr},t}) =(\mathcal{C}_1 \cap\mathcal{C}_2^{\independent_{\Tr},t}) \oplus (\mathcal{C}_2 \cap \mathcal{C}_1^{\independent_{\Tr},t})$. Consequently,
$$
 \mathcal{C}/(\mathcal{C} \cap \mathcal{C}^{\independent_{\Tr,t}})  =(\mathcal{C}_1 \oplus \mathcal{C}_2)  /\bigl((\mathcal{C}_1 \cap\mathcal{C}_2^{\independent_{\Tr,t}}) \oplus (\mathcal{C}_2 \cap \mathcal{C}_1^{\independent_{\Tr,t}})\bigr) \cong \ \mathcal{C}_1/ (\mathcal{C}_1 \cap\mathcal{C}_2^{\independent_{\Tr,t}})  \oplus \mathcal{C}_2/(\mathcal{C}_2 \cap \mathcal{C}_1^{\independent_{\Tr,t}}).
 $$
\end{proof}

The CSS-like construction described in the proposition below is obtained by applying Theorem~\ref{thm:dt3_restatement} to the code $\mathcal{C} = \mathcal{C}_1 \oplus \mathcal{C}_2$. Its proof, via Lemma~\ref{lem:oplus}, is an exercise in the relevant definitions, and is omitted.

\begin{prop}\label{prop:CSS} (CSS-like Construction) Let $\mathcal{C}_1$ and $\mathcal{C}_2$ be two  additive codes of length $n$ over the Galois ring $\text{GR}(p^b,m).$  From $\mathcal{C}_1$ and $\mathcal{C}_2$,  we can construct  an $((n,K,D;c))$ EAQECC over $\text{GR}(p^b,m),$ where the minimum number, $c$, of pairs of maximally entangled qudits needed for the construction is equal to  $\left\lceil \frac1{2m}\bigl[\rank(\mathcal{C}_1/(\mathcal{C}_1 \cap \mathcal{C}_2^{\independent_{\Tr}})) + \rank(\mathcal{C}_2/(\mathcal{C}_2 \cap \mathcal{C}_1^{\independent_{\Tr}}))\bigr]  \right\rceil$, the minimum distance is

\begin{equation*}
D~=~\left\{\begin{array}{ll}
\min\{ d_{\mathrm{H}}(\mathcal{C}_1^{\independent_{\Tr}}), d_{\mathrm{H}}(\mathcal{C}_2^{\independent_{\Tr}})\} & \text{~~if~~} \mathcal{C}_2^{\independent_{\Tr}} \subseteq \mathcal{C}_1 \\
\min\{d_{\mathrm{H}}(\mathcal{C}_1^{\independent_{\Tr}}\setminus \mathcal{C}_2 ), d_{\mathrm{H}}(\mathcal{C}_2^{\independent_{\Tr}}\setminus \mathcal{C}_1 ) \} & \text{~~otherwise}\,,
\end{array}\right.
\end{equation*}
and the dimension $K$ is bounded as $p^{bm(n+c)}/ (|\mathcal{C}_1||\mathcal{C}_2| \, {p^{\sum_{t=1}^{b-1} (b-t)\rho_t}}) \leq  K \leq p^{bm(n+c)}/(|\mathcal{C}_1||\mathcal{C}_2|)$, with $\rho_t := \rank(\mathcal{C}_1/(\mathcal{C}_1 \cap \mathcal{C}_2^{\independent_{\Tr},t-1}))+\rank(\mathcal{C}_2/(\mathcal{C}_2 \cap \mathcal{C}_1^{\independent_{\Tr},t-1}))-\rank(\mathcal{C}_1/(\mathcal{C}_1 \cap \mathcal{C}_2^{\independent_{\Tr},t}))-\rank(\mathcal{C}_2/(\mathcal{C}_2 \cap \mathcal{C}_1^{\independent_{\Tr},t}))$.   

If either \emph{(a)}~$\mathcal{C}_1$ and  $\mathcal{C}_2$ are  free modules over $\mathbb{Z}_{p^b}$ or \emph{(b)}~$\mathcal{C}_1/(\mathcal{C}_1 \cap \mathcal{C}_2^{\independent_{\Tr}})$ and  $\mathcal{C}_2/(\mathcal{C}_2 \cap \mathcal{C}_1^{\independent_{\Tr}})$ are free modules over $\mathbb{Z}_{p^b}$, then $K = {p^{bm(n+c)}}/{(|\mathcal{C}_1||\mathcal{C}_2|)}.$  In the case of \emph{(b)}, we additionally have $$c=\left\lceil \frac1{2m}\bigl[\rank(\mathcal{C}_1)+\rank(\mathcal{C}_2) - \rank(\mathcal{C}_1 \cap \mathcal{C}_2^{\independent_{\Tr}}) - \rank(\mathcal{C}_2 \cap \mathcal{C}_1^{\independent_{\Tr}})\bigr]  \right\rceil.$$
\end{prop}

\section{Proof of Proposition~\ref{prop:extsame}} \label{app:extsame_proof}

 Let  $$\mathcal{G}=\{(v_1,w_1),(v_2,w_2),\ldots, (v_{e+d},w_{e+d}),(x_1,y_1),(x_2,y_2),\ldots, (x_{e},y_{e})\}$$ be a minimal generating set of $\mathcal{C}$ as a $\mathbb{Z}_{p^b}$-module such that, for $i = 1,2\ldots,e$, the generators $(v_i,w_i)$ and $(x_i,y_i)$ form hyperbolic pairs, and the generators $(v_i,w_i)$, $i = e+1,\ldots,e+d$, are isotropic. Further, let $\mathcal{H}=\{(e_1,f_1),(e_2,f_2),\ldots, $ $(e_{t+s},f_{t+s}),$ $(g_1,h_1),(g_2,h_2),\ldots, (g_{t},h_{t})\}$ be any minimal generating set in standard form  of $\mathcal{C}^{\perp_{\Tr}}$  as a $\mathbb{Z}_{p^b}$-module,  such that, for $i = 1,2,\ldots,t$, the generators $(e_i,f_i)$ and $(g_i,h_i)$ form hyperbolic pairs, and the generators $(e_i,f_i)$, $i = t+1,\ldots,t+s$, are isotropic.

Let $\tau : \mathcal{C}^{\perp_{\Tr}}  \to \mathcal{C}^{\perp_{\Tr}}/(\mathcal{C}^{\perp_{\Tr}} \cap \mathcal{C})$ be the canonical projection map that takes $(v,w) \in \mathcal{C}^{\perp_{\Tr}}$ to the coset $(v,w) + (\mathcal{C}^{\perp_{\Tr}} \cap \mathcal{C})$. By Proposition~\ref{prop:count}(a), $2t = \rank(\mathcal{C}^{\perp_{\Tr}}/(\mathcal{C}^{\perp_{\Tr}} \cap \mathcal{C})).$ As $\rank(\mathcal{C}^{\perp_{\Tr}}/(\mathcal{C}^{\perp_{\Tr}} \cap \mathcal{C})) \geq 2m$, we get $t\geq m.$ Again using the fact that $\mathcal{H}$ is a minimal generating set of $\mathcal{C}^{\perp_{\Tr}}$ and  $ \rank(\mathcal{C}^{\perp_{\Tr}}/(\mathcal{C}^{\perp_{\Tr}} \cap \mathcal{C}))=2t,$ we get $\{\tau((e_1,f_1)),\tau((e_2,f_2)),\ldots, \tau((e_{t},f_{t})),$ $\tau((g_1,h_1)),$ $\tau((g_2,h_2)),\ldots,\tau((g_{t},h_{t}))\}$ is a minimal generating set of $\mathcal{C}^{\perp_{\Tr}}/(\mathcal{C}^{\perp_{\Tr}} \cap \mathcal{C}) $ as a free module over $\mathbb{Z}_{p^b}$. 
 Thus $\{\tau((e_1,f_1)),$ $\tau((e_2,f_2)),\ldots, $ $\tau((e_{t},f_{t})),\tau((g_1,h_1)),\tau((g_2,h_2)),\ldots, \tau((g_{t},h_{t}))\}$ is a linearly independent set over $\mathbb{Z}_{p^b}$ --- this holds by using Proposition \ref{prop:free}.  This  implies that $\{(e_1,f_1),(e_2,f_2),\ldots, (e_{t},f_{t}),(g_1,h_1),(g_2,h_2),\ldots, (g_{t},h_{t})\}$ is a linearly independent set over $\mathbb{Z}_{p^b}.$ Using the fact that $\mathcal{C}^{\perp_{\Tr}}/(\mathcal{C}^{\perp_{\Tr}} \cap \mathcal{C}) $ is a free module over $\mathbb{Z}_{p^b}$ and by Theorem \ref{prop:count}(b), we get that $\{(e_i,f_i) : t+1 \leq i \leq t+s\}$ is a minimal generating set of $\mathcal{C}^{\perp_{\Tr}} \cap \mathcal{C}$ as a $\mathbb{Z}_{p^b}$-module.

For $1 \leq i \leq m,$ $\chi(f_ig_i-e_ih_i)$ is a $p^b$-th root of unity, and $(e_i,f_i)$ and $(g_i,h_i)$ form hyperbolic pairs, so let $\chi(f_ig_i-e_ih_i)=\zeta^{\theta_{i}},$ where $\zeta = \exp(\frac{2\pi i}{p^b})$ and $0< \theta_{i} < p^b.$  Let $\gamma_1,\gamma_2,\ldots,\gamma_m$ be a  basis of $\text{GR}(p^b,m)$ as a free module over $\mathbb{Z}_{p^b}$ and let $\beta_1,\beta_2,\ldots,\beta_m$ be the dual basis of $\gamma_1,\gamma_2,\ldots,\gamma_m.$
Let $\mathcal{M} \subseteq \text{GR}(p^b,m)^{2(n+1)}$ be an additive code generated by  
\begin{equation}\label{Mdual} \mathcal{H}'=\{(e_i,0,f_i,-\beta_{i}), (g_i,\theta_i\gamma_{i},h_i,0) : 1 \leq i \leq m \} \cup \{ (e_{k},0,f_{k},0),(g_{\ell},0,h_{\ell},0) : m+1 \leq k \leq t+s \text{~and~} m+1 \leq \ell \leq t \} .\end{equation}
  Note that the generators in $\mathcal{H}'$ are obtained from the generators in $\mathcal{H}$ by adding extra coordinates. Therefore, the fact that $\mathcal{H}$ is a minimal generating set of $\mathcal{C}^{\perp_{\Tr}}$ implies that $\mathcal{H}'$ is a minimal generating set of $\mathcal{M}$. 
  It is straightforward to verify that the generators $(e_k,0,f_k,0)$ and $(g_k,0,h_k,0)$ for $m+1 \leq k \leq t$ form hyperbolic pairs, and the generators $(e_j,0,f_j,0)$, $j = t+1,\ldots,t+s$, and $(e_i,0,f_i,-\beta_{i})$ and $(g_i,\theta_i\gamma_{i},h_i,0)$,  $1 \leq i \leq m$, are isotropic. 
By using Lemma \ref{card}, we get  \begin{equation}\label{ecard1}|\mathcal{M}^{\perp_{\Tr}}|=\frac{p^{bm(2n+2)}}{|\mathcal{M}|}=\frac{p^{bm(2n+2)}}{|\mathcal{C}^{\perp_{\Tr}}|}= p^{2bm} |\mathcal{C}|.\end{equation}

Now consider an additive code $\mathcal{N} \subseteq \text{GR}(p^b,m)^{2(n+1)}$  generated by 
$$\mathcal{A}=\{(e_i,0,f_i,-\beta_{i}), (g_i,\theta_i\gamma_{i},h_i,0) : 1 \leq i \leq m \} \cup \{  (v_{j},0,w_{j},0),(x_k,0,y_k,0) : 1 \leq j \leq e+d ,  1\leq  k \leq e\}.$$  It is straightforward to verify that $\mathcal{N} \subseteq \mathcal{M}^{\perp_{\Tr}}.$ Here we assert that $|\mathcal{N}|=|\mathcal{C}|(p^{b})^{2m},$ which will, using \eqref{ecard1}, imply that  $\mathcal{N} =\mathcal{M}^{\perp_{\Tr}}.$
We will prove this assertion by first showing that $\mathcal{A}$ is a minimal generating set of $\mathcal{N}$ as a $\mathbb{Z}_{p^b}$-module. 

Assume, to the contrary, that $\mathcal{A}$ is not a minimal generating set of $\mathcal{N}$ as a $\mathbb{Z}_{p^b}$-module. This means that
$$\sum \limits_{i=1}^{m} (a_i (e_i,0,f_i,-\beta_{i})+b_i (g_i,\theta_i\gamma_{i},h_i,0))+\sum \limits_{j=1}^{e+d}  c_j (v_{j},0,w_{j},0) +\sum \limits_{k=1}^{e} d_k (x_k,0,y_k,0) =0 $$
for some $a_i,b_i,c_j,d_k \in \mathbb{Z}_{p^b}$ with at least one among $a_i,b_i,c_j,d_k$ being a unit in $\mathbb{Z}_{p^b}.$
This gives
\begin{equation}\label{ES1}
\sum \limits_{i=1}^{m} (a_i (e_i,f_i)+b_i (g_i,h_i)) = -\sum \limits_{j=1}^{e+d}  c_j (v_{j},w_{j}) -\sum \limits_{k=1}^{e} d_k (x_k,y_k) \in \mathcal{C}.
\end{equation}
From this and using the fact that $(e_i,f_i), (g_i,h_i) \in \mathcal{C}^{\perp_{\Tr}},$ we get $\sum \limits_{i=1}^{m} (a_i (e_i,f_i)+b_i (g_i,h_i)) \in \mathcal{C}^{\perp_{\Tr}} \cap \mathcal{C},$ which implies that $\tau(\sum \limits_{i=1}^{m} (a_i (e_i,f_i)+b_i (g_i,h_i)))=\mathcal{C}^{\perp_{\Tr}} \cap \mathcal{C}.$ As $\{\tau((e_1,f_1)),\tau((e_2,f_2)),\ldots, $ $\tau((e_{m},f_{m})),\tau((g_1,h_1)),$ $\tau((g_2,h_2)),$ $\ldots, \tau((g_{m},h_{m}))\}$ is a linearly independent set over $\mathbb{Z}_{p^b},$ we get $a_i=b_i=0$ for $1 \leq i \leq m.$ This, by using \eqref{ES1}, implies that 
\begin{equation}\label{ES2}
\sum \limits_{j=1}^{e+d}  c_j (v_{j},w_{j}) +\sum \limits_{k=1}^{e} d_k (x_k,y_k)=0.
\end{equation}
From this and using the facts that at least one among $a_i,b_i,c_j$ and $d_k$ is a unit in $\mathbb{Z}_{p^b}$, and $a_i=b_i=0$ for $1 \leq i \leq m,$ we get that at least one among $c_j$ and $d_k$ is a unit in $\mathbb{Z}_{p^b}$. This, by using \eqref{ES2}, implies that $\{  (v_{j},w_{j}),(x_k,y_k) : 1 \leq j \leq e+d ,  1\leq  k \leq e\}$ is not a minimal generating set of $\mathcal{C}$ as a $\mathbb{Z}_{p^b}$-module, which is a contradiction. 

Thus, $\mathcal{A}$ is a minimal generating set of $\mathcal{N}$ as a $\mathbb{Z}_{p^b}$-module. From this and the fact that $\{(e_1,f_1),(e_2,f_2),\ldots, (e_{m},f_{m}),(g_1,h_1),(g_2,h_2),\ldots, (g_{m},h_{m})\}$ is a linearly independent set over $\mathbb{Z}_{p^b}$, while $\mathcal{C}$ does not contain any $\mathbb{Z}_{p^b}$ linear combination from this set, it follows that $|\mathcal{N}|=|\mathcal{C}|(p^{b})^{2m}$. Thus, via \eqref{ecard1}, $\mathcal{N} =\mathcal{M}^{\perp_{\Tr}}$. Clearly, the generators $(v_{k},0,w_{k},0)$ and $(x_k,0,y_k,0),$  $1\leq  k \leq e$, form hyperbolic pairs, and the generators $(v_{j},0,w_{j},0),$ $e+1 \leq j \leq e+d$ and, $(e_i,0,f_i,-\beta_{i})$ and $(g_i,\theta_i\gamma_{i},h_i,0),$  $1 \leq i \leq m$, are isotropic. This verifies item (i) in the statement of the proposition.

For the second item, we start by noting that 
\begin{equation*}
\left\langle\{(e_j,0,f_j,0): t+1 \leq j \leq t+s\} \cup \{(e_i,0,f_i,-\beta_{i}),(g_i,\theta_i\gamma_{i},h_i,0) : 1 \leq i \leq m\}\right\rangle \subseteq \mathcal{M} \cap \mathcal{M}^{\perp_{\Tr}} .
\end{equation*}
Next we observe that if 
$\sum \limits_{i=m+1}^{t} (a_i (e_i,0,f_i,0)+b_i (g_i,0,h_i,0)) \in \mathcal{M} \cap \mathcal{M}^{\perp_{\Tr}} $ for some $a_i,b_i \in \mathbb{Z}_{p^b},$ then 
$$\sum \limits_{i=m+1}^{t} (-a_i (e_i,0,f_i,0)-b_i (g_i,0,h_i,0))=\sum \limits_{i=1}^{m} (a_i (e_i,0,f_i,-\beta_{i})+b_i (g_i,\theta_i\gamma_{i},h_i,0))+\sum \limits_{j=1}^{e+d}  c_j (v_{j},0,w_{j},0) +\sum \limits_{k=1}^{e} d_k (x_k,0,y_k,0)$$
for some $a_i,b_i,c_j,d_k \in \mathbb{Z}_{p^b}.$

This gives
\begin{equation*}
\sum \limits_{i=1}^{t} (a_i (e_i,f_i)+b_i (g_i,h_i)) = -\sum \limits_{j=1}^{e+d}  c_j (v_{j},w_{j}) -\sum \limits_{k=1}^{e} d_k (x_k,y_k) \in \mathcal{C}.
\end{equation*}

From this and using the fact that $(e_i,f_i), (g_i,h_i) \in \mathcal{C}^{\perp_{\Tr}},$ we get $\sum \limits_{i=1}^{t} (a_i (e_i,f_i)+b_i (g_i,h_i)) \in \mathcal{C}^{\perp_{\Tr}} \cap \mathcal{C}.$ This implies that   $\tau(\sum \limits_{i=1}^{t} (a_i (e_i,f_i)+b_i (g_i,h_i)))=0.$ Further, as $\{\tau((e_{1},f_{1})),\tau((g_{1},h_{1})),\ldots, \tau((e_{{t}},f_{{t}})),\tau((g_{{t}},h_{{t}}))\}$ is a linearly independent set over $\mathbb{Z}_{p^b},$ we get $a_i=b_i=0$ for $1 \leq i \leq t.$  Thus
\begin{equation*}
\mathcal{M} \cap \mathcal{M}^{\perp_{\Tr}} =\langle\{(e_j,0,f_j,0): t+1 \leq j \leq t+d\} \cup \{(e_i,0,f_i,-\beta_{i}),(g_i,\theta_i\gamma_{i},h_i,0) : 1 \leq i \leq m\}\rangle.
\end{equation*}


From this and by using \eqref{Mdual} and the fact that $\{(e_i,f_i) : t+1 \leq i \leq t+d\}$ is a minimal generating set of $\mathcal{C}^{\perp_{\Tr}} \cap \mathcal{C}$ as a $\mathbb{Z}_{p^b}$-module, we get  $d_s\bigl({\mathcal{M}} \setminus \mathcal{M}^{\perp_{\Tr}}\bigr) \ge d_s\bigl({\mathcal{C}}^{\perp_{\Tr}} \setminus \mathcal{C}\bigr) $ if $\text{rank} (\mathcal{C}^{\perp_{\Tr}}/(\mathcal{C}  \cap  \mathcal{C}^{\perp_{\Tr}})) >2m$ and $ d_s\bigl({\mathcal{M}} \bigr) \ge d_s\bigl({\mathcal{C}}^{\perp_{\Tr}}\bigr) $ if $\text{rank} (\mathcal{C}^{\perp_{\Tr}}/(\mathcal{C}  \cap  \mathcal{C}^{\perp_{\Tr}})) =2m.$
This verifies item (ii) in the statement of the proposition.

Finally, we will verify item (iii). We use the notation $0^{\mu}$, for a non-negative integer $\mu$, to denote a string of $\mu$ $0$s. For $1 \leq i \leq e,$ $\chi(w_i \cdot x_i-v_i \cdot y_i)$ is a $p^b$-th root of unity and $(v_i,w_i)$ and $(x_i,y_i)$ form hyperbolic pairs, so let $\chi(w_i \cdot x_i-v_i \cdot y_i)=\zeta^{z_{i}},$ where  $0< z_{i} < p^b.$  Let $\mathcal{C}'\subseteq \text{GR}(p^b,m)^{2(n+c)}$ be an additive code generated by $\mathcal{G}' = \{u'_{11},u'_{12},\ldots,u'_{e1},u'_{e2},b'_1,\ldots,b'_d\}$, with 
$$
 u_{i1}' := (v_i,0^c,w_i,0^k,-\beta_{\ell},0^{c-1-k}) \text{ and }  u'_{i2} = (x_i,0^{k},z_i\gamma_{\ell},0^{c-1-k},y_i,0^c)
 $$ for $i = km+ \ell$, with $0 \le k \le c-2$ and $1 \le \ell \le m$  and, 

  $$
 u_{i1}' := (v_i,0^c,w_i,0^{c-1},-\beta_{\ell}) \text{ and }  u'_{i2} = (x_i,0^{c-1},z_i\gamma_{\ell},y_i,0^c)
 $$ for $i = (c-1)m + \ell$, with $1 \le \ell \le e-(c-1)m$, and  $b'_j = (v_{e+j},0^c,w_{e+j},0^c) \text{ for } j=1,2\ldots,d$. Clearly, $\mathcal{C}'\subseteq \text{GR}(p^b,m)^{2(n+c)}$ is a $\chi$-self-orthogonal extension $\mathcal{C}$ with entanglement degree equal to $c,$   obtained using the Two-Step Construction.  
By Theorem \ref{EAQZ}, the minimum entanglement degree of any $\chi$-self-orthogonal extension of $\mathcal{C}$ is equal to  $\left\lceil { \frac{1}{2m} \rank(\mathcal{C}/(\mathcal{C} \cap \mathcal{C}^{\perp_{\Tr}}))  }\right\rceil=\left\lceil  \frac{e}{m}\right\rceil=c.$ This implies that $\mathcal{C}'\subseteq \text{GR}(p^b,m)^{2(n+c)}$ is a minimal $\chi$-self-orthogonal extension of an additive code $\mathcal{C} \subseteq \text{GR}(p^b,m)^{2n}.$ 

Let $\mathcal{M}'\subseteq \text{GR}(p^b,m)^{2(n+c+1)}$ be an additive code generated by $\mathcal{A}' = \{v'_{11},v'_{12},\ldots,v'_{e1},v'_{e2},w'_1,\ldots,w'_d,x'_1,\ldots,$ $x'_m,$ $y'_1,\ldots,y'_m\}$, with 
$$
 v_{i1}' := (v_i,0^{c+1},w_i,0^{k+1},-\beta_{\ell},0^{c-1-k}) \text{ and }  v'_{i2} = (x_i,0^{k+1},z_{i}\gamma_{\ell},0^{c-1-k},y_i,0^{c+1})
 $$ for $i = km+ \ell$, with $0 \le k \le c-2$ and $1 \le \ell \le m$, and

  $$
 v_{i1}' := (v_i,0^{c+1},w_i,0^{c},-\beta_{\ell}) \text{ and }  v'_{i2} = (x_i,0^{c},z_{i}\gamma_{\ell},y_i,0^{c+1})
 $$ for $i = (c-1)m + \ell$, with $1 \le \ell \le e-(c-1)m$, and  $w'_j = (v_{e+j},0,w_{e+j},0) \text{ for } j=1,2\ldots,d$, $x'_k =(e_k,0^{c+1},f_k,-\beta_{k},0^{c}),$ and $y'_k=(g_k,\theta_k\gamma_{k},0^{c},h_k,0^{c+1})$ for $k=1,2,\ldots,m.$ Clearly, $\mathcal{M}'\subseteq \text{GR}(p^b,m)^{2(n+c+1)}$ is a $\chi$-self-orthogonal extension $\mathcal{M}^{\perp_{\Tr}} $ with entanglement degree equal to $c,$   obtained using the Two-Step Construction.  
 By Theorem \ref{EAQZ}, the minimum entanglement degree of any $\chi$-self-orthogonal extension of $\mathcal{M}^{\perp_{\Tr}}$ is equal to  $\left\lceil { \frac{1}{2m} \rank(\mathcal{M}^{\perp_{\Tr}}/(\mathcal{M}^{\perp_{\Tr}} \cap \mathcal{M}))  }\right\rceil=\left\lceil  \frac{e}{m}\right\rceil=c.$ This implies that $\mathcal{M}'\subseteq \text{GR}(p^b,m)^{2(n+c+1)}$ is a minimal $\chi$-self-orthogonal extension of an additive code $\mathcal{M}^{\perp_{\Tr}} \subseteq \text{GR}(p^b,m)^{2n}.$ 
Note that $|\mathcal{M}'|=|\mathcal{C}'|p^{2mb}$ --- this follows from the facts that $\{(e_1,f_1),(e_2,f_2),\ldots, (e_{m},f_{m}),(g_1,h_1),$ $(g_2,h_2),$ $\ldots, (g_{m},h_{m})\}$ is a linearly independent set over $\mathbb{Z}_{p^b}$ and $\mathcal{C}$ does not contain any $\mathbb{Z}_{p^b}$-linear combination from this set.

\section{EAQECCs over $\text{GR}(p^b,m)$ from Additive Codes over  $\text{GR}(p^b,2m)$}\label{app:trace-alt-form} 

There is a well-known equivalence between (a)~quantum stabilizer codes of length $n$ over the field $\mathbb{F}_q$ and (b)~additive codes of length $n$ over $\mathbb{F}_{q^2}$ that are self-orthogonal with respect to a certain trace-alternating form \cite{calder}, \cite[Section~IV-B]{ketkar}.
In this appendix, we provide the machinery needed (see Proposition~\ref{prop:equiv} below) to extend this equivalence to codes over Galois rings. This paves the way for the use of additive codes over $\text{GR}(p^b,2m)$ for the purpose of constructing EAQECCs over $\text{GR}(p^b,m).$ 

We begin by observing that $\text{GR}(p^b,2m)$ is a degree-$2$ extension of $\text{GR}(p^b,m).$ Moreover, there exists a monic polynomial $g(x)\in \text{GR}(p^b,m)[x]$ of degree $2$ such that its reduction $(g(x) \text{ mod } p) \in \mathbb{F}_{p^m}[x]$ is irreducible and primitive over $\mathbb{F}_{p^m}$ and $g(x)$ divides $x^{p^{2m}-1}-1$ in $\text{GR}(p^b,m)[x].$ Thus we have  $\text{GR}(p^b,2m) \cong \text{GR}(p^b,m) [x]/ \langle g(x) \rangle.$ 
If $\delta=x+\langle g(x) \rangle,$ then clearly $g(\delta)=0,$ and every element of $\text{GR}(p^b,2m) (\cong \text{GR}(p^b,2m)[x]/ \langle g(x) \rangle)$ has a unique representation as $$r=r_0+r_1\delta\text{~with~} r_0,r_1 \in \text{GR}(p^b,m).$$
This implies that $\text{GR}(p^b,2m)$ is a free module of rank $2$ over  $\text{GR}(p^b,m)$ with $\{1,\delta\}$ as a basis (see \cite[Ch.\ 14]{wan}).  

Define $\sigma: \text{GR}(p^b,m)^{2n} \to \text{GR}(p^b,2m)^n $ as follows:
  $$\sigma(a_1,a_2,\ldots,a_n, b_1,b_2,\ldots,b_n) =(a_1+\delta b_1,a_2+\delta b_2,\ldots,a_n+\delta b_n).$$ Clearly, $\sigma$ is a bijective and  additive map, thus a $\mathbb{Z}_{p^b}$-module isomorphism. 
  Throughout this section, whenever for a tuple $(a,b)~ (=(a_1,\ldots,a_n, b_1,\ldots,b_n)) \in \text{GR}(p^b,m)^{2n},$   we consider its image $\sigma(a,b)=a+\delta b\in \text{GR}(p^b,2m)^n,$ it is to be implicitly understood that $a+\delta b=(a_1+\delta b_1,a_2+\delta b_2,\ldots,a_n+\delta b_n) \in \text{GR}(p^b,2m)^n.$   
Clearly, $wt_s((a,b) )=w_H(\sigma((a,b)),$ thus $\sigma$ is an isometry in this sense. 

We define a trace-alternating form on $\text{GR}(p^b,2m)^{n}$ as follows: for $ a+\delta b$ and $ a'+\delta b' \in \text{GR}(p^b,2m)^{n},$
$${\langle a+\delta b |  a'+\delta b' \rangle}_{\mathrm{a}} = \Tr(ba'-b'a).$$ Here, ${\Tr}: \text{GR}(p^b,m) \to \mathbb{Z}_{p^b}$ is the generalized trace map of  $\text{GR}(p^b,m)$ relative to $\mathbb{Z}_{p^b}.$

\begin{defin} For a subset $\mathcal{C} \subseteq \text{GR}(p^b,2m)^{n}$, we define
\begin{itemize}
\item the trace-alternating dual, $\mathcal{C}^{\perp_{\mathrm{a}}}$, of $\mathcal{C}$ as
$$\mathcal{C}^{\perp_{\mathrm{a}}}=\{a'+\delta b' \in \text{GR}(p^b,2m)^{n}  : {\langle a+\delta b |  a'+\delta b' \rangle}_{\mathrm{a}} = 0 \text{~for all~}  a+\delta b \in \mathcal{C}\},$$
\item and for $t = 0,1,\ldots,b$, the trace-alternating $t$-dual of $\mathcal{C}$ as 
$$\mathcal{C}^{\perp_{\mathrm{a},t}} = \{a'+\delta b'  \in \text{GR}(p^b,2m)^{n} ~:~ {\langle a+\delta b |  a'+\delta b' \rangle}_{\mathrm{a}}  \equiv 0 \!\!\! \pmod{p^{b-t}} \text{~for all~}  a+\delta b  \in \mathcal{C}\}.$$
\end{itemize}
\end{defin}

As  $\sigma$ is a $\mathbb{Z}_{p^b}$-module isomorphism, a subset $\mathcal{C}$ of $\text{GR}(p^b,m)^{2n}$ is an additive code over $\text{GR}(p^b,m)$  of length $2n$ if and only if 
$\sigma{(\mathcal{C})}$ is an additive code over $\text{GR}(p^b,2m)$  of length $n.$
The proposition below follows immediately from the relevant definitions.

 \begin{prop}\label{prop:equiv} Let $\mathcal{C}\subseteq \text{GR}(p^b,m)^{2n}$ be an additive code, and let $\overline{\mathcal{C}} = \sigma(\mathcal{C})$ be its image in $\text{GR}(p^b,2m)^n$ under the mapping $\sigma$. 
Then, $\overline{\mathcal{C}}^{\perp_{\mathrm{a}}} = \sigma(\mathcal{C}^{\perp_{\Tr}}),$ and for $t = 0,1,\ldots,b$, we have $\overline{\mathcal{C}}^{\perp_{\mathrm{a},t}} = \sigma(\mathcal{C}^{\perp_{\Tr,t}}).$ As a consequence, for $t = 0,1,\ldots,b$,  $\rank(\mathcal{C}/(\mathcal{C} \cap \mathcal{C}^{\perp_{\Tr,t}}))=\rank(\overline{\mathcal{C}}/(\overline{\mathcal{C}} \cap \overline{\mathcal{C}}^{\perp_{\mathrm{a},t}})).$ Moreover, $\mathcal{C}/(\mathcal{C} \cap \mathcal{C}^{\perp_{\Tr}}) $ is a free module over $\mathbb{Z}_{p^b}$ if and only if $\overline{\mathcal{C}}/(\overline{\mathcal{C}} \cap \overline{\mathcal{C}}^{\perp_{\mathrm{a}}})$  is a free module over $\mathbb{Z}_{p^b}.$ 
\end{prop}

\begin{rem}\label{rem:field} If $\mathcal{C}$ is an additive code over a finite field $\mathbb{F}_q$, then $\overline{\mathcal{C}}^{\perp_{\mathrm{a}}}$ is the same as the trace-alternating dual defined in \cite[Section IV-B]{ketkar}. Moreover, if  $\overline{\mathcal{C}}$ is a linear code over $\mathbb{F}_{q^2}$, then by \cite[Lemma 18]{ketkar},   $\overline{\mathcal{C}}^{\perp_{\mathrm{a}}}$ is equal to the Hermitian dual of $\overline{\mathcal{C}}$. Thus, for example, Proposition~1 of Lison\v{e}k and Singh \cite{liso} may be viewed as a particular case of our Theorem 3.1(c).
\end{rem}

The following result is essentially a restatement of Theorem~\ref{thm:dt3_restatement}, expressed in the language of this appendix. This generalizes analogous results in Ketkar et al.\ \cite[Theorem 15]{ketkar} and Galindo et al.\ \cite[Theorem 3]{galindo}.

\begin{thm}\label{thm:dt3_equiv}[Equivalent to Theorem~\ref{thm:dt3_restatement}] Let $\mathcal{C} \subseteq \text{GR}(p^b,2m)^{n}$ be an additive code over the Galois ring $\text{GR}(p^b,2m)$.   From $\mathcal{C}$, we can construct an  $((n,K,D;c))$ EAQECC over $\text{GR}(p^b,m),$ where the minimum number, $c$, of pairs of maximally entangled qudits needed for the construction is equal to  $\left\lceil \frac1{2m}\rank(\mathcal{C}/(\mathcal{C} \cap \mathcal{C}^{\perp_{\mathrm{a}}}))  \right\rceil$, the minimum distance is 
\begin{equation*}
D~=~\left\{\begin{array}{ll}
d_{\mathrm{H}}(\mathcal{C}^{\perp_{\mathrm{a}}}) & \text{~~if~~} \mathcal{C}^{\perp_{\mathrm{a}}} \subseteq \mathcal{C} \\
d_{\mathrm{H}}(\mathcal{C}^{\perp_{\mathrm{a}}}\setminus \mathcal{C} )  & \text{~~otherwise}\,
\end{array}\right.
\end{equation*}
and the dimension $K$ is bounded as $p^{bm(n+c)}/ (|\mathcal{C}| \, {p^{\sum_{t=1}^{b-1} (b-t)\rho_t}}) \leq  K \leq p^{bm(n+c)}/|\mathcal{C}|$, with $\rho_t := \rank(\mathcal{C}/(\mathcal{C} \cap \mathcal{C}^{\perp_{\mathrm{a},t-1}}))-\rank(\mathcal{C}/(\mathcal{C} \cap \mathcal{C}^{\perp_{\mathrm{a},t}}))$.  
If either \emph{(a)}~$\mathcal{C}$ is free or \emph{(b)}~$\mathcal{C}/(\mathcal{C} \cap \mathcal{C}^{\perp_{\mathrm{a}}})$ is a free module over $\mathbb{Z}_{p^b}$, then $K = p^{bm(n+c)}/|\mathcal{C}|$. In the case of \emph{(b)}, we additionally have $c=\left\lceil \frac1{2m}[\rank(\mathcal{C}) - \rank(\mathcal{C} \cap \mathcal{C}^{\perp_{\mathrm{a}}})]\right\rceil$.\end{thm}

The following corollary immediately follows from Proposition~\ref{prop:equiv} and Theorem  \ref{thm:same}, and is a generalization of \cite[Theorem  16]{luo1}.

\begin{cor}\label{cor:same}  Let $\mathcal{C} \subseteq \text{GR}(p^b,2m)^{n}$ be an additive code over the Galois ring $\text{GR}(p^b,2m)$ such that the additive code $\mathcal{C}^{\perp_{\mathrm{a}}}/(\mathcal{C}  \cap  \mathcal{C}^{\perp_{\mathrm{a}}})$ is free as a $\mathbb{Z}_{p^b}$-module of rank atleast  $2m.$  Set $c:=\left\lceil \frac1{2m}\rank(\mathcal{C}/(\mathcal{C} \cap \mathcal{C}^{\perp_{\mathrm{a}}}))  \right\rceil$ and $D:=d_{\mathrm{H}}(\mathcal{C}^{\perp_{\mathrm{a}}}\setminus \mathcal{C} ).$  Then, the existence of an $((n,K,D;c))$ EAQECC over  $\text{GR}(p^b,m)$ constructed from $\mathcal{C}$ using Theorem~\ref{thm:dt3_equiv}   implies the existence of an $((n+1,\frac{1}{p^{bm}}K,D';c))$ EAQECC over $\text{GR}(p^b,m)$ with 
 \begin{equation*}
D'~\geq ~\left\{\begin{array}{ll}
d_{\mathrm{H}}(\mathcal{C}^{\perp_{\mathrm{a}}})  & \text{~~if~~}  \text{rank} (\mathcal{C}^{\perp_{\mathrm{a}}}/(\mathcal{C}  \cap  \mathcal{C}^{\perp_{\mathrm{a}}})) =2m; \\
D   & \text{~~if~~} \text{rank} (\mathcal{C}^{\perp_{\mathrm{a}}}/(\mathcal{C}  \cap  \mathcal{C}^{\perp_{\mathrm{a}}})) >2m.\,
\end{array}\right.
\end{equation*}
\end{cor}

\end{document}